\pdfoutput=1

\RequirePackage{snapshot}
\documentclass[12pt,fabfinalm,fabtocsmall]{fabarticle}

\usepackage[OT1]{fontenc}

\usepackage[utf8]{inputenc}

\usepackage
[a4paper,tmargin=3truecm,bmargin=3truecm,rmargin=2.5truecm,lmargin=2.5truecm,twoside,verbose=true]{geometry}
\usepackage{amsmath,amssymb,latexsym}

\usepackage{stmaryrd,wasysym,upgreek,mathrsfs,dsfont,mathtools}
\usepackage[english]{babel}
\usepackage{graphicx,color}
\usepackage{slashed,subfig}

\allowdisplaybreaks[1]

\usepackage[square,numbers,sort&compress]{natbib}

\usepackage{bibentry}

\usepackage[pdftex]{hyperref}
\hypersetup{%
  colorlinks=true,%
  linkcolor=black,%
  anchorcolor=black,%
  citecolor=black,%
  urlcolor=black,%
  pdftitle={},%
  pdfauthor={Fabien Vignes-Tourneret},%
  pdfsubject={},%
  pdfkeywords={},%
  bookmarksopen=false,%
}

\usepackage{tikz}
\usetikzlibrary{arrows,backgrounds}

\usepackage[hyperref,thmmarks,amsmath]{fabntheorem}
\usepackage{fabmesthm-en}
\usepackage{fabmacromath}

\usepackage{pdfsync}

\numberwithin{equation}{section}

\DeclareMathOperator*{\degre}{deg}
\newcommand{\degp}[1]{\degre(#1)}

\title{Topological graph polynomials\\and quantum field theory}
\subtitle{Part II: Mehler kernel theories}

\author{T. Krajewski$^{1,2}$, V. Rivasseau$^{1}$ and F. Vignes-Tourneret$^{3}$}

\begin{document}
\nobibliography*
\maketitle

\begin{center}
{\small\it 1) Laboratoire de Physique Th\'eorique, CNRS UMR 8627,\\
Universit\'e Paris XI,  F-91405 Orsay Cedex, France\\
\medskip
2) on leave, Centre de Physique Th\'eorique, CNRS UMR 6207\\
CNRS Luminy, Case 907, F-13288 Marseille Cedex 9\\
and Universit\'e de Provence\\
3 place Victor Hugo, F-13331 Marseille Cedex 3 France\\  
\medskip
3)Institut Camille Jordan, CNRS UMR 5208\\
Universit\'e Claude Bernard Lyon 1\\
43 boulevard du 11 novembre 1918\\
F-69622 Villeurbanne cedex France\\  

\medskip

E-mail: {\tt krajew@cpt.univ-mrs.fr, rivass@th.u-psud.fr,\\ vignes@math.univ-lyon1.fr}}
\end{center}

\bigskip

\begin{abstract}
We define a new topological polynomial extending the Bollob\'as-Riordan one, which obeys a four-term reduction relation
of the deletion/contraction type and has a natural behavior under partial duality.
This allows to write down a completely explicit combinatorial evaluation of the polynomials,
occurring in the parametric representation 
of the non-commutative Grosse-Wulkenhaar quantum field theory. An explicit solution 
of the parametric representation for commutative field theories based on the Mehler kernel is also provided.
\end{abstract}

\vfill
\noindent
CPT-P001-2010\\
LPT-ORSAY 10-02

\newpage

{ \footnotesize
  \tableofcontents}

\section*{Introduction}
\addcontentsline{toc}{section}{Introduction}
In \citep{Krajewski2008aa} the relation between the parametric
representation of Feynman graph amplitude
\citep{Nakanishi1971aa,Itzykson1980aa} and the universal topological polynomials of graph theory was explicited.
This was done for theories with ordinary propagators of the Laplace type, whose parametric 
representation is based on the heat kernel. These theories were defined either on 
ordinary flat commutative space or on noncommutative Moyal-Weyl flat noncommutative space. The parametric polynomials turned out to be evaluations of the multivariate version of the Tutte polynomial \citep[see][]{Sokal2005aa} in the commutative case and of the Bollob\'as-Riordan one in the noncommutative case \citep{Moffatt2008ab}. 

However heat-kernel based noncommutative theories such as the $\phi^{\star 4}_4$ model 
show a phenomenon called UV/IR mixing which usually prevents
them from being renormalizable. The first renormalizable noncommutative 
quantum field theory, the Grosse-Wulkenhaar model \citep{GrWu03-1,GrWu04-3,Rivasseau2005bh},
is based on a propagator made out of the Laplacian plus a harmonic potential,
hence the parametric representation of these models involve the Mehler kernel rather
than heat kernel. The physical interest of such theories also stems from the fact that
constant magnetic fields also induce such Mehler-type kernels. 

Since the Mehler kernel is quadratic in direct space, such theories have computable parametric representations but
which are more complicated than the ordinary ones. The corresponding topological polynomials were defined and first studied 
in \citep{gurauhypersyman}, then extended to covariant theories in   
\citep{Tanasa:2007fk}. However a global expression has been found only for the leading
part of these polynomials under rescaling and a full explicit solution 
was not found until now. This is what we provide here.

We have found that the corresponding universal polynomials, defined on ribbon graphs with flags, are not based on 
the usual contraction-deletion relations on ordinary graphs but on slightly generalized notions
which involve four canonical operations which act on them,
the usual deletion and contraction plus an anticontraction and a
superdeletion. These last two operations are analogous to contraction
and deletion, but create extra flags. Moreover, our new polynomial is covariant under Chmutov's partial duality \citep{Chmutov2007aa}, thus extending the invariance property of the multivariate  Bollob\'as-Riordan polynomial \citep{Vignestourneret2008aa}.\\

This paper is organized as follows. In section \ref{sec:ribbon-graphs}
the definitions of ribbon graphs (with flags) and of partial duality
are given.

Section \ref{sec:biject-betw-class} is a mathematical
prelude to the study of the polnomials defining the parametric
representation of the Grosse-Wulkenhaar model. There we define
bijections between several classes of sub(ribbon)graphs. 

In section \ref{sec:some-proofs-remarks} the new polynomial
is defined, together with its redution relation, relationship with
other known polynomials and properties under partial duality.

In section \ref{sec:feynm-ampl-grosse} the Grosse-Wulkenhaar model  and its parametric representation is recalled,
following closely the notations of \citep{gurauhypersyman}. 

In section \ref{sec:hyperb-polyn-as} we prove that the corresponding topological polynomials are particular evaluations
of the topological polynomial of section \ref{sec:some-proofs-remarks}.

In section \ref{sec:vari-limit-cases} various limits of the model are performed. The particular case of the 
commutative limit is computed and the corresponding commutative Mehler-based
Symanzik polynomials are written down.

\section{Ribbon graphs}
\label{sec:ribbon-graphs}

There are several equivalent definitions of ribbon graphs:
topological, combinatorial, in between. We will first give the
topological definition and some basic facts about ribbon graphs. Then
we will give a purely combinatorial definition which allows us to
slightly generalize ribbon graphs to ribbon graphs with flags.

\begin{rem}
  In the following, and unless explicitely stated, when we write
  \emph{graph}, the reader should read \emph{ribbon graph}. 
\end{rem}

\subsection{Basics}
\label{sec:basics}

A ribbon graph $G$ is a (not necessarily orientable) surface with boundary represented as the union of two sets of closed topological discs called vertices $V(G)$ and edges $E(G)$. These sets satisfy the following:
\begin{itemize}
  \item vertices and edges intersect by disjoint line segment,
  \item each such line segment lies on the boundary of precisely one vertex and one edge,
  \item every edge contains exactly two such line segments.
\end{itemize}
Figure \ref{RibbonEx1} shows an example of a ribbon graph. Note that
we allow the edges to twist (giving the possibility to the surfaces
associated to the ribbon graphs to be non-orientable). A priori an
edge may twist more than once but the polynomials we are going to consider only depend on the parity of the number of twists (this is indeed the relevant information to count the boundary components of a ribbon graph) so that we will only consider edges with at most one twist.
\begin{figure}[!htp]
  \centering
  \includegraphics[scale=.8]{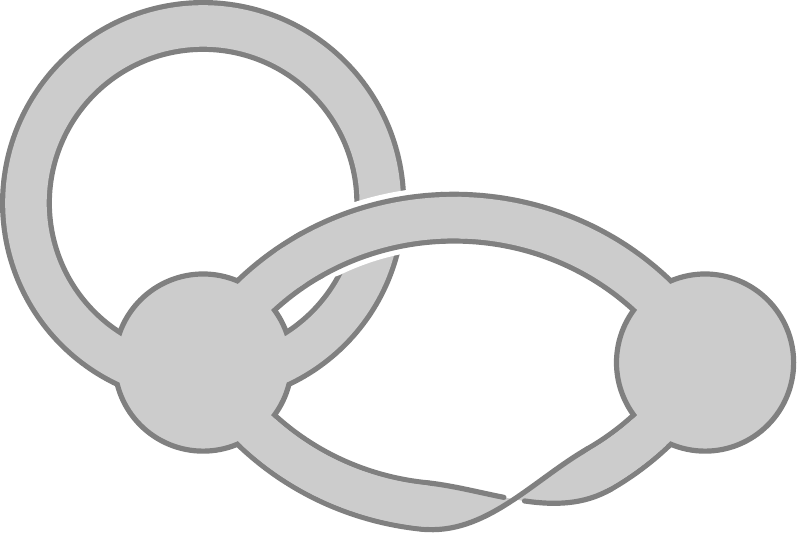}
  \caption{A ribbon graph}
  \label{RibbonEx1}
\end{figure}

\begin{defn}[Notations]\label{def:notations}
  Let $G$ be a ribbon graph. In the rest of this article, we will use
  the following notations:
\begin{itemize}
\item $v(G)=\card V(G)$ is the number of vertices of $G$,
\item $e(G)=\card E(G)$ is the number of edges of $G$,
\item $k(G)$ its number of components,
\item for all $E'\subset E(G)$, $F_{E'}$ is the spanning sub(ribbon)
  graph of $G$ the edge-set of which is $E'$ and
\item for all $E'\subset E(G)$, $E'^{c}\defi E(G)\setminus E'$.
\end{itemize}
\end{defn}

\paragraph{Loops}
\label{sec:loops}

Contrary to the graphs, the ribbon graphs may contain four different
kinds of loops. A loop may be \textbf{orientable} or not, a
\textbf{non-orientable} loop being a twisting edge. Let us consider
the general situations of figure \ref{fig:loopRibbon}. The boxes $A$
and $B$ represent any ribbon graph so that the picture \ref{OrLoop}
(resp.\@ \ref{NonOrLoop}) describes any ribbon graph $G$ with an
orientable (resp.\@ a non-orientable) loop $e$ at vertex $v$. A loop is
said \textbf{nontrivial} if there is a path in $G$ from $A$ to $B$
which is not the trivial path only made of $v$. If not the loop is called \textbf{trivial} \citep{Bollobas2002aa}.
\begin{figure}[!htp]
  \centering
  \subfloat[An orientable loop]{{\label{OrLoop}}\includegraphics[scale=.8]{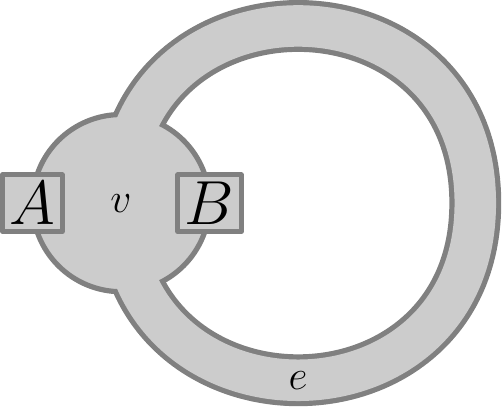}}\hspace{2cm}
  \subfloat[A non-orientable loop]{{\label{NonOrLoop}}\includegraphics[scale=.8]{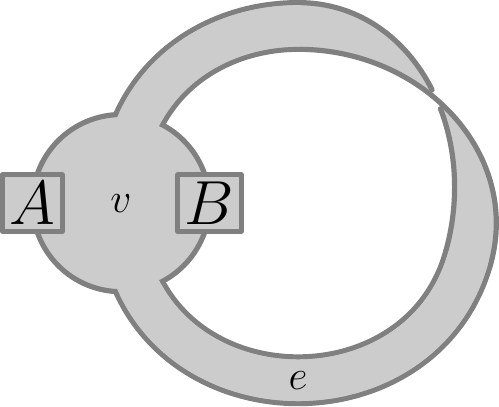}}
  \caption{Loops in ribbon graphs}
  \label{fig:loopRibbon}
\end{figure}

\subsection{Combinatorial maps}
\label{sec:combinatorial-maps}

Based on \cite{Tutte1984aa}, we slightly generalize the notion of combinatorial
map to combinatorial map with flags. We will use it as a
(purely combinatorial) definition for (possibly non-orientable) ribbon
graphs with flags.

\begin{defn}[Combinatorial map with flags]\label{def:CombMaps}
  Let $X$ be a finite set of even cardinality. Its members are called \emph{crosses}. A combinatorial map with flags on $X$ is a triple $(\sigma_{0},\theta,\sigma_{1})$ of permutations on $X$ which obey the following axioms:
  \begin{enumerate}
    \renewcommand{\labelenumi}{A.\theenumi}
  \item $\theta^{2}=\sigma_{1}^{2}=\id$ and $\theta\sigma_{1}=\sigma_{1}\theta$.\label{ax:edge1}
  \item $\theta$ is fixed-point free. Moreover if $x$ is any cross, $\theta x$ and $\sigma_{1}x$ are distinct.\label{ax:edge2}
  \item $\sigma_{0}\theta =\theta\sigma_{0}^{-1}$.\label{ax:vertex1}
  \item For each cross $x$, the orbits of $\sigma_{0}$ through $x$ and $\theta x$ are distinct.\label{ax:vertex2}
  \end{enumerate}
\end{defn}

Let us now explain why such a combinatorial map describes a ribbon
graph $G$ with flags. The involution $\theta$ being fixed-point free,
the set $X$ is partitioned into pairs of the form $\lb x,\theta x\rb$,
namely the orbits of $\theta$. The involution $\sigma_{1}$ may have
fixed points. Note that if $x$ is a fixed point of $\sigma_{1}$, so is
$\theta x$ because $\theta$ and $\sigma_{1}$ commute, see axiom \ref{ax:edge1}. The pairs $\lb x,\theta x\rb$ of fixed points of $\sigma_{1}$ form the set $F(G)$ of \textbf{flags} of $G$.

Let us denote by $F_{X}$ the set of fixed points of $\sigma_{1}$. Then
$X\setminus F_{X}\fide H_{X}$ has a cardinality which is a multiple of
$4$. $H_{X}$ is partitioned into the orbits of $\theta$. The set
$H(G)$ of pairs $\lb x,\theta x\rb,\,x\in H_{X}$ is the set of
\textbf{half-edges} of $G$. $H_{X}$ can also be partitioned into the
orbits of the group $E_{G}$ generated by $\theta$ and
$\sigma_{1}$. Each orbit is of the form $\lb x,\theta x, \sigma_{1}
x,\sigma_{1}\theta x\rb$. Thanks to axiom \ref{ax:edge2}, the members
of a given orbit are all distinct. Each orbit contains two distinct
half-edges and is therefore called an \textbf{edge}. We write $E(G)$
for the set of orbits of $E_{G}$ on $H_{X}$. It is the set of edges of
$G$.

The elements of the set $\HR(G)\defi F(G)\cup H(G)$ made of the orbits of $\theta$ on
$X$ are called \textbf{half-ribbons}.

Finally we describe the vertices of $G$. $\sigma_{0}$ being a permutation, $X$ can be partitioned into its cycles. Each cycle is of the form $C(\sigma_{0},x)\defi (x,\sigma_{0}x,\dotsc,\sigma_{0}^{m-1}x)$ where $m$ is the least integer such that $\sigma_{0}^{m}x=x$. Thanks to axiom \ref{ax:vertex2}, the cycles through $x$ and $\theta x$ are distinct. But they have the same length. Indeed $\sigma_{0}^{m}x=x\iff \theta x=\theta\sigma_{0}^{-m}x\iff \theta x=\sigma_{0}^{m}\theta x$ thanks to axiom \ref{ax:vertex1}. The cycle $C(\sigma_{0},\theta x)$ can be formed from $C(\sigma_{0},x)$:
\begin{align}
  C(\sigma_{0},\theta x)=&(\theta x,\sigma_{0}\theta x,\dotsc,\sigma_{0}^{m-1}\theta x)\\
  =&(\theta x,\theta\sigma_{0}^{-1}x,\dotsc,\theta\sigma_{0}^{-m+1}x)\\
  =&(\theta x,\theta\sigma_{0}^{m-1}x,\dotsc,\theta\sigma_{0}x).
\end{align}
Thus $C(\sigma_{0},\theta x)$ is formed from $C(\sigma_{0},x)$ by reversing the cyclic order of the elements and then premultiplying each by $\theta$. We express this relation by saying that $C(\sigma_{0},x)$ and $C(\sigma_{0},\theta x)$ are \emph{conjugate}. A pair of conjugate orbits of $\sigma_{0}$ is called a \textbf{vertex} of $G$.\\

We now examplify the previous definition with the ribbon graph $G$ of
figure \ref{fig:RibbonFlags}. The set of crosses is
$X=[1,12]\cap\bZ$. Using the cyclic representation, the three
permutations defining this graph are:
\begin{subequations}
  \begin{align}
    \sigma_{0}=&(1,3)(4,2)(6,9,11,8)(5,7,12,10),\\
    \theta=&(1,2)(3,4)(5,6)(7,8)(9,10)(11,12),\\
    \sigma_{1}=&(1,5)(2,6)(3,8)(4,7)(9)(10)(11)(12).
  \end{align}
\end{subequations}
\begin{figure}[!htp]
  \centering
  \includegraphics[scale=.8]{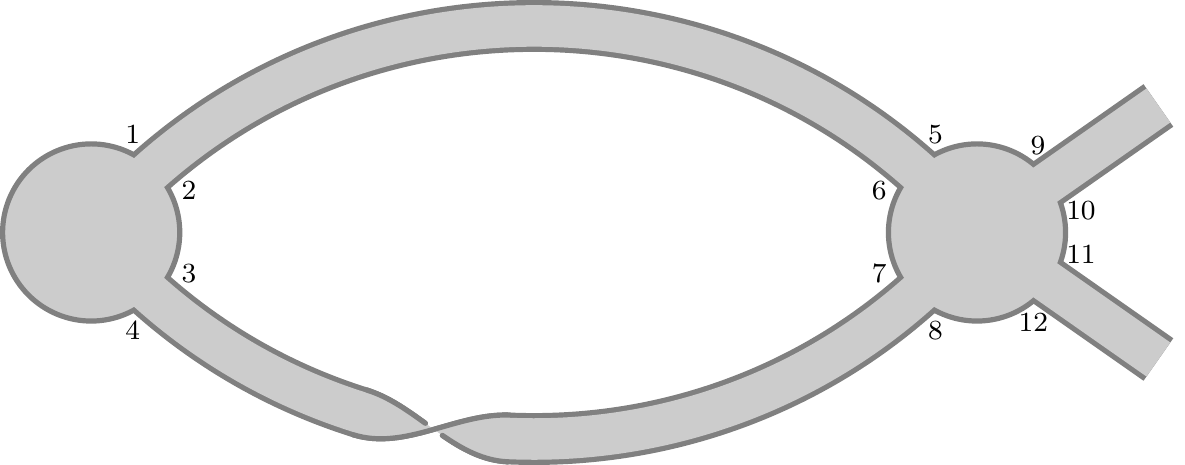}
  \caption{A ribbon graph $G$ with flags}
  \label{fig:RibbonFlags}
\end{figure}
As noticed above, the set $X$ is partitioned into pairs which are the
orbits of $\theta$. Those pairs which are fixed by $\sigma_{1}$ are
called flags:
\begin{align}
  F(G)=&\set{\set{9,10},\set{11,12}}.
\end{align}
The half-edges of $G$ are the orbits of $\theta$ which are not fixed
by $\sigma_{1}$:
\begin{align}
  H(G)=&\set{\set{1,2},\set{3,4},\set{5,6},\set{7,8}}
  \intertext{and the edges of $G$ are thus}
  E(G)=&\set{\set{1,2,5,6},\set{3,4,7,8}}.
\end{align}
Finally, $G$ has two vertices:
\begin{align}
  v_{1}=&\set{(1,3),(2,4)},\quad v_{2}=\set{(6,9,11,8),(5,7,12,10)}.
\end{align}

\begin{rem}
  A ribbon graph \emph{without} flag is represented by three
  permutations $\sigma_{0},\theta$ and $\sigma_{1}$ obeying
  definition \ref{def:CombMaps} with, in addition, $\sigma_{1}$
  fixed-point free.
\end{rem}

\begin{defn}[Subgraphs]\label{def:subgraphs}
  Let $G$ be a ribbon graph, possibly with flags. A subgraph of $G$
  consists in a graph, the edge-set of which is a subset of
  $E(G)$ and the flag-set of which is a subset of $F(G)$. A \emph{cutting}
  subgraph of $G$ is a graph the half-ribbon-set of which is a subset
  of $\HR(G)$. By convention, if the half-ribbon set of a cutting
  spanning subgraph contains the two halfs of an edge, the subgraph
  contains this edge. The set of spanning (cutting) subgraph of $G$ is
  $\Subg(G)$ ($\SubgF(G)$).

  Moreover if the edges and flags of $G$ are labelled, the (cutting)
  subgraphs of $G$ inherit the labels of $G$ with the following
  convention: the two half-edges of a given edge of $G$ share the same
  label in the cutting subgraphs of $G$.

\end{defn}
In contrast to a subgraph, a cutting subgraph may have flags coming
both from the flags of $G$ and from half-edges of $G$. Each edge of
$G$ is made of two half-edges. A subgraph contains, in particular, some
of the edges of $G$ whereas a cutting subgraph may contain a half-edge
of an edge without taking the other half, see figure
\ref{fig:subgraphs} for examples.
\begin{figure}[!htp]
  \centering
  \subfloat[A subgraph of the graph of figure \ref{RibbonEx1}]{{\label{fig:subgraphEx}}\includegraphics[scale=.6]{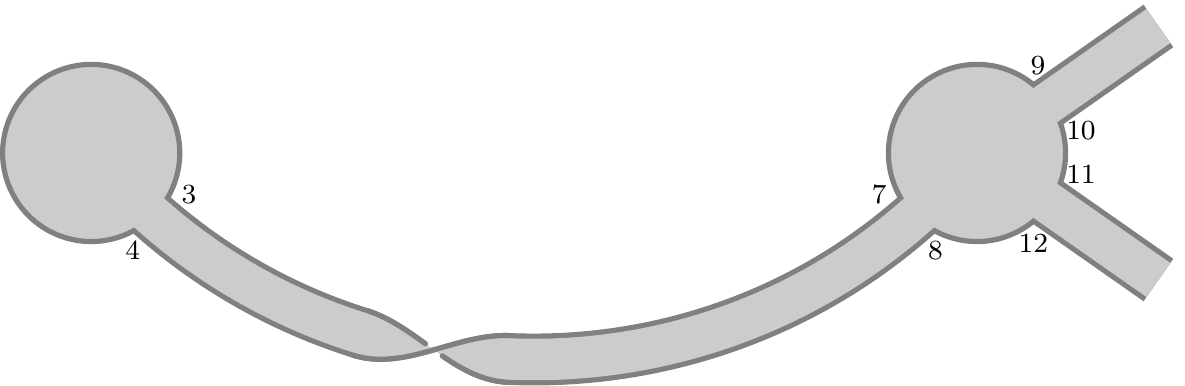}}\hspace{1cm}
  \subfloat[A cutting subgraph of the graph of figure \ref{RibbonEx1}]{{\label{fig:CutSubgraph}}\includegraphics[scale=.6]{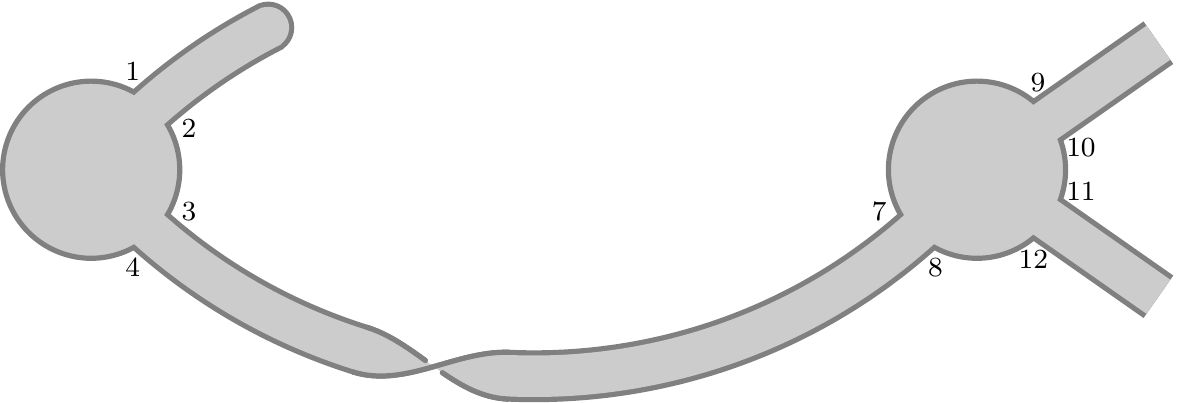}}\\
  \caption{Subgraphs}
  \label{fig:subgraphs}
\end{figure}

\subsection{Operations on edges}
\label{sec:operations-edges}

From a ribbon graph with flags, we can define two other ribbon graphs
with flags either by deleting an edge or by cutting it:
\begin{defn}[Operations on ribbon graphs with flags]\label{def:operationsEdges}
  Let $G$ be a ribbon graph with flags and $e\in E(G)$ any of its edges. We define the two following operations:
  \begin{itemize}
  \item the \textbf{deletion} of $e$, written $G-e$,
  \item the \textbf{cut} of $e$, written $G\cut e$, which consists in replacing $e$ by two flags attached at the former end-vertices (or end-vertex) of $e$, respecting the cyclic order at these (this) vertices (vertex).
  \end{itemize}
\end{defn}
In the combinatorial map representation of a ribbon graph $G$, an edge $e$
corresponds to a set of four crosses:
$e=\set{x_{1},x_{2},x_{3},x_{4}},\,\forall\, 1\les i\les 4,\,x_{i}\in
X(G)$. The graph $G-e$ has $X\setminus e$ as set of crosses and the
restriction of $\sigma_{0},\theta$ and $\sigma_{1}$ to $X\setminus e$
as defining permutations.

Let $\phi$ be any member of the group generated by $\sigma_{0},\theta$ and $\sigma_{1}$. For any subset $E'\subset X$, we let $\phi_{E'}$ be the following map:
\begin{align}
  \phi_{E'}\defi&%
  \begin{cases}
    \phi&\text{on }E',\\
    \id&\text{on }X\setminus E'\fide\bar{E'}.
  \end{cases}
\end{align}
The graph $G\cut e$ is defined on the same crosses as $G$ and given by
$\sigma_{0},\theta$ and $\sigma_{1X'}$ where $X'=X\setminus e$. For
example, considering the ribbon graph of figure \ref{fig:RibbonFlags},
and if $e=\set{1,2,5,6}$, $G\cut e$ is the ribbon graph, with flags, of figure \ref{fig:ExCut}.
\begin{figure}[!htp]
  \centering
  \includegraphics[scale=.8]{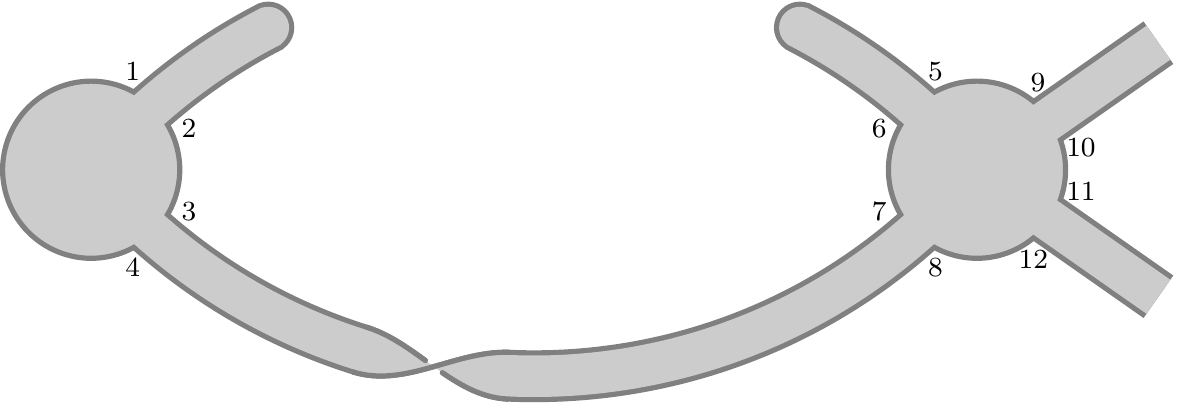}
  \caption{Cutting an edge}
  \label{fig:ExCut}
\end{figure}

\subsection{Natural duality}
\label{sec:natural-duality}

For ribbon graphs without flags, there is a well-known notion of duality, hereafter
called natural duality, also known as Euler-Poincaré duality. From a
given ribbon graph $G$, it essentially consists in forming another
ribbon graph $G^\star$, the vertices of which are the faces (or
boundary components) of $G$ and the faces of which are the vertices of
$G$. The edges of $G^\star$ are in bijection with those of $G$.

Every ribbon graph can be drawn on a surface of minimal genus such
that no two of its edges intersect. To build the dual $G^{\star}$ of
$G$, first draw $G$ on such a surface. Then place a vertex into each
face of $G$. Each such face is homeomorphic to a disk. Then draw an
edge between two vertices of $G^{\star}$ each time the corresponding
faces of $G$ are separated by an edge in $G$.

At the combinatorial level, if $G=(\sigma_{0},\theta,\sigma_{1})$,
then $G^{\star}=(\sigma_{0}\theta\sigma_{1},\sigma_{1},\theta)$, the cycles of $\sigma_{0}\theta\sigma_{1}$ representing the faces of $G$. If
$G$ has flags, we define its natural dual $G^{\star}$ by $(\sigma_{0}\theta_{H_{X}}\sigma_{1},\sigma_{1H_{X}}\theta_{F_{X}},\theta_{H_{X}}\sigma_{1F_{X}})$, see
figure \ref{fig:NatDualEx} for an example.
\begin{figure}[!htp]
  \centering
  \includegraphics[scale=.9]{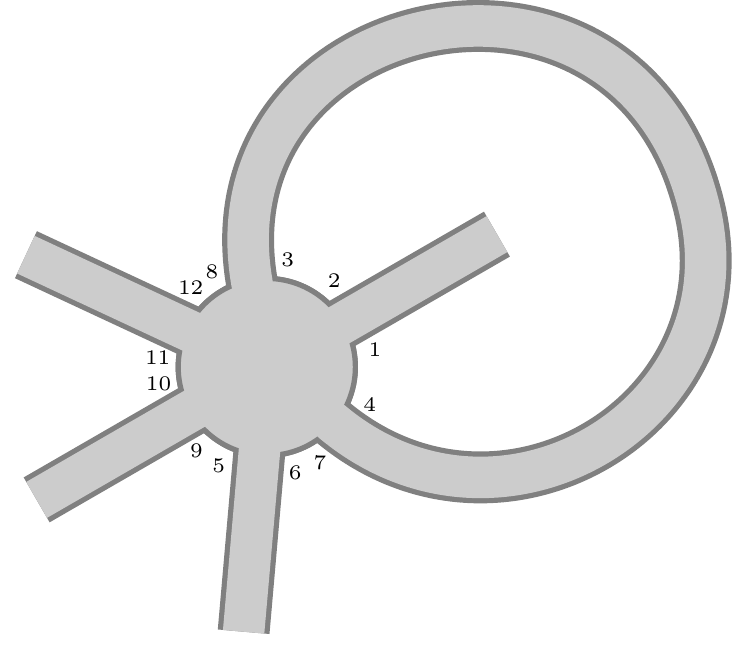}
  \caption{The natural dual of the graph of figure \ref{fig:ExCut}}
  \label{fig:NatDualEx}
\end{figure}

\subsection{Partial duality}
\label{sec:partial-duality}

S.~Chmutov introduced a new ``generalised duality'' for ribbon graphs
which generalises the usual notion of duality (see \citep{Chmutov2007aa}). In \citep{Moffatt2008aa}, I.~Moffatt renamed this new duality as ``partial duality''. We adopt this designation here. We now describe the construction of a partial dual graph and give a few properties of the partial duality.\\

Let $G$ be a ribbon graph and $E'\subset E(G)$. Let
$\check{F}_{E'}\defi G\cut E'^{c}$ be the
spanning subgraph with flags of $G$, the edge-set of which is $E'$ and
the flag-set of which is made of the cut edges in
$E'^{c}=E(G)\setminus E'$. We will
construct the dual $G^{E'}$ of $G$ with respect to the edge-set $E'$. The general idea is the
following. We consider the spanning subgraph with flags
$\check{F}_{E'}$. Then we
build its natural dual $\check{F}_{E'}^{\star}$. Finally we reglue the
edges previously cut in $E'^{c}$.\\

More precisely, at the level of the combinatorial maps, the
construction of the partial dual $G^{E'}$ of $G$ goes as follows:
\begin{center}
  \begin{tikzpicture}[>=stealth]
    \definecolor{fabgray}{gray}{.6};
    \foreach \u in {9} 
    {
    \node (un) at (0,0) [shape=rectangle]
    {$G=(\sigma_{0},\theta,\sigma_{1})$};
    \node (deux) at (\u,0) [shape=rectangle]
    {$\check{F}_{E'}=(\sigma_{0},\theta,\sigma_{1E'})$};
    \node (trois) at (\u,-.5*\u) [shape=rectangle] {$\check{F}^{\star}_{E'}=(\sigma_{0}\theta_{E'}
      \sigma_{1E'},\sigma_{1E'}\theta_{E'^{c}},\theta_{E'})$};
    \node (quatre) at (0,-.5*\u) [shape=rectangle] {$G^{E'}=(\sigma_{0}\theta_{E'}
      \sigma_{1E'},\sigma_{1E'}\theta_{E'^{c}},\sigma_{1E'^{c}}\theta_{E'})$};
    \draw [thick,->,fabgray] (un.east)--(deux.west) node [above,text width=3cm,text
    centered,midway] {\color{black} cut};
    \draw [thick,->,fabgray] (deux.south)--(trois.north) node [right,text width=3cm,text
    centered,midway] {\color{black} natural duality};
    \draw [thick,->,fabgray] (trois.west)--(quatre.east) node [above,text width=3cm,text
    centered,midway] {\color{black} glue};
    \draw [thick,->,fabgray] (un.south)--(quatre.north) node [left,text width=3cm,text
    centered,midway] {\color{black} partial duality};
  }
  \end{tikzpicture}
\end{center}
Figure \ref{PartDualEx} shows an example of the construction of a
partial dual. The direct ribbon graph is drawn on figure
\ref{fig:ribbonEx}. We choose $E'=\set{\set{3,4,7,8}}$ and the
subgraph $\check{F}_{E'}$ is depicted on figure
\ref{fig:ribbonExReprArrow}. Its natural dual $\check{F}^{\star}_{E'}$
  is on figure \ref{fig:BoundComp}. Finally the partial dual $G^{E'}$
  of $G$ is shown on figure \ref{fig:DualExComb}.
\begin{figure}[!htp]
  \centering
  \subfloat[A ribbon graph $G$]{{\label{fig:ribbonEx}}\includegraphics[scale=.6]{ribbons-1}}\hspace{1cm}
  \subfloat[The subgraph $\check{F}_{E'}$ with $E'=\set{\set{3,4,7,8}}$]{{\label{fig:ribbonExReprArrow}}\includegraphics[scale=.6]{ribbons-2}}\\
  \subfloat[The natural dual $\check{F}^{\star}_{E'}$ of $\check{F}_{E'}$]{{\label{fig:BoundComp}}\includegraphics[scale=.6]{ribbons-4}}\hspace{1cm}
  \subfloat[The partial dual $G^{E'}$ of $G$]{{\label{fig:DualExComb}}\includegraphics[scale=.6]{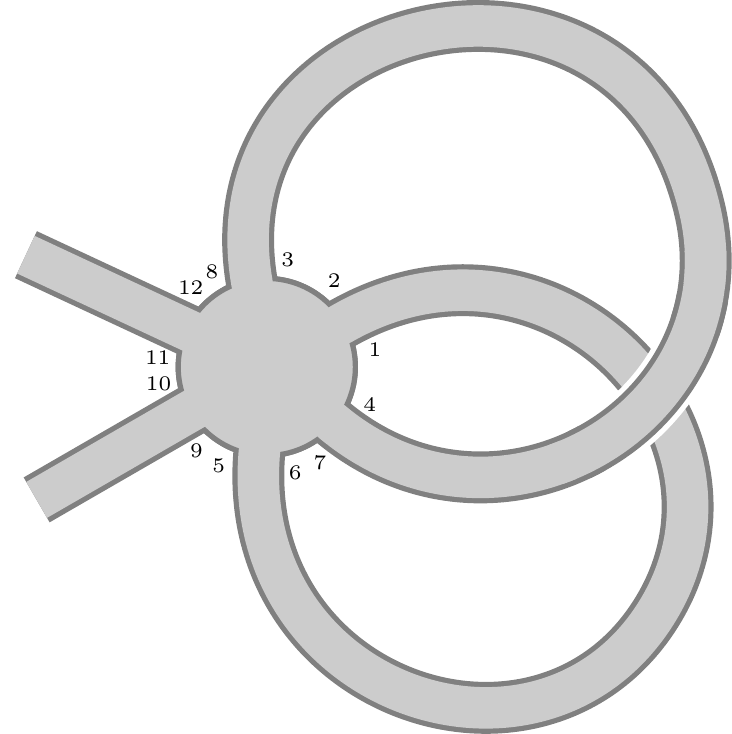}}
  \caption{Construction of a partial dual}
  \label{PartDualEx}
\end{figure}\\

S.~Chmutov proved among other things the following basic properties of the partial duality:
\begin{lemma}[\citep{Chmutov2007aa}]
  \label{lem:SimpleProp}
  For any ribbon graph $G$ and any subset of edges $E'\subset E(G)$, we have
  \begin{itemize}
  \item $(G^{E'})^{E'}=G$,
  \item $G^{E(G)}=G^{\star}$ and
  \item if $e\notin E'$, then $G^{E'\cup\{e\}}=(G^{E'})^{\{e\}}$.
  \end{itemize}
\end{lemma}
His proof relies on graphical and commonsensical arguments. Here we
would like to point out that the combinatorial map point of view allows
very direct algebraic proofs.

For example, an interesting exercise consists in checking that the partial duality is an involution, namely that $(G^{E'})^{E'}=G$:\\ $(G^{E'})^{E'}=\lbt\sigma_{0}\theta_{E'}\sigma_{1E'}(\sigma_{1E'}\theta_{E'^{c}})_{E'}(\sigma_{1E'^{c}}\theta_{E'})_{E'},(\sigma_{1E'^{c}}\theta_{E'})_{E'}(\sigma_{1E'}\theta_{E'^{c}})_{E'^{c}},(\sigma_{1E'^{c}}\theta_{E'})_{E'^{c}}(\sigma_{1E'}\theta_{E'^{c}})_{E'}\rbt=(\sigma_{0},\theta,\sigma_{1})$.\\

We can also prove that for any subset $E'\subset E(G)$ and any $e\in E'^{c}$, $(G^{E'})^{\{e\}}=G^{E'\cup\{e\}}$.

\begin{proof}
  We define $E''\defi E'\cup\{e\}$.
  \begin{align}
    G^{E''}=&(\sigma_{0}\theta_{E''}\sigma_{1E''},\sigma_{1E''}\theta_{E''^{c}},\sigma_{1E''^{c}}\theta_{E''}) \\
    \sigma_{0}\big((G^{E'})^{\{e\}}\big)=&\sigma_{0}\theta_{E'}\sigma_{1E'}(\sigma_{1E'}\theta_{E'^{c}})_{e}(\sigma_{1E'^{c}}\theta_{E'})_{e}=\sigma_{0}\theta_{E'}\sigma_{1E'}\theta_{e}\sigma_{1e}=\sigma_{0}\theta_{E''}\sigma_{1E''} \\
    \theta\big((G^{E'})^{\{e\}}\big)=&(\sigma_{1E'^{c}}\theta_{E'})_{e}(\sigma_{1E'}\theta_{E'^{c}})_{e^{c}}=\sigma_{1e}\sigma_{1E'}\theta_{E'^{c}\setminus\{e\}}=\sigma_{1E''}\theta_{E''^{c}}\\
    \sigma_{1}\big((G^{E'})^{\{e\}}\big)=&(\sigma_{1E'^{c}}\theta_{E'})_{e^{c}}(\sigma_{1E'}\theta_{E'^{c}})_{e}=\sigma_{1E'^{c}\setminus\{e\}}\theta_{E'}\theta_{e}=\sigma_{1E''^{c}}\theta_{E''}
  \end{align}
\end{proof}

The partial duality allows an interesting and fruitful definition of the contraction of an edge:
\begin{defnb}[Contraction of an edge \citep{Chmutov2007aa}]\label{def:Contraction}
  Let $G$ be a ribbon graph and $e\in E(G)$ any of its edges. We define the contraction of $e$ by:
  \begin{align}
    G/e\defi&G^{\{e\}}-e.\label{eq:ContractionDef}
  \end{align}
\end{defnb}
From the definition of the partial duality, one easily checks that,
for an edge incident with two different vertices, the definition
\ref{def:Contraction} coincides with the usual intuitive contraction
of an edge. The contraction of a loop depends on its orientability,
see figures \ref{fig:OrLoopContraction} and
\ref{fig:NonOrLoopContraction}.

Different definitions of the contraction of a loop have been used in
the litterature. One can define $G/e\defi G-e$. In
\citep{Huggett2007aa}, S.~Huggett and I.~Moffatt give a definition
which leads to surfaces which are not ribbon graphs anymore. The
definition \ref{def:Contraction} maintains the duality between
contraction and deletion.
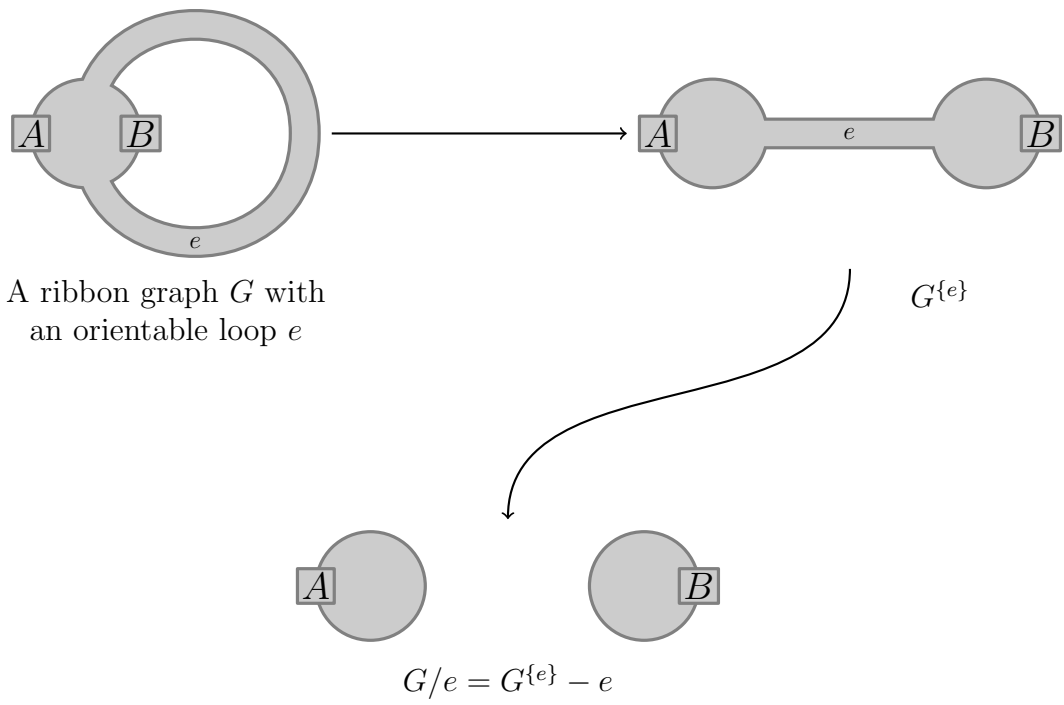
\begin{figure}[!htp]
  \centering
  \begin{tikzpicture}
    \foreach \u in {9} 
    {
    \node (un) at (0,0) [shape=rectangle,label={[text width=5cm,text centered]below:A ribbon graph $G$ with an orientable loop $e$}]
    {\includegraphics[scale=.8]{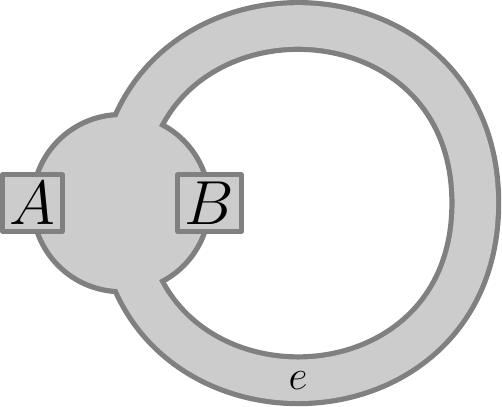}};
    \node (deux) at (\u,0) [shape=rectangle,label=-70:$G^{\{e\}}$]
    {\includegraphics[scale=.8]{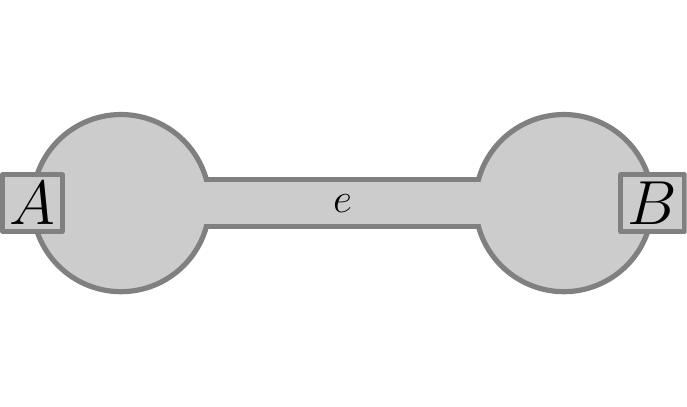}};
    \node (trois) at (\u/2,-2*\u/3) [shape=rectangle,label=below:${G/e=G^{\{e\}}-e}$]%
    {\includegraphics[scale=.8]{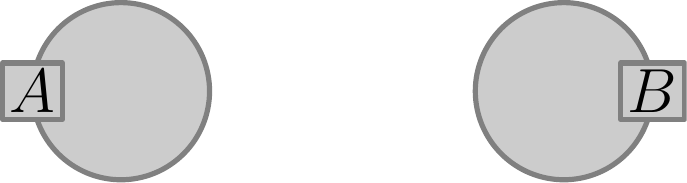}};
    
    \draw [thick,->] (un.east)--(deux.west) 
    ;
    \draw [thick,->] (deux) to [out=-90,in=90] (trois);
}
  \end{tikzpicture}
  \caption{Contraction of an orientable loop}
  \label{fig:OrLoopContraction}
\end{figure}
\begin{figure}[!htp]
  \centering
  \begin{tikzpicture}
    \foreach \u in {9} 
    {
    \node (un) at (0,0) [shape=rectangle,label={[text width=5cm,text centered]below:A ribbon graph $G$ with a non-orientable loop $e$}]
    {\includegraphics[scale=.8]{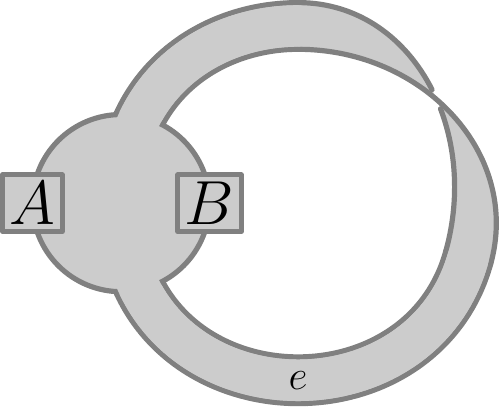}};
    \node (deux) at (\u,0) [shape=rectangle,label=-70:$G^{\{e\}}$]
    {\includegraphics[scale=.8]{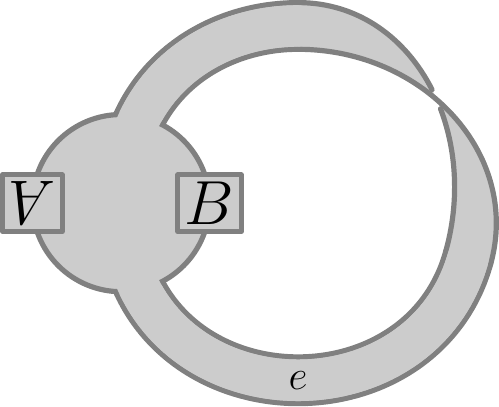}};
    \node (trois) at (\u/2,-2*\u/3) [shape=rectangle,label=below:${G/e=G^{\{e\}}-e}$]%
    {\includegraphics[scale=.8]{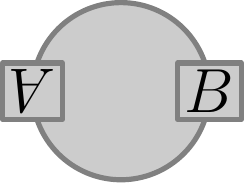}};
    
    \draw [thick,->] (un.east)--(deux.west) 
    ;
    \draw [thick,->] (deux) to [out=-90,in=90] (trois);
}
  \end{tikzpicture}
  \caption{Contraction of a non-orientable loop}
  \label{fig:NonOrLoopContraction}
\end{figure}
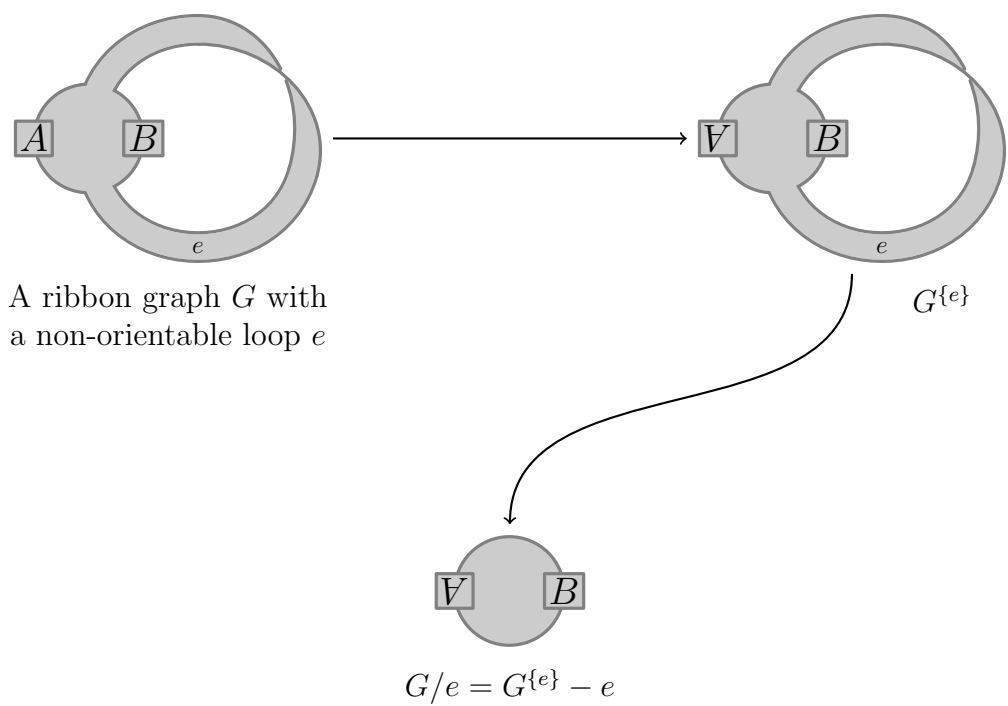

\section{Bijections between classes of subgraphs}
\label{sec:biject-betw-class}

This section consists in a mathematical preliminary to the study of
the $\HU$ polynomial, introduced in section
\ref{sec:feynm-ampl-grosse}. This ribbon graph invariant is a key
ingredient of the parametric representation of the Grosse-Wulkenhaar
model amplitudes.\\

Let $G$ be a ribbon graph. For any $E'\subset E(G)$, there
is a natural bijection between $E(G)$ and $E(G^{E'})$. This leads to a bijection between the spanning
subgraphs of $G$ and those of $G^{E'}$. In particular, it is true for
$E'=\set e$ with $e\in E(G)$. Representing a bijective map
by the following symbol $\bij$, we have:
\begin{align}
  s:\Subg(G)\bij&\ \Subg(G^{\set e}),\ |\Subg(G)|=|\Subg(G^{\set e})|=2^{|E(G)|}.\label{eq:Bijsubg}
\end{align}
The map $s$ extends trivially on ribbon graphs with flags. In the
following, we will be interested in maps betweens different classes of
subgraphs. We are going to generalize $s$ to odd and even
(cutting) (colored) subgraphs. A special case of these bijections will
be used in section \ref{sec:crit-model-omeg1} to prove the factorization of the
polynomial $\HU(G;\bt,\mathbf{1})$.

\subsection{Subgraphs of fixed parity}
\label{sec:parity-fixed-subgr}

\begin{defn}[Odd and even graphs]
  A (ribbon) graph (with flags) is said of fixed parity if all the
  degrees of its vertices have the same parity. It is odd (resp.\@
  even) if all its vertices are of
  odd (resp.\@ even) degrees. Given a ribbon graph $G$, with or without flags, we denote by
  $\Odd(G)$ (resp.\@ $\Ev(G)$) the set of odd (resp.\@ even) spanning
  subgraphs of $G$.
\end{defn}
We would like to know if the bijection $s$ of equation
(\ref{eq:Bijsubg}) preserves the subclasses of odd (even)
subgraphs. It is easy to see that it is not the case, as the following
example shows.\\

Let us consider the ribbon graph $G$ made of two vertices and two
edges between those two vertices. $G$ is sometimes called a (planar)
$2$-banana, see figure \ref{fig:2banana}.
\begin{figure}[!htp]
  \centering
  \subfloat[$G=$ a
  $2$-banana]{{\label{fig:2banana}}\includegraphics[scale=.8]{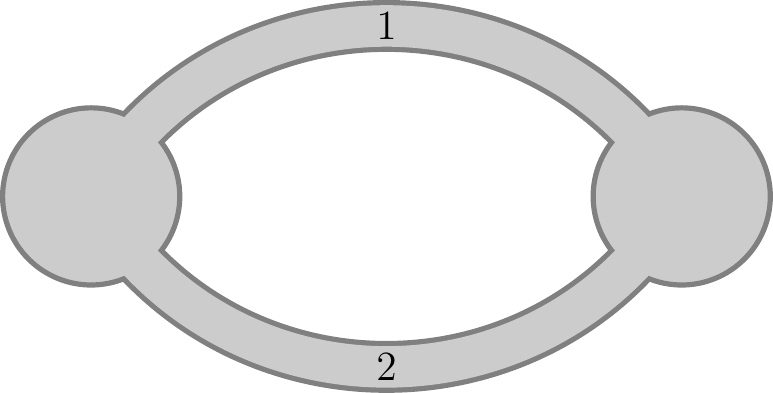}}\hspace{2cm}
  \subfloat[$G^{\set 1}=$ a non-planar
  $8$]{{\label{fig:NP8}}\includegraphics[scale=.8]{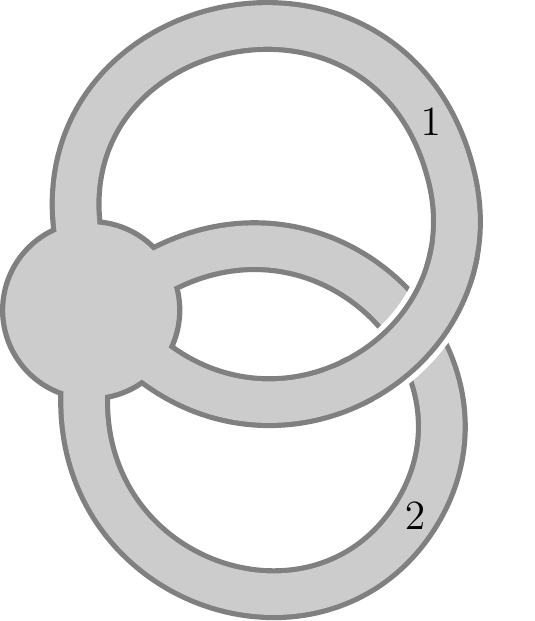}}
  \caption{Partial duals}
  \label{fig:ExBijSubg}
\end{figure}
We have $\Odd(G)=\set{\set{1},\set{2}}$,
$\Ev(G)=\set{\emptyset,\set{1,2}}$ whereas $\Odd(G^{\set
  1})=\emptyset$ and $\Ev(G^{\set 1})=\set{\emptyset,\set 1,\set
  2,\set{1,2}}$. This means that there exist graphs and edges such
that $s$ does not preserve the classes of odd and even subgraphs. Note
however that there may be graphs $G$ and/or subsets $E'\subset E(G)$
such that the natural bijection between subgraphs of $G$ and $G^{E'}$
let the classes $\Odd$ and $\Ev$ invariant. This is trivially the case
for self-dual graphs $G$ and $E'=E(G)$. Classifying the graphs and
subsets of edges such that $s$ let some classes of subgraphs invariant
clearly deserves further study. Nevertheless, here, we will restrict
ourselves to bijections valid for any $G$ and any $e\in E(G)$.

\subsection{Colored subgraphs}
\label{sec:colored-subgraphs}

Going back to the example of figure \ref{fig:ExBijSubg}, we have
$|\Ev(G)|=2$ and $|\Ev(G^{\set 1})|=4$ but
$2^{v(G)}|\Ev(G)|=2^{v(G^{\set 1})}|\Ev(G^{\set 1})|=2^{3}$. For any
graph $g$, $2^{v(g)}$ is the number of colorings of $V(g)$ with two
colors. This means that there exists a bijection between the colored even
subgraphs of $G$ and $G^{\set 1}$. This is actually true for any
ribbon graph with flags and any edge.
\begin{rem}
  This is clearly not the case for the odd subgraphs, as shows the
  example of the $2$-banana. Note also that, in general, there is no bijection
  between the colored subgraphs of a graph $G$ and of its partial duals
  $G^{\set e}$, the number of vertices of $G$ and $G^{\set e}$ being
  usually different.
\end{rem}
\begin{defn}[Colored graphs]\label{def:ColGraphs}
  A colored (ribbon) graph $G$ is a (ribbon) graph and a subset $C(G)$
  of $V(G)$. The set of colored odd (resp.\@ even) subgraphs of $G$ is
  denoted by $\COdd(G)$ (resp.\@ $\CEv(G)$).
\end{defn}
\begin{lemma}\label{lem:CEven}
  Let $G$ be a ribbon graph with flags. For any edge $e\in E(G)$,
  there is a bijection between $\CEv(G)$ and $\CEv(G^{\set e})$.
\end{lemma}
\begin{proof}
  \begin{align}
    \Ev(G)=&\set{B\subset E(G)\tqs F_{B}\text{ is even}}\\
    =&\set{B'\subset
      E(G)\setminus\set e\tqs F_{B'}\text{ is even}}\cup
    \set{B'\subset E(G)\setminus\set e\tqs F_{B'\cup\set e}\text{ is
        even}}\\
    =&\bigcup_{B'\subset E(G)\setminus\set e}\set{B\in\set{B',B'\cup\set
      e}\tqs F_{B}\text{ is even}}\\
  \fide&\bigcup_{B'\subset E(G)\setminus\set e}\Ev_{B'}(G).
 \end{align}
 For any $B',B''\subset E(G)\setminus\set e,
 \Ev_{B'}(G)\cap\Ev_{B''}(G)=\emptyset$. Moreover $\Ev_{B'}(G)$ may have a cardinality of $0,1$ or $2$. We now prove
 that $\forall B'\subset E(G)\setminus\set e,
 |\CEv_{B'}(G)|=|\CEv_{B'}(G^{\set e})|$, which would prove lemma
 \ref{lem:CEven}.

We distinguish between
 two cases: either $e$ is a loop (in $G$) or it is not.
 \begin{enumerate}
 \item \underline{$e$ is a loop}: let $v$ be the endvertex of $e$. It
   may be represented as in figure \ref{fig:BijEvLoop1}.
   \begin{figure}[!htp]
     \centering
     \subfloat[The vertex $v$ in
     $F_{B}$]{{\label{fig:BijEvLoop1}}\includegraphics[scale=.8]{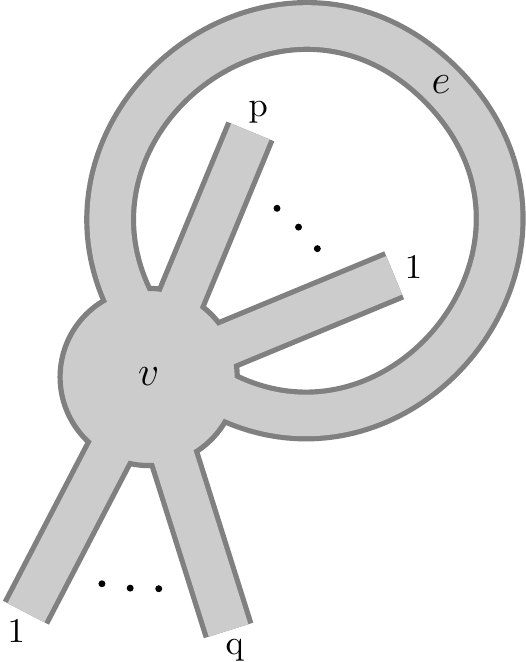}}\hspace{2cm}
     \subfloat[The corresponding situation
     in $F_{B\cup\set e}^{\set e}$]{{\label{fig:BijEvLoop2}}\includegraphics[scale=.8]{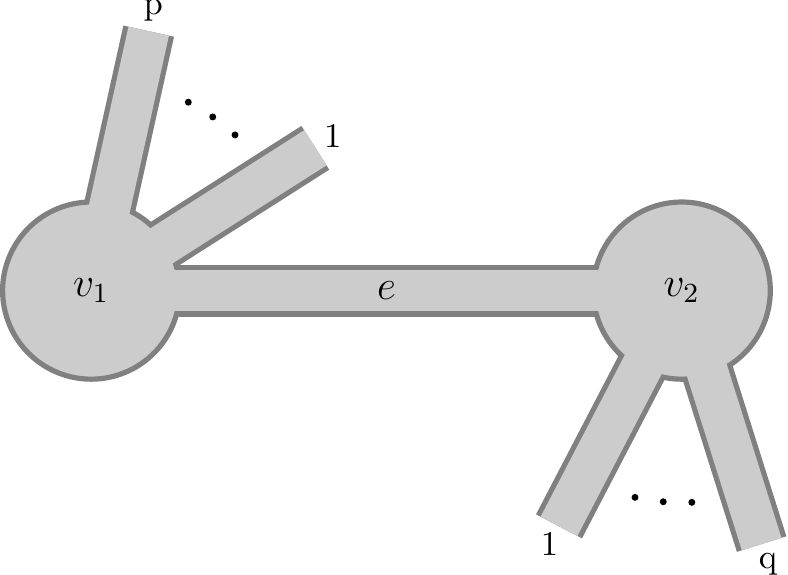}}
     \caption{Bijection in case of a loop}
     \label{fig:BijFig}
   \end{figure}
   \begin{itemize}
   \item $p$ and $q$ have the same parity: $v$ is even in $F_{B}$ with
     or without $e$, then $|\Ev_{B'}(G)|=2$ and
     $|\CEv_{B'}(G)|=2\times 2$. If $p$ and $q$ are odd, $F_{B}$ is even
     in $G^{\set e}$ iff $e\in B$, see figure \ref{fig:BijEvLoop2}. On the contrary, if $p$ and $q$ are
     even, $F_{B}$ is even in $G^{\set e}$ iff $e\notin B$. Then
     $|\CEv_{B'}(G^{\set e})|=1\times 2^{2}$.
   \item $p$ and $q$ have different parities:
     $|\CEv_{B'}(G)|=|\CEv_{B'}(G^{\set e})|=0$.
   \end{itemize}
 \item \underline{$e$ is not a loop}: using $G=\lbt G^{\set
     e}\rbt^{\set e}$, this is the same situation as in the preceding
   case with $G$ replaced by $G^{\set e}$.
 \end{enumerate}
\end{proof}

\subsection{Cutting subgraphs}
\label{sec:cutting-subgraphs}

Both from a mathematical and physical point of view, it is quite
natural to consider not only spanning subgraphs but also spanning
cutting subgraphs of a ribbon graph $G$. For any $e\in E(G)$,
$|\SubgF(G)|=|\SubgF(G^{\set e})|=2^{\HR(G)}=2^{F(G)+2E(G)}$. Thus
there exists a (natural) bijection between those two sets. What about
the odd (resp.\@ even) cutting subgraphs?
\begin{defn}
  Let $G$ be a ribbon graph with flags. We denote by $\OddF(G)$
  (resp.\@ $\EvF(G)$) the set of odd (resp.\@ even) spanning cutting subgraphs of $G$.
\end{defn}
It is easy to check that there is no bijection between $\OddF(G)$
(resp.\@ $\EvF(G)$) and $\OddF(G^{\set e})$ (resp.\@ $\EvF(G^{\set
  e})$). For example, let us consider once more the graphs of figure
\ref{fig:ExBijSubg}. We have $|\OddF(G)|=|\EvF(G)|=4$ whereas
$|\OddF(G^{\set 1})|=|\EvF(G^{\set 1})|=8$.

\subsection{Colored cutting subgraphs}
\label{sec:color-cutt-subgr}

\begin{defn}\label{def:ColCuttGraphs}
  Let $G$ be a ribbon graph with flags. The set of colored cutting
  spanning subgraphs of $G$ is $\CSubgF(G)$. The set of odd (resp.\@ even) colored cutting
  spanning subgraphs of $G$ is denoted by $\COddF(G)$ (resp.\@ $\CEvF(G)$).
\end{defn}
As in the case of colored subgraphs, there is generally no
bijection between $\CSubgF(G)$ and $\CSubgF(G^{\set e})$, because
$v(G)\neq v(G^{\set e})$ usually. Nevertheless we have
\begin{lemma}\label{lem:ColCutSubg}
  Let $G$ be a ribbon graph with flags. For any $e\in E(G)$, there is
  a bijection $\chi_{G}^{\set e}$ between $\COddF(G)$ (resp.\@ $\CEvF(G)$) and $\COddF(G^{\set e})$
  (resp.\@ $\CEvF(G^{\set e})$).
\end{lemma}
\begin{proof}
  Let us denote by $\stackrel{\rightarrow}{e}$ and
  $\stackrel{\leftarrow}{e}$ the two half-edges of $e$. Let us define
  $\langle
  e\rangle\defi\set{\stackrel{\rightarrow}{e},\stackrel{\leftarrow}{e}}$. Recall
  that, by convention (see definition \ref{def:subgraphs}), when both
  halves of an edge $e$ appear in a subset $H\subset\HR(G)$, it means
  that $e\in E(F_{H})$.
 \begin{align}
    \OddF(G)=&\set{H\subset\HR(G)\tqs F_{H}\text{ is odd}}\\
    =&\bigcup_{H'\subset\HR(G)\setminus\langle
      e\rangle}\set{A\subset\langle e\rangle\tqs F_{H\cup A}\text{ is
        odd}}\\
    \fide&\bigcup_{H'\subset\HR(G)\setminus\langle
      e\rangle}\OddF_{H'}(G).
  \end{align}
We define $\EvF_{H'}(G)$ the same way. We prove that $\COddF_{H'}(G)$
and $\COddF_{H'}(G^{\set e})$ have the same cardinality for any
$H'\subset\HR(G)\setminus\langle e\rangle$. We let the case of even
subgraphs to the reader.

Using once more $G=\lbt G^{\set
e}\rbt^{\set e}$, it is enough to prove it in the case $e$ is a loop
in $G$. The situation is thus again the one of figure
\ref{fig:BijFig}, where now figure \ref{fig:BijEvLoop1} represents
the endvertex of $e$ in $F_{H'\cup A\cup\set e}$.
\begin{itemize}
\item If $p$ and $q$ have the same parity, in order $F_{H'\cup A}$ to
  be odd, $A$ contains one of the two
  half-edges of $e$: $e$ is a flag in $F_{H'\cup A}$. Thus
  $|\COddF_{H'}(G)|=2\times 2$. If $p$ and $q$ are odd, $F_{H'\cup A}$ is odd
     in $G^{\set e}$ iff $A=\emptyset$, see figure \ref{fig:BijEvLoop2}. On the contrary, if $p$ and $q$ are
     even, $F_{H'\cup A}$ is odd in $G^{\set e}$ iff $A=\set{\stackrel{\rightarrow}{e},\stackrel{\leftarrow}{e}}$. Then
     $|\COddF_{H'}(G^{\set e})|=1\times 2^{2}$.
   \item If $p$ and $q$ have different parities: $F_{H'\cup A}$ is odd
     iff $A=\emptyset$ or
     $\set{\stackrel{\rightarrow}{e},\stackrel{\leftarrow}{e}}$ which
     implies $|\COddF_{H'}(G)|=2\times 2$. Let us say that $p$ is odd
     and $q$ even. There is only one possibility for $A$ such that
     $F_{H'\cup A}$ is odd. Namely $A$ should only contain the
     half-edge of $e$ which is hooked to the vertex incident with the
     other $q$ half-ribbons. Thus $|\COddF_{H'}(G^{\set e})|=1\times 2^{2}$.
\end{itemize}
\end{proof}
We have proven the existence of a bijection between $\COddF(G)$
(resp.\@ $\CEvF(G)$) and $\COddF(G^{\set e})$ (resp.\@ $\CEvF(G^{\set
  e})$). To exhibit such a bijection, one would need to choose a
convention for the coloring of the vertices $v_{1}$ and $v_{2}$, see
figure \ref{fig:BijFig}, depending on the color of $v$ and on the fact
that $e$ belongs or not to $A$, as an edge or as a flag.

\paragraph{Properties of $\mathbf{\chi_G^{\set e}}$}
\label{sec:prop-chi_gs-e}
Here we precise the bijection $\chi_{G}^{\set e}$ of
lemma \ref{lem:ColCutSubg}. This will be useful in section \ref{sec:vari-limit-cases}.
  \begin{defn}[Partitions by flags]\label{def:partitions}
    For any ribbon graph $G$ with flags, the set $\OddF(G)$ (resp.\@ $\EvF(G)$) can be
    partitioned into subsets of cutting subgraphs labelled by their
    flag-set. For all $F'\subset F(G)\cup E(G)$, we
    write $\OddF(G)\flag F'$ (resp.\ $\EvF(G)\flag F'$) the set of
    all odd (resp.\@ even) spanning cutting
    subgraphs of $G$ with flag-set $F'$.    
  \end{defn}
  For any $F',F''\subset F(G)\cup E(G)$, we obviously have
  $(\OddF(G)\flag F')\cap(\OddF(G)\flag F'')=\emptyset$. Moreover,
  \begin{align}
    \OddF(G)=&\bigcup_{F'\subset F(G)\cup E(G)}\OddF(G)\flag F'.
 \end{align}
  These definitions of partitions and subsets of
  $\OddF(G)$ can be applied, mutatis mutandis, to $\COddF(G)$,
  $\EvF(G)$ and $\CEvF(G)$.
  \\
  Let $F'\subset\HR(G)$ and $F'_{e}$ be the set $F'\cup\set e$ if
  $e\notin F'$ and $F'\setminus\set e$ if $e\in F'$. Then, just by
  looking at the proof of lemma \ref{lem:ColCutSubg}, one sees that
  $\chi_{G}^{\set e}$ is a one-to-one map between $\COddF(G)\flag F'$
  (resp.\@ $\CEvF(G)\flag F'$) and $\COddF(G)\flag F'_{e}$ (resp.\@
  $\CEvF(G)\flag F'$).

\section{A new topological graph polynomial}
\label{sec:some-proofs-remarks}

Graph polynomials are graph invariants which encode part of the
information contained in the graph structure. These polynomials allow
an algebraic study of graphs, which is usually easier than a direct
approach.

Recently, B.~Bollob\'as and O.~Riordan \citep{Bollobas2002aa} defined
such a polynomial invariant for ribbon graphs. Here we introduce a
generalization of their polynomial, defined for ribbon graphs
with flags or external legs. It turns out that a certain evaluation of
this new topological graph invariant $\cQ$ enters the parametric
representation of the Feynman amplitudes of the Grosse-Wulkenhaar
model.\\

In the following, we will denote by bold letters, sets of
variables attached to edges or vertices of a graph. For example, given a (ribbon) graph $G$,
$\bx\defi\set{x_{e}}_{e\in E(G)}$. Moreover, for any $A\subset E(G)$,
we use the following short notation: $x^{A}\defi\prod_{e\in A}x_{e}$.

\begin{defn}[The $\cQ$ polynomial]\label{def:QDef}
  Let $G$ be a ribbon graph with flags. We define the following polynomial:
  \begin{align}
    \cQ_{G}(\bx,\by,\bz,\bw,\br)\defi&\sum_{A\subset E(G)}\sum_{B\subset E(G^{A})}x^{A^{c}\cap B^{c}}y^{A\cap B^{c}}z^{A\cap B}w^{A^{c}\cap B}r^{V(F_{B})},\label{eq:QDef}
\end{align}
 where we implicitely use the canonical bijection between $E(G)$ and $E(G^{A})$, and $r^{V(F_{B})}\defi\prod_{v\in V(F_{B})}r_{\degp v}$.
\end{defn}

\subsection{Basic properties}
\label{sec:defin-basic-prop}

\begin{prop}
${\cal Q}_{G}$ is multiplicative over disjoint unions and obeys the scaling relations
\begin{equation}
{\cal Q}_{G}(\lambda\bx,\lambda\by,\lambda \mu^{-2}\bz,\lambda\mu^{-2}\bw,\mu\cdot\br)=\lambda^{|E(G)|}\mu^{|F(G)|}{\cal Q}_G(\bx,\by,\bz,\bw,\br)
\end{equation}
where  $|E(G)|$ is the number of edges of $G$, $|F(G)|$ its number of
flags and $\mu\cdot\br$ is the sequence
$(\mu^{n}r_{n})_{n\in {\Bbb N}}$.
\end{prop}
The proof of this proposition is obvious.\\

In contrast with the Tutte or the \BRp, $\cQ$ satisfies a four-term
reduction relation. This relation generalizes the usual
contraction/deletion relation and reflects the two natural
operations (see definition \ref{def:operationsEdges}) one can make on
a ribbon graph with flags and on any of its partial dual.
\begin{lemma}[Reduction relation]\label{lem:RedRelD}
  Let $G$ be a ribbon graph with flags and $e$ any of its edges. Then,
  \begin{equation}
    \begin{split}
      \cQ_{G}(\bx,\by,\bz,\bw,\br)=&x_{e}\cQ_{G-e}(\bx^{
        e},\by^{ e},\bz^{ e},\bw^{
        e},\br)+y_{e}\cQ_{G^{ e}-e}(\bx^{ e},\by^{
        e},\bz^{ e},\bw^{ e},\br)\\
      &+z_{e}\cQ_{G^{ e}\vee
        e}(\bx^{ e},\by^{ e},\bz^{ e},\bw^{
        e},\br)+w_{e}\cQ_{G \vee e}(\bx^{ e},\by^{ e},\bz^{
        e},\bw^{ e},\br),
    \end{split}\label{eq:RedRelD}
  \end{equation}
  where, for any $a\in\set{\bx,\by,\bz,\bw}$ and any $e\in E(G)$,
  $a^{e}\defi \set{a_{e'}}_{e'\in E(G)\setminus\set e}$.
\end{lemma}
\begin{proof}
  Let $e$ be any edge of $G$. Refering to the definition (\ref{eq:QDef}) of $\cQ$, we distinguish between four cases, whether $e$ belongs to $A$ or not, to $B$ or not:
  \begin{equation}
    \begin{split}
      \cQ_{G}(\bx,\by,\bz,\bw,\br)=&x_{e}P_{1}(\bx^{
        e},\by^{ e},\bz^{ e},\bw^{
        e},\br)+y_{e}P_{2}(\bx^{ e},\by^{
        e},\bz^{ e},\bw^{ e},\br)\\
      &+z_{e}P_{3}(\bx^{ e},\by^{ e},\bz^{ e},\bw^{
        e},\br)+w_{e}P_{4}(\bx^{ e},\by^{ e},\bz^{
        e},\bw^{ e},\br).
    \end{split}
  \end{equation}
  The polynomial $P_{1}$ corresponds to the case $e\notin A$ and $e\notin B$. There is a canonical bijection $\varphi_{-}$ (resp.\@ $\varphi_{-}^{\star}$) between $E(G)\setminus\set e$ and $E(G-e)$ (resp.\@ between $E(G^{A})$ and $E((G-e)^{\varphi_{-}(A)})$). For any $A\subset E(G)$ and any $B\subset E(G^{A})$, we have
  \begin{subequations}
    \begin{align}
      A^{c}=&(\varphi_{-}(A))^{c}\cup\set e,&B^{c}=&(\varphi^\star_{-}(B))^{c}\cup\set e\\
      x^{A^{c}\cap B^{c}}=&x_{e}x^{(\varphi_{-}(A))^{c}\cap (\varphi^{\star}_{-}(B))^{c}}&y^{A\cap B^{c}}=&y^{\varphi_{-}(A)\cap (\varphi^{\star}_{-}(B))^{c}}\\
      z^{A\cap B}=&z^{\varphi_{-}(A)\cap\varphi^{\star}_{-}(B)}&w^{A^{c}\cap
        B}=&w^{(\varphi_{-}(A))^{c}\cap\varphi^{\star}_{-}(B)}.
    \end{align}
  \end{subequations}
Let us now check that $r^{V(F_{B})}=r^{V(F_{\varphi^{\star}_{-}(B)})}$. The left hand side of this equation encodes the degree sequence of $F_{B}\subset G^{A}$. But as $B$ does not contain $e$, $F_{B}$ can be considered as a subgraph of $G^{A}-e=(G-e)^{\varphi_{-}(A)}$ and $F_{B}$ is then isomorphic to $F_{\varphi_{-}^{\star}(B)}$. Their degree sequences are thus equal to each other. We have
\begin{align}
  x_{e}P_{1}(\bx^{e},\by^{ e},\bz^{ e},\bw^{e},\br)=&\sum_{A\subset E(G)\setminus\set e}\sum_{B\subset E(G^{A})\setminus\set e}x^{A^{c}\cap B^{c}}y^{A\cap B^{c}}z^{A\cap B}w^{A^{c}\cap B}r^{V(F_{B})}\\
  =&x_{e}\sum_{A\subset E(G)\setminus\set e}\sum_{B\subset E(G^{A})\setminus\set e}x^{(\varphi_{-}(A))^{c}\cap (\varphi^{\star}_{-}(B))^{c}}y^{\varphi_{-}(A)\cap (\varphi^{\star}_{-}(B))^{c}}\nonumber\\
  &\hspace{3.5cm}z^{\varphi_{-}(A)\cap\varphi^{\star}_{-}(B)}w^{(\varphi_{-}(A))^{c}\cap\varphi^{\star}_{-}(B)}r^{V(F_{\varphi^{\star}_{-}(B)})}\\
  =&x_{e}\sum_{A\subset E(G-e)}\sum_{B\subset E((G-e)^{A})}x^{A^{c}\cap B^{c}}y^{A\cap B^{c}}z^{A\cap B}w^{A^{c}\cap B}r^{V(F_{B})}\\
  =&x_{e}\cQ_{G-e}(\bx^{e},\by^{ e},\bz^{ e},\bw^{e},\br).
\end{align}

As we have seen, the difficulty only resides in the proof of the conservation of the $r$-part. Thus, for the three other cases, we only focus on that. The polynomial $P_{2}$ corresponds to the case $e\in A$ and $e\notin B$. Let $\varphi_{+}$ denote the canonical bijection between $\set{A\subset E(G)\tqs e\in A}$ and $E(G/e)$. As $B$ does not contain $e$, $F_{B}$ can also be considered as a subgraph of $G^{A}-e=(G^{e}-e)^{A\setminus\set e}=(G/e)^{\varphi_{+}(A)}$. This proves that $P_{2}=\cQ_{G^{e}-e}$.

The polynomial $P_{4}$ corresponds to the case $e\notin A$ and $e\in B$. As $e\notin A$, $G^{A}-e=(G-e)^{A}$ and the vertex sets $V(G^{A})$ and $V(G^{A}-e)$ are the same. But as $B$ contains $e$, erasing this edge would produce a different degree sequence for $F_{\varphi_{-}^{\star}(B)}$. So, we have to keep track of the contribution of $e$ to the degree sequence of $F_{B}$ by cutting it instead of deleting it: $P_{4}=\cQ_{G\vee e}$.

Finally the polynomial $P_{3}$ corresponds to the case $e\in A$ and $e\in B$. Such sets $A$ are in one-to-one correspondence with the subsets of $E(G^{e}-e)$. The vertex sets $V(G^{A})$ and $V((G^{e}-e)^{A\setminus\set e})$ are the same but once more, as $e\in B$, we can't delete $e$ but cut it instead: $P_{3}=\cQ_{G^{e}\vee e}$.
\end{proof}

Lemma \ref{lem:RedRelD} allows to give an alternative definition of the $\cQ$ polynomial:
\begin{defn}\label{def:RecursDefQ}
  Let $G$ be a ribbon graph with flags and $e$ any of its edges,
  \begin{equation}
    \begin{split}
      \cQ_{G}(\bx,\by,\bz,\bw,\br)=&x_{e}\cQ_{G-e}(\bx^{
        e},\by^{ e},\bz^{ e},\bw^{
        e},\br)+y_{e}\cQ_{G^{ e}-e}(\bx^{ e},\by^{
        e},\bz^{ e},\bw^{ e},\br)\\
      &+z_{e}\cQ_{G^{ e}\vee
        e}(\bx^{ e},\by^{ e},\bz^{ e},\bw^{
        e},\br)+w_{e}\cQ_{G \vee e}(\bx^{ e},\by^{ e},\bz^{
        e},\bw^{ e},\br).
    \end{split}\label{eq:RecDefQ}
  \end{equation}
  Otherwise $G$ consists of isolated vertices with flags and
  \begin{align}
    \cQ_{G}(\bx,\by,\bz,\bw,\br)=&\prod_{v\in V(G)}r_{\degp v}.\label{eq:FormTermQ}
  \end{align}
\end{defn}
It is remarkable that equations (\ref{eq:RecDefQ}) and
(\ref{eq:FormTermQ}) lead to a well defined polynomial in the sense
that it is independent of the order in which the edges are chosen. The
proof of the existence of such a polynomial consists essentially in the proof of lemma
\ref{lem:RedRelD}. The polynomial which results of this recursive
process is the $\cQ$ polynomial of definition \ref{def:QDef}. The
uniqueness of the result is obvious since if $e\in E(G)$ then
$\cQ_{G}$ is uniquely determined by $\cQ_{G-e}$, $\cQ_{G^{e}-e}$,
$\cQ_{G\vee e}$ and $\cQ_{G^{e}\vee e}$ \citep{Bollobas1998aa}.

\subsection{Relationship with other polynomials}
\label{sec:relat-with-other}

\begin{itemize}
\item \underline{The \BRp}: if we set $z=w=0$ and $x=1$ for all edges and $r_{n}=r$ (independent of $n$), we recover the multivariate Bollob\'as-Riordan polynomial at $q=1$, in its multivariate formulation (see \citep{Moffatt2008ab})
\begin{equation}
{\cal Q}_{\Gamma}(1,y,0,0,r)=\sum_{A\subset E(G)}\Big(\prod_{e\in A} y_{e}\Big)\,r^{v(G^{A})}
\end{equation}
where $v(G^{A})$ is the number of vertices of $G^{A}$ ie the number of
connected components of the boundary of $F_{A}$. Note that the
evaluation $y=w=0$, $x=1$ and $r_{n}=r$ gives the same result.
\item \underline{The dimer model}: if we set $y=z=0$, $x=1$ and $r_{n}=0$ except $r_{1}=1$, then we recover, for a graph without flags,  the partition function of the dimer model on this graph
\begin{equation}
{\cal Q}_{\Gamma}(1,0,0,w,r)=\sum_{C\subset E(G)\atop\mbox{\tiny dimer configuration}}\Big(\prod_{e\in C} w_{e}\Big),
\end{equation}
with $w_{e}=\mathrm{e}^{\beta \varepsilon_{e}}$ the Boltzmann weight. Here, each vertex contains a monomer that can form a dimer with an adjacent monomer, if the edge $e$ supports a dimer then its energy is $-\varepsilon_{e}$. A dimer configuration (also known as a perfect matching in graph theory) is obtained when each monomer belongs to exactly one dimer. In the recent years, the dimer model has proven to be of great mathematical interest (see \cite{Kenyon2009aa} for a recent review).
\item \underline{The Ising model}: for $y=z=0$, $x_{e}=\cosh(\beta J_{e})$, $w_{e}=\sinh(\beta J_{e})$, $r_{2n}=2$ and $r_{2n+1}=0$, we recover the partition function of the Ising model. Recall that the latter is obtained by assigning spins $\sigma_{v}\in\left\{-1,+1\right\}$ to each vertex with an interaction along the edges encoded by the Hamiltonian
\begin{equation} 
H(\sigma)=-\sum_{e=(v,v')\in E}J_{e}\sigma_{v}\sigma_{v'},
\end{equation}
with $J_{e}$ an edge dependent coupling constant. The partition
function is the sum over all spins configurations of the Boltzmann weight
\begin{equation}
Z_{\mbox{\tiny Ising}}=\sum_{\sigma}\mathrm{e}^{-\beta H}.
\end{equation}
Using the identity
\begin{equation}
\mathrm{e}^{\beta J_{e}\sigma_{v}\sigma_{v'}}=\cosh(\beta J_{e})+\sigma_{v}\sigma_{v'}\sinh(\beta J_{e})
\end{equation}
for each edge, we can perform the high temperature expansion of  the partition function 
\begin{equation}
Z_{\mbox{\tiny Ising}}=\sum_{\sigma}\left\{\sum_{C\subset E}\big(\prod_{e\notin C}\cosh(\beta J_{e})\big)
\big(\prod_{e\in C}\sinh(\beta J_{e})\sigma_{v}\sigma_{v'}\big)\right\}
\end{equation}
Then, the sum over all spins vanishes unless each vertex is matched by a even number of edges in $C$, so that
\begin{equation}
Z_{\mbox{\tiny Ising}}={\cal Q}_{G}(x,0,0,w,r)
\end{equation}
with the specified value of $x$, $w$ and $r$. Note that the extra power of 2 arising from the sum over spins corresponds to $r_{2n}=2$. 
 \end{itemize}

\subsection{Partial duality of $\cQ$}
\label{sec:partial-duality-cq}

One of the most interesting properties of the $\cQ$ polynomial is that
it transforms nicely under partial duality.
\begin{thm}[Partial duality]\label{thm:PartialDualityQ}
  Let $G$ be a ribbon graph with flags and $e\in E(G)$ be any edge of $G$. We have
  \begin{align}
    \cQ_{G^{\set{e}}}(\bx,\by,\bz,\bw,\br)=\cQ_{G}(\bx_{E\setminus\set{e}}\by_{\set e},\bx_{\set e}\by_{E\setminus\set{e}},\bz_{E\setminus\set{e}}\bw_{\set e},\bz_{\set e}\bw_{E\setminus\set{e}},\br).\label{eq:PartDualQ}
  \end{align}
\end{thm}
\begin{proof}
  Each monomial of $\cQ$ is labelled by two sets of edges $A\subset E(G)$ and $B\subset E(G^{A})$:
  \begin{align}
    \cQ_{G}(\bx,\by,\bz,\bw,\br)\fide&\sum_{A,B}M_{(A,B)}(G;\bx,\by,\bz,\bw,\br)\\
    \fide&\sum_{A,B}\cM_{(A,B)}(G;\bx,\by)\cN_{(A,B)}(G;\bz,\bw)\br^{V(F_{B})}.
  \end{align}
  For any $e\in
  E(G)$, let $\phi_{e}$ be the following map:
  \begin{align}
    \phi_{e}: \bigcup_{A\subset E(G)}A\times E(G^{A})&\to\bigcup_{A'\subset E(G^{\set{e}})}A'\times E\big((G^{\set{e}})^{A'}\big)\label{eq:BijMonom}\\
    (A,B)&\mapsto (A\Delta\set{e},B).\nonumber
  \end{align}
  $\phi_{e}$ is clearly a bijection for any edge $e$. Note that $\big(G^{\set{e}}\big)^{A\Delta\set{e}}=G^{A}$ which
  implies (with a slight abuse of notation) that, for any $F_{B}\subset G^{A}$, $\br^{V(F_{B})}=\br^{V(F_{\phi_{e}(B)})}$ where $F_{\phi_{e}(B)}\subset (G^{\set e})^{\phi_{e}(A)}$. Let $\bx^{\star\set e}$ be $\bx_{E\setminus\set{e}}\by_{\set{e}}$, $\by^{\star\set e}$ be
  $\bx_{\set{e}}\by_{E\setminus\set{e}}$, $\bz^{\star\set e}$ be $\bz_{E\setminus\set{e}}\bw_{\set{e}}$ and $\bw^{\star\set e}$ be $\bz_{\set{e}}\bw_{E\setminus\set{e}}$. To prove the
  theorem, we prove that
  $M_{\phi_{e}((A,B))}(G^{\set{e}};\bx,\by,\bz,\bw,\br)=M_{(A,B)}(G;\bx^{\star\set e},\by^{\star\set e},\bz^{\star\set e},\bw^{\star\set e},\br)$.
  \begin{align}
    A'\defi A\Delta\set{e}=&
    \begin{cases}
      A\cup\set{e}&\text{if $e\notin A$,}\\
      A\setminus\set{e}&\text{if $e\in A$}.
    \end{cases}\\
    A'^{c}=A^{c}\Delta\set{e}=&
    \begin{cases}
      A^{c}\setminus\set{e}&\text{if $e\notin A$,}\\
      A^{c}\cup\set{e}&\text{if $e\in A$}.
    \end{cases}
    \intertext{If $e\in B$, $B^{c}\cap A'=B^{c}\cap A$ and $B^{c}\cap
      A'^{c}=B^{c}\cap A^{c}$. Thus, in this case, $\cM_{\phi_{e}((A,B))}(G^{\set e};\bx,\by)$ is obviously equal to $\cM_{(A,B)}(G;\bx^{\star\set e},\by^{\star\set e})$. So let us focus on the terms involving $\bz$ and $\bw$.}
    B\cap A'=&
    \begin{cases}
      (B\cap A)\cup\set{e}&\text{if $e\in B\cap A^{c}$,}\\
      (B\cap A)\setminus\set{e}&\text{if $e\in B\cap A$.}
    \end{cases}\\
    B\cap A'^{c}=&
    \begin{cases}
      (B\cap A^{c})\setminus\set{e}&\text{if $e\in B\cap A^{c}$,}\\
      (B\cap A^{c})\cup\set{e}&\text{if $e\in B\cap A$.}
    \end{cases}
    \intertext{Then we have}
    \cN_{(A',B)}(G^{\set{e}};\bz,\bw,\br)=&z^{B\cap A'}w^{B\cap A'^{c}}\nonumber\\
    =&
    \begin{cases}
      z^{(B\cap A)\cup\set{e}}w^{(B\cap A^{c})\setminus\set{e}}&\text{if $e\in B\cap A^{c}$,}\\
      z^{(B\cap A)\setminus\set{e}}w^{(B\cap A^{c})\cup\set{e}}&\text{if $e\in B\cap A$}
    \end{cases}\\
    =&\cN_{(A,B)}(G;\bz^{\star\set e},\bw^{\star\set e},\br).\nonumber
\end{align}
If $e\notin B$, we use exactly the same argument with $\cN$ replaced by $\cM$, $z$ by $y$ and $w$ by $x$.
\end{proof}
\begin{cor}\label{cor:PartialDuality}
  For any ribbon graph $G$ with flags and any subset $E'\subset E(G)$,
  we have
  \begin{align}
    \cQ_{G^{E'}}(\bx,\by,\bz,\bw,\br)=&\cQ_{G}(\bx_{E\setminus
      E'}\by_{E'},\bx_{E'}\by_{E\setminus E'},\bz_{E\setminus
      E'}\bw_{E'},\bz_{E'}\bw_{E\setminus E'},\br)\label{eq:PartialDualityQCor}
  \end{align}
\end{cor}
\begin{proof}
  It relies on:
  \begin{enumerate}
  \item for any $e\in E'$, $G^{E'}=\big(G^{E'\setminus\set{e}}\big)^{\set{e}}$,
  \item a repeated use of theorem \ref{thm:PartialDualityQ}.
  \end{enumerate}
\end{proof}
\section{Feynman amplitudes of the Grosse-Wulkenhaar model}
\label{sec:feynm-ampl-grosse}

\subsection{The action functional}

The Grosse-Wulkenhaar model is defined by the action functional
\begin{equation}
S[\phi]=S_{0}[\phi]+S_{\text{int}}[\phi],
\end{equation}
where $\phi$ is a real valued function on Euclidean space ${\Bbb R}^{D}$. The free part of the action is
\begin{equation}
S_{0}[\phi]=\frac{1}{2}\int d^{D}x\,\phi(x)\big(-\Delta+\widetilde{\Omega}^{2}x^{2}\big)\phi(x),
\end{equation}
where $\Delta$ is the Laplacian on Euclidean space ${\Bbb R}^{D}$ and $\widetilde{\Omega}=\frac{2\Omega}{\theta}$ (with $\Omega,\theta>0$) the frequency of the corresponding harmonic oscillator. In a system of units such that $\hbar=c=1$, the only remaining dimension is length and $\Omega$ is dimensionless. 

Its kernel ${\cal K}_{\widetilde{\Omega}}(x,y)$ defined by
\begin{equation}
\int d^{D}z\,\delta^{D}(x-z) \big(-\Delta_{z}+\widetilde{\Omega}^{2}z^{2}\big)\,{\cal K}_{\widetilde{\Omega}}(z,y)=\delta^{D}(x-y),
\end{equation}
with $\delta^{D}$ the Dirac distribution on ${\Bbb R}^{D}$, is the Mehler kernel 
\begin{equation}
{\cal K}_{\widetilde{\Omega}}(x,y)=\left(\frac{\widetilde{\Omega}}{2\pi}\right)^{D/2}\int_{0}^{\infty}\frac{d\alpha}{\big[\sinh2\widetilde{\Omega}\alpha\big]^{D/2}}
\exp-\frac{\widetilde{\Omega}}{4}\left\{(x-y)^{2}\coth\widetilde{\Omega}\alpha+(x+y)^{2}\tanh\widetilde{\Omega}\alpha \right\}.\end{equation}
To avoid ultraviolet divergences, we introduce a cut-off as a lower bound on the integral over $\alpha$,
\begin{multline}
{\cal K}_{\widetilde{\Omega}}(x,y)\rightarrow\cr\left(\frac{\widetilde{\Omega}}{2\pi}\right)^{D/2}\int_{1/\Lambda^{2}}^{\infty}\frac{d\alpha}
{\big[\sinh2\widetilde{\Omega}\alpha\big]^{D/2}}
\exp-\frac{\widetilde{\Omega}}{4}\left\{(x-y)^{2}\coth\widetilde{\Omega}\alpha+(x+y)^{2}\tanh\widetilde{\Omega}\alpha \right\}.\label{Mehler}
\end{multline}
Since this paper is not concerned with the limit $\Lambda\rightarrow \infty$, we will always self-understand that the integration over $\alpha$ ranges over $\left[\frac{1}{\Lambda^{2}},\infty\right]$. Later on, it will also prove convenient to introduce $t=\tanh(\widetilde{\Omega}\alpha)$ as well as the short and long variables
\begin{equation}
u=\frac{1}{\sqrt{2}}(x-y)\quad\mathrm{and}\quad v=\frac{1}{\sqrt{2}}(x+y),
\end{equation}
so that the propagator reads
\begin{equation}
{\cal K}_{\widetilde{\Omega}}(x,y)=\left(\frac{\widetilde{\Omega}}{2\pi}\right)^{D/2}\int_{1/\Lambda^{2}}^{\infty}d\alpha\left[\frac{(1-t^{2})}{2t}\right]^{D/2}
\exp-\frac{1}{2}\left\{\frac{\widetilde{\Omega}}{t}\,u^{2}+\widetilde{\Omega} t\,v^{2} \right\}.\label{Mehler2}
\end{equation}

The interaction term is derived form the Moyal product 
\begin{equation}
f\star g\,(x)=\frac{1}{\pi^{D}|\det\Theta|}\int\, d^{D}y\,d^{D}z\,f(x+y)f(x+z)\mathrm{e}^{-2\imath y\Theta^{-1}z},
\end{equation}
with $\Theta$ a real, non degenerate,  antisymmetric $D\times D$ matrix, with $D$ even.  In the sequel, we assume\footnote{Otherwise the amplitude cannot be written as \eqref{amplitude} and the hyperbolic polynomial are not defined.} that $\Theta=\theta J$, with $\theta>0$ and $J$ the antisymmetric $D\times D$ block diagonal matrix made of $2\times 2$ blocks $\begin{pmatrix}0&1\cr -1&0\end{pmatrix}$.  We define the interaction term  as
\begin{equation}
S_{\text{int}}[\phi]=
\sum_{n\geq1}\frac{g_{n}}{n}
\int d^{D}x\, \phi^{{\star}n}(x),\label{intstar}
\end{equation}
where $g_{n}\in {\Bbb R}$ are coupling constants. In the sequel, it will be necessary to express explicitely the Moyal interaction as a functional of the fields 
\begin{equation}
S_{\text{int}}[\phi]=
\sum_{n\geq1}\frac{g_{n}}{n}
\int d^{D}x_{1}\cdots d^{D}x_{n}{\cal V}_{n}(x_{1},\dots,x_{n})\phi(x_{1})\cdots\phi(x_{n}).
\end{equation}
${\cal V}_{n}(x_{1},\dots,x_{n})$ is a distibution on $({\Bbb R}^{D})^{n}$, invariant under cyclic permutations,
\begin{equation}
{\cal V}_{n}(x_{1},x_{2},\dots,x_{n})={\cal V}_{n}(x_{n},x_{1},\dots,x_{n-1}).\label{cyclic}
\end{equation}
In the commutative limit $\theta\rightarrow0$, it reduces to a product of Dirac distributions 
\begin{equation}
\lim_{\theta\rightarrow 0}{\cal V}_{n}(x_{1},\dots,x_{n})=\prod_{j\atop
j\neq i}\delta(x_{i}-x_{j})\label{vertexcomm},
\end{equation}
which is invariant under all permutations of $\left\{1,2,\dots,n\right\}$.

Turning back to the noncommutative case  $\theta\neq 0$,  in lower degree we have ${\cal V}_{1}(x_{1})=1$ and ${\cal V}_{2}(x_{1},x_{2})=\delta^{D}(x_{1}-x_{2})$. The first interesting interaction is ${\cal V}_{3}$
\begin{equation}
{\cal V}_{3}(x_{1},x_{2},x_{3})=
\frac{1}{{(\pi\theta)^{D}}}
\exp-\frac{2\imath}{\theta}\Big\{x_{1}\cdot J x_{2}+x_{2}\cdot J x_{3}+x_{3}\cdot J x_{1}\Big\}.
\end{equation}
The last expression of ${\cal V}_{3}$ is very convenient since we can associate to it a triangle with vertices $x_{1}$, $x_{2}$ and $x_{3}$ drawn in cyclic order around its boundary, oriented counterclockwise. In the sequel, it will be convenient to express higher order vertices using triangles glued together in a tree-like manner.

\begin{prop}
\label{planetree}
Let $T$ be a plane tree (i.e. a connected acyclic graph embedded in the plane) with all its inner vertices of degree $3$ and its edges labelled using the index set $I$ and  let $i_{1},\dots,i_{n}$ be the cyclically ordered labels of some of the edges attached  to the leaves (terminal vertices), in counterclockwise order around the tree. Then,
\begin{equation}
{\cal V}_{n}(x_{i_{1}},\dots,x_{i_{n}})=\int\prod_{i\in I-\left\{i_{1},\dots,i_{n}\right\}}d^{D}x_{i}\hskip-0.2cm\prod_{v\atop
\mathrm{ vertices\, of\, }T}\hskip-0.2cm
\frac{\exp-\frac{2\imath}{\theta}\Big\{x_{i_{v}}\cdot J x_{j_{v}}+x_{j_{v}}\cdot J x_{k_{v}}+x_{k_{v}}\cdot J x_{i_{v}}\Big\}}{(\pi\theta)^{D}},
\label{vertex}
\end{equation}
with $i_{v},j_{v},k_{v}$ the labels of the cyclically ordered edges incident to $v$.
\end{prop}
\begin{figure}[!htp]
  \centering
  \includegraphics[scale=.6]{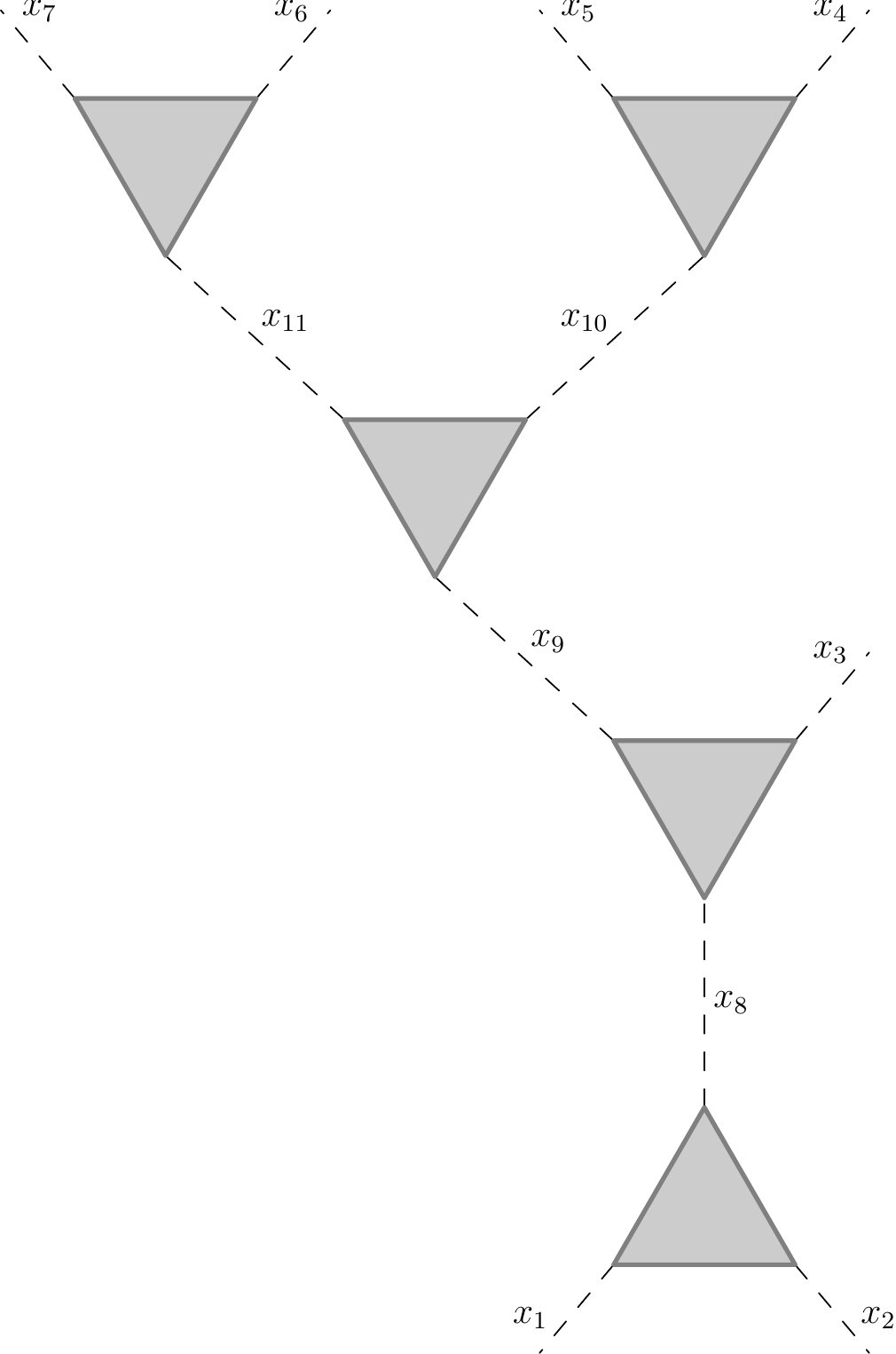}
  \caption{A heptagonal Moyal vertex}
  \label{fig:TreeVertex}
\end{figure}
\begin{proof}
  Let us prove this result by induction on the number of inner
  vertices of $T$. If $T$ has a single inner vertex, then the equality
  for $n=3$ is trivial whereas for $n=1,2$ it results from the
  identity
  \begin{equation}
    \frac{1}{(\pi\theta)^{D}}\int d^{D}y\exp-{\textstyle \frac{2\imath}{\theta}}\left\{y\cdot J z\right\}=\delta^{D}(z).\label{delta}
  \end{equation}
  Next, we suppose the result valid for all trees of order less than
  $m$ and consider a tree $T$ of order $m+1$. Cut an inner edge in $T$
  with label $i_{0}$, which splits $T$ into $T'$ and $T''$.  Without
  loss of generality, let us assume that $i_{1},\dots,i_{n'},i_{0}$
  are the labels of the leaves of $T'$. Then, we separate the vertices
  of $T$ into vertices of $T'$ and $T''$ and use the induction
  assumption for $T'$ and $T''$,
  \begin{align}
    \int\prod_{i\in
      I-\left\{i_{1},\dots,i_{n}\right\}}d^{D}x_{i}\hskip-0.2cm\prod_{v\atop
      \mathrm{ vertices\, of\, }T}\hskip-0.2cm
    \frac{\exp-\frac{2\imath}{\theta}\Big\{x_{i_{v}}\cdot J
      x_{j_{v}}+x_{j_{v}}\cdot J x_{k_{v}}+x_{k_{v}}\cdot J
      x_{i_{v}}\Big\}}{(\pi\theta)^{D}} =\cr \int
    d^{D}x_{i_{0}}\,{\cal
      V}_{n'+1}(x_{i_{1}},\dots,x_{i_{n'}},x_{i_{0}})\,{\cal
      V}_{n-n'+1}(x_{i_{0}},x_{i_{n'+1}},\dots,x_{i_{n}}).
  \end{align}

  To conclude, we need the following lemma.
\begin{lemma}
  The vertices of Moyal interaction obey
  \begin{equation}
    \int d^{D}y\,{\cal V}_{n'+1}(x_{i_{1}},\dots,x_{i_{n'}},y)\,{\cal V}_{n-n'+1}(y,x_{i_{n'+1}},\dots,x_{i_{n}})=
    {\cal V}_{n}(x_{i_{1}},\dots,x_{i_{n}})
  \end{equation}
  for any integer $1\leq n'\leq n-1$.
\end{lemma}
\begin{proof}[of the lemma]
  \begin{gather}
    \int d^{D}y\,{\cal V}_{n'+1}(x_{i_{1}},\dots,x_{i_{n'}},y)\,{\cal
      V}_{n-n'+1}(y,x_{i_{n'+1}},\dots,x_{i_{n}})=\cr \int
    d^{D}y\,d^{D}y\,'\,\delta^{D}(y-y')\,{\cal
      V}_{n'+1}(x_{i_{1}},\dots,x_{i_{n'}},y)\,{\cal
      V}_{n-n'+1}(y',x_{i_{n'+1}},\dots,x_{i_{n}})=\cr \int
    \frac{d^{D}k}{(2\pi)^{D}}d^{D}y\,d^{D}y\,'\,\,{\cal
      V}_{n'+1}(x_{i_{1}},\dots,x_{i_{n'}},y)e_{k}(y)\;{\cal
      V}_{n-n'+1}(y',x_{i_{n'+1}},\dots,x_{i_{n}})e_{-k}(y')
  \end{gather}
  with $e_{k}(x)=\exp\imath kx$. Smearing out with functions
  $f_{i_{1}},\dots,f_{i_{n}}$, we thus have
  \begin{gather}
    \int d^{D}y\,d^{D}x_{i_{1}}\cdots d^{D}x_{i_{n}}\, {\cal
      V}_{n'+1}(x_{i_{1}},\dots,x_{i_{n'}},y)\,{\cal
      V}_{n-n'+1}(y,x_{i_{n'+1}},\dots,x_{i_{n}})
    f_{i_{1}}(x_{i_{1}})\dots f_{i_{n}}(x_{i_{n}}) =\cr \int
    \frac{d^{D}k}{(2\pi)^{D}}\, \int d^{D}x f_{i_{1}}\star\cdots \star
    f_{i_{n'}}\star e_{k}(x) \, \int d^{D}x' f_{i_{n'+1}}\star\cdots
    \star f_{i_{n}}\star e_{-k}(x')=\cr \int
    \frac{d^{D}k}{(2\pi)^{D}}\, \int d^{D}x f_{i_{1}}\star\cdots \star
    f_{i_{n'}}(x) e_{k}(x) \, \int d^{D}x' f_{i_{n'+1}}\star\cdots
    \star f_{i_{n}}(x') e_{-k}(x')=\cr \int d^{D}x
    \,f_{i_{1}}\star\cdots \star
    f_{i_{n'}}(x)\;f_{i_{n'+1}}\star\cdots \star f_{i_{n}}(x)=\cr \int
    d^{D}x \,f_{i_{1}}\star\cdots \star f_{i_{n'}} \star
    f_{i_{n'+1}}\star\cdots \star f_{i_{n}}(x)\cr \int
    d^{D}x_{i_{1}}\cdots d^{D}x_{i_{n}}\,{\cal
      V}_{n}(x_{i_{1}},\dots,x_{i_{n}}) \,f_{i_{1}}(x_{i_{1}}) \cdots
    f_{i_{n}}(x_{i_{n}})
  \end{gather}
  where we have repeatedly used
  \begin{equation}
    \int d^{D}x \,f\star g(x)=\int d^{D}x f(x)g(x).
  \end{equation}
\end{proof}
The lemma ends the proof of \eqref{vertex}.
\end{proof}

In what follows, we always assume that such a tree has been chosen for
every vertex, all choices leading to the same distribution ${\cal
  V}_{n}$. Moreover, since ${\cal V}_{3}$ is conveniently represented
as a triangle, we represent the contribution of each vertex of $T$ as
a triangle whose vertices are called corners, see figure
\ref{fig:TreeVertex} for an example.

\subsection{Parametric representation and the hyperbolic polynomials}

Formal perturbative quantum field theory can be compactly formulated within the background field method. In this approach, the  main object is the background field effective action  defined by the expansion over Feynman graphs (we normalize the path integral in such a way that it takes the value $1$ when all the coupling constants vanish)
\begin{align}
-\log\int[{\cal D}\chi]\,\exp-\left\{S_{0}[\chi]+ S_{\text{int}}[\phi+\chi]\right\}
=\hskip3cm{}\cr 
-\sum_{G\,\mathrm{connected}\,\mathrm{ribbon}\,\mathrm{graph}\atop\mathrm{with}\,f(G)\,\mathrm{flags}}
{\textstyle \frac{(-g)^{v(G)}}{S_{G}f(G)!}}
\,\int \prod_{1\leq i\leq f(G)}d^{D}x_{i}\quad{\cal A}_{G}(x_{1},\dots,x_{f(G)})\prod_{1\leq i\leq f(G)}\phi(x_{i}).\label{background}
\end{align}
Since the interaction vertices are invariant under cyclic permutations (see \eqref{cyclic}), the sum runs over all ribbon graphs. The graph also have $f(G)$ flags, which are half-lines that carry the labels of the field insertions $\phi(x_{1})\cdots\phi(x_{f(G)})$. $S_{G}$ is the symmetry factor of the graph (cardinality of the automorphism group of the graph, leaving the flags fixed), $(-g)^{V(G)}=\prod_{v\in V(G)}(-g_{d_{v}})$, with $d_{v}$ the degree of $v$ and ${\cal A}_{G}$ is the amplitude, to be defined below. 

In the sequel, it will prove convenient to  allow edge dependent oscillator frequencies $\Omega_{e}$, so that we recover the amplitude appearing in \eqref{background} by setting $\Omega_{e}=\Omega$.

\begin{defn}
Let $G$ be a ribbon graph with flags and let us attach a variable  $x_{i}\in{\Bbb R}^{D}$ to each flag of $G$ and $\Omega_{e}>0$ to each edge.  The (generalized) amplitude of a ribbon graph with flags is the distribution defined as
\begin{equation} 
{\cal A}_{G}\big[\Omega,x\big]=\int \prod_{i\notin F(G)} d^{D}y_{i}
\prod_{e\in E(G)}{\cal K}_{\frac{2\Omega_{e}}{\theta}}(y_{i_{e,+}},y_{i_{e,-}})\prod_{v\in V(G)}{\cal V}_{d_{v}}(y_{i_{v,1}},\dots,y_{i_{v.d_{v}}}),\label{defamplitude}
\end{equation}
where we integrate over variables $y_{i}\in{\Bbb R}^{D}$ associated to each half-edge of $G$, with the convention that for a flag we set $y_{i}=x_{f}$ without integrating over $x_{i}$. $y_{i_{e,+}},y_{i_{e,-}}$ the variables attached to the ends of $e$ (the order does not matter since the Mehler kernel is symmetric) and $y_{i_{v,1}},\dots,y_{i_{v.d_{v}}}$ the variables attached in cyclic order around vertex $v$.
\end{defn}

In the commutative case $\theta=0$, the vertex \eqref{vertexcomm} enforces the identification of all the corners (internal and external) attached to the same vertex and is invariant under all permutations of the half-edges incoming to a vertex. Therefore, the amplitude is assigned to ordinary (i.e. non ribbon) graphs, with flags replaced by external vertices.

\begin{defn}\label{defamplitudecommutative}
Let $\underline{G}=(V,V_{\mbox{\tiny ext}},E)$ be a graph  with $V_{\mbox{\tiny ext}}\subset V$ the external vertices to which variables $x_{v}\in{\Bbb R}^{D}$ are assigned.  Let us attach a variable $y_{v}\in{\Bbb R}^{D}$ to each vertex of $\underline{G}$, with the convention that  $y_{v}=x_{v}$ for an external vertex. The (generalized) commutative amplitude of a graph with external vertices is defined as
\begin{align} 
{\cal A}_{\underline{G}}^{\mathrm{ commutative}}\big[\Omega_{e},x_{v}\big]=\int \prod_{v\in V-V_{\mbox{\tiny ext}}} d^{D}y_{v}
\prod_{e\in E}{\cal K}_{\Omega_{e}}(y_{v_{e,+}},y_{v_{e,-}}),
\label{DefGenAmplitude}
\end{align}
with $\Omega_{e}$ the edge dependent frequency and $v_{e,+}$ and $v_{e,-}$ the vertices $e$ is attached.
\end{defn}
The commutative amplitude is recovered as a limiting case.
\begin{prop}\label{commlimrop}
Let $G$ be a ribbon graph with flags and let $V_{\mbox{\tiny ext}}(G)$ be the subset of vertices of $G$ carrying flags. Then, for the graph with external vertices $\underline{G}=(V(G),V_{\mbox{\tiny ext}}(G),E(G))$
\begin{align}
{\cal A}_{\underline{G}}^{\mathrm{ commutative}}\,\big[\Omega_{e}, x_{v}\big]\,\prod_{v\in V^{\mbox{\tiny ext}}(G)}\bigg\{\prod_{f\in F_{v}}\delta(x_{v}-x_{f})\bigg\}
=\lim_{\theta\rightarrow 0}{\cal A}_{G}\big[{\textstyle\frac{\theta}{2}}\Omega_{e},x_{f}\big],
\end{align} 
with $\underline{G}=(E(G),V(G), V_{\mbox{\tiny ext}}(G))$ and $F_{v}(G)$ the set of flags attached to $v$ in $G$ and $(x_{v})_{\in V_{\mbox{\tiny ext}}(G)}$ and $(x_{f})_{f \in F(G)}$ independent variables.
\end{prop}
\begin{proof}
  Note that in the commutative case we use oscillators of frequency
  $\Omega_{e}$ instead of $\frac{2\Omega_{e}}{\theta}$. Then,
  proposition \ref{commlimrop} follows immediately from
  \eqref{vertexcomm}.
\end{proof}

Even in the general noncommutative case, the integral over all the corners is Gau\ss ian, thanks to peculiar form of the Mehler kernel \eqref{Mehler} and of the Moyal vertex \eqref{vertex}. Therefore, the amplitude can be expressed in parametric form as follows, as was first shown in \cite{gurauhypersyman} .

\begin{theorem/definition}
The generalized amplitude \eqref{amplitude} of the Grosse-Wulkenhaar
model for a ribbon graph $G$ (which does not contain an isolated vertex with an even number of flags) with $e(G)$ edges, $v(G)$ vertices and $f(G)$ flags carrying variables $x_{i}\in{\Bbb R}^{D}$  is
\begin{equation}
{\cal A}_{G}(x)=\int {\textstyle \prod_{e}\!d\alpha_{e}}\, \left[\frac{2^{f(G)}\prod_{e}\Omega_{e}(1-t_{e}^{2})}
{(2\pi\theta)^{e(G)+f(G)-v(G)}\mathrm{HU}_{G}(\Omega,t)}
\right]^{D/2}
\exp-\frac{1}{\theta}\left\{\frac{\mathrm{HV}_{G}(\Omega,t,x)}{\mathrm{HU}_{G}(\Omega,t)}\right\},\label{amplitude}
\end{equation}
where the first hyperbolic polynomial $\mathrm{HU}_{G}(\Omega,t)$ is a polynomial in the edge variables $\Omega_{e}$ and $t_{e}=\tanh\frac{2\Omega_{e}\alpha_{e}}{\theta}$ and the second  hyperbolic polynomial $\mathrm{HV}_{G}(\Omega,t,x)$ is a linear combination of the products $x_{i}\cdot x_{j}$ and $x_{i}\cdot Jx_{j}$, whose coefficients are polynomials in $\Omega_{e}$ and $t_{e}$. 
\label{defHUHV}
\end{theorem/definition}
\begin{proof}
  The key idea is to write the amplitude \eqref{DefGenAmplitude} as a
  Gau\ss ian integral. To begin with, let us first derive a more
  systematic expression of ${\cal A}_{G}$. First, we represent each
  vertex using a plane tree made of triangles, as in proposition
  \ref{planetree}. The corners of the triangles attached to the flags
  of $G$ are the external corners while the other corners over which
  we integrate are called internal corners.  The internal corners come
  in three types: related by an edge, common to two triangles or
  isolated. In this last case, the variable attached to the internal
  corner acts as a Lagrange multiplier, as in \eqref{delta}. Since all the
  triangles are oriented counterclockwise, we define an antisymmetric
  adjacency matrix $\zeta$ between the corners (internal and external)
  by
  \begin{equation}
    \left\{
      \begin{array}{rcrl}
        \zeta_{ij}&=&1&\mbox{if there is a triangle edge oriented from $i$ to $j$},\\
        \zeta_{ij}&=&-1&\mbox{if there is a triangle edge oriented from $j$ to $i$},\\
        \zeta_{ij}&=&0&\mbox{if there is no triangle edge between $i$ and $j$}.\\
      \end{array}
    \right.
  \end{equation}
  Let us denote by $C^{\mbox{\tiny int}}_{v}$ (resp. $C^{\mbox{\tiny
      ext}}_{v}$ ) the set of internal (resp. external) corners
  attached to the vertex $v$ and define the matrix $\alpha$
  (resp. $\beta$, $\gamma$) by restricting $\zeta$ to the lines and
  columns in $C^{\mbox{\tiny int}}_{v}$ (resp. lines in
  $C^{\mbox{\tiny int}}_{v}$ and columns in $C^{\mbox{\tiny
      ext}}_{v}$, lines and columns in $C^{\mbox{\tiny
      ext}}_{v}$). Using \eqref{vertex}, the contribution of the
  vertex $v$ to ${\cal A}_{G}$ can be written as
  \begin{equation}
    \frac{1}{(\pi\theta)^{D|T_{v}|}}
    \exp-\frac{\imath}{\theta}\Big\{
    \sum_{i,j\in C^{\mbox{\tiny int}}_{v}}\alpha_{ij}\,x_{i}\cdot Jx_{j}+
    2\sum_{i\in C^{\mbox{\tiny int}}_{v},\,j\in C^{\mbox{\tiny ext}}_{v}}\beta_{ij}\,y_{i}\cdot Jx_{j}+
    \sum_{i,j\in C^{\mbox{\tiny ext}}_{v}}\gamma_{ij}\,x_{i}\cdot Jx_{j}\Big\},
  \end{equation} 
  with $|T_{v}|$ the number of triangles used in the chosen tree-like
  representation of $v$.

  In order to define the short and long variables for all edges, we
  choose an arbitrary orientation on the edges of $G$ and introduce
  the incidence matrix $\epsilon$ between the edges and the internal
  corners
  \begin{equation}
    \left\{
      \begin{array}{rcrl}
        \epsilon_{ei}&=&1&\mbox{if $e$ arrives at $i$},\\
        \epsilon_{ei}&=&-1&\mbox{if $e$ leaves $i$},\\
        \epsilon_{ei}&=&0&\mbox{if $e$ is not attached to $i$}.\\
      \end{array}
    \right.
  \end{equation} 
  The long and short variables associated with the edge $e$ are
  \begin{equation}
    u_{e}={\textstyle\frac{1}{\sqrt{2}}} \big(\sum_{i}\epsilon_{ei}y_{i}\big)\quad\mathrm{and}
    \quad v_{e}={\textstyle \frac{1}{\sqrt{2}}}\big(\sum_{i}|\epsilon_{ei}|y_{i}\big),
  \end{equation}
  with $x_{i}$ the variables attached to the corners. We enforce these
  relations by inserting $\delta$-functions with Lagrange multipliers
  $\lambda_{e}$ and $\mu_{e}$ in the definition of ${\cal A}_{G}$
  \begin{equation}
    \int \frac{d\lambda_{e}}{(\pi\theta)^{D}}\exp-\frac{2\imath}{\theta}\left\{\lambda_{e}\cdot J\left[u_{e}-\frac{1}{\sqrt{2}}\left(\sum_{i}\epsilon_{ei}y_{i}\right)\right]\right\}
  \end{equation}
  and
  \begin{equation}
    \int \frac{d\mu_{e}}{(\pi\theta)^{D}}\exp-\frac{2\imath}{\theta}\left\{\mu_{e}\cdot J\left[v_{e}-\frac{1}{\sqrt{2}}\left(\sum_{i}|\epsilon_{ei}|y_{i}\right)\right]\right\}.
  \end{equation}
  Gathering all the terms together, the expression of the amplitude
  reads
  \begin{equation} {\cal A}_{G}=\int{\textstyle
      \prod_{e}\!d\alpha_{e}}\quad\frac{1}{{\cal N}}\times\int d^{N}X
    \exp\left\{-\frac{1}{2}{}^{t}XAX+\imath{}^{t}XB+\frac{C}{2}\right\},\label{Gaussian}
  \end{equation}
  where
  \begin{equation}
    X=\sqrt{\frac{2}{\theta}}\Big(u_{e},v_{e},\lambda_{e},\mu_{e},y_{i}\Big)
  \end{equation}
  is a variable in ${\Bbb R}^{N}$ with $N=4e(G)D+|C^{\mbox{\tiny
      int}(G)}|D$. $A$ is a symmetric $N\times N$ matrix
  \begin{equation}
    A=\begin{pmatrix}
      \mbox{diag}(\Omega_{e}/t_{e})\otimes \mathrm{I}_{D}&0&\imath\mathrm{I}_{e(G)}\otimes J&0&0\cr
      0&\mbox{diag}(\Omega_{e}t_{e})\otimes \mathrm{I}_{D}&0&\imath\mathrm{I}_{e(G)}\otimes J&0\cr
      -\imath\mathrm{I}_{e(G)}\otimes J&0&0&0&\frac{\imath}{\sqrt{2}}\epsilon\otimes J\cr
      0&-\imath\mathrm{I}_{e(G)}\otimes J&0&0&\frac{\imath}{\sqrt{2}}|\epsilon|\otimes J\cr
      0&0&-\frac{\imath}{\sqrt{2}}{}^{t}\epsilon\otimes J&-\frac{\imath}{\sqrt{2}}{}|^{t}\epsilon|\otimes J&\imath\alpha\otimes J
    \end{pmatrix},
  \end{equation}
  with $\mathrm{I}_{M}$ the identity $M\times M$ matrix. $B\in {\Bbb
    R}^{N}$ and $C\in\imath{\Bbb R}$ are defined by
  \begin{equation}
    B=\sqrt{\frac{2}{\theta}}\,\Big(0,0,0,0\sum_{j\in C^{\mbox{\tiny ext}}_{v}}\beta_{ij} Jx_{j}\Big)
    \quad\mbox{and}\quad
    C=-\frac{2\imath}{\theta}\sum_{i,j\in C^{\mbox{\tiny ext}}_{v}}\gamma_{ij}\, x_{i}\cdot Jx_{j}.
  \end{equation}
  Finally, the normalization factor is
  \begin{equation} {\cal
      N}=\prod_{e}\left[\frac{2\Omega_{e}(1-t_{e}^{2})}{\theta\times
        2\pi \times
        2t_{e}}\right]^{D/2}\times\frac{1}{(\pi\theta)^{2e(G)D}}\times\frac{1}{(\pi\theta)^{|T(G)|D}}\times
    \Big(\frac{\theta}{2}\Big)^{N/2},
  \end{equation}
  whose respective contributions are the normalization factors of the
  Mehler kernels \eqref{Mehler}, the $\delta$ functions for the short
  and long variables, the contributions of the vertices ($|T(G)|$ is
  the total number of triangles in the representation of all vertices
  of $G$) and the Jacobian of the change of variables to $X$.

  We are now in a position to perform the Gau\ss ian integration over
  $X$ in \eqref{Gaussian},
  \begin{equation} {\cal A}_{G}=\int{\textstyle
      \prod_{e}\!d\alpha_{e}}\quad\frac{(2\pi)^{N/2}}{{\cal
        N}\sqrt{\det A}}\times\int d^{N}Z
    \exp\left\{-\frac{1}{2}{}^{t}BA^{-1}B+\frac{C}{2}\right\},\label{Gaussian2}
  \end{equation}
  where we assumed $A$ to be invertible, as it should be the case by
  its construction. Alternatively, one could have replaced $A$ by
  $A=\lambda \mathrm{I}_{N}$ with $\lambda$ large enough and show
  afterwards that the limit $\lambda\rightarrow 0$ is well
  defined. For simplicity, we do not do this here and will show, in proposition \ref{nonzero}, that $\det A>0$ provided $G$ does not
  contain an isolated graph with an even number of flags.

  To simplify the normalization factor, let us first derive a
  topological relation between the number of triangles and internal
  corners of any representation of the vertices of $G$ using
  triangles. In each case, the graph obtained by joining the center of
  adjacent triangles is a forest (i.e. a graph without cycles) with
  $|T(G)|$ vertices and $v(G)$ connected components, so that there are
  $|T(G)|-v(G)$ corners common to two triangles. Next, each
  triangle has $3$ corners, which are either attached to flags or
  internal corners, with the internal corners common to two triangles
  counted twice. Accordingly, $3|T(G)|=|C^{\mbox{\tiny
      int}(G)}|+f(G)+(|T(G)|-v(G))$, so that
  \begin{equation}
    2|T(G)|=|C^{\mbox{\tiny int}(G)}|+f(G)-v(G).\label{numberoftriangles}
  \end{equation} 
  Using this relation, we get
  \begin{equation}
    \frac{(2\pi)^{D/2}}{{\cal N}}=\frac{(2\pi\theta)^{|C^{\mbox{\tiny int}(G)}|}}{(2\pi\theta)^{e(G)}(\pi\theta)^{2||T(G)||}2^{{|C^{\mbox{\tiny int}(G)}|}}}=
    2^{f(G)-v(G)}(2\pi\theta)^{v(G)-e(G)-f(G)}.
  \end{equation}

  To define $\mathrm{HU}_{G}(\Omega,t)$ it is helpful to note that the
  matrix $A$ can be written as
  \begin{equation}
    A=D\otimes \mathrm{I}_{D}+R\otimes \imath J,
  \end{equation} 
  with $D$ diagonal and $R$ antisymmetric. The matrix
  $P^{-1}\imath JP$ with $P$ the $D\times D$ block diagonal matrix
  made of $2\times 2$ blocks
  $\begin{pmatrix}\frac{\imath}{\sqrt{2}}&
    -\frac{\imath}{\sqrt{2}}\cr\frac{1}{\sqrt{2}}&\frac{1}{\sqrt{2}}\end{pmatrix}$
  is diagonal with blocks
  $\begin{pmatrix}1&0\cr0&-1\end{pmatrix}$. Therefore,
  \begin{equation}
    \det A=\big[\det(D+R)\big]^{D/2}\times\big[\det(T-R)\big]^{D/2}=\big[\det(D+R)\big]^{D}
  \end{equation}  
  since $\det(T-R)=\det{}^{t}(D-R)=\det(D+R)$. Thus,
  \begin{equation}
    \mathrm{HU}_{G}(\Omega,t)=2^{v(G)}\,\big[{\textstyle \prod_{e}t_{e}}\big]\,\det(D+R)\label{defHU}
  \end{equation}
  is a polynomial in $t_{e}$ (because of the multiplication by
  $\prod_{e}t_{e}$) and in $\Omega_{e}$ and
  \begin{equation}
    \frac{(2\pi)^{N/2}}{{\cal N}\sqrt{\det A}}=\left[\frac{2^{f(G)}\prod_{e}\Omega_{e}(1-t_{e}^{2})}
      {(2\pi\theta)^{e(G)+f(G)-v(G)}\mathrm{HU}_{G}(\Omega,t)}
    \right]^{D/2},
  \end{equation}
  which corresponds to the prefactor in \eqref{amplitude}.

  Finally, taking into account \eqref{defHU}, we define the second
  hyperbolic polynomial as
  \begin{equation}
    \mathrm{HV}_{G}(\Omega,t,x)=2^{v(G)}\theta\big[{\textstyle \prod_{e}t_{e}}\big]\,\det(D+R)\big[{}^{t}BA^{-1}B+C\big].
    \label{defHV}
  \end{equation}
  The only non trivial assertion to check is its polynomial dependence
  on $t_{e}$. The latter follows from
  \begin{equation}
    A^{-1}=(D+R)^{-1}\otimes\big({\textstyle \frac{1+\imath J}{2}}\big)+(T-R)^{-1}\otimes\big({\textstyle \frac{1-\imath J}{2}}\big),
  \end{equation}
  so that $\big[{\textstyle \prod_{e}t_{e}}\big]\,\det(D+R)A^{-1}$ is
  a matrix of polynomials in $t_{e}$ since $ \big[{\textstyle
    \prod_{e}t_{e}}\big]\,\det(D+R)(D+R)^{-1}$ and $\big[{\textstyle
    \prod_{e}t_{e}}\big]\,\det(T-R)(T-R)^{-1}$ are.
\end{proof}

\begin{rem}
When expressed in term of $\Omega$ and $t$, both hyperbolic polynomials $\mathrm{HU}_{G}(\Omega,t)$ and $\mathrm{HV}_{G}(\Omega,t,x)$ do not depend on $\theta$. This is the consequence of the use of the Mehler kernel ${\cal K}_{\widetilde{\Omega}}$ in the kinetic term, with $\widetilde{\Omega}=\frac{2\Omega}{\theta}$. However, there is an implicit $\theta$ dependence in the relation between $t$ and $\alpha$.
\end{rem}

\section{Hyperbolic polynomials as graph polynomials}
\label{sec:hyperb-polyn-as}
\subsection{Reduction relation for the first hyperbolic polynomial}

In general, it is not very convenient to study the hyperbolic polynomials starting from the relations \eqref{defHU} and \eqref{defHV}. It is preferable to compute the determinants by a series of successive reductions, instead of trying to manipulate them in one go. This leads to the following reduction relation.   

\begin{thm}
\label{reductionthm}
The first hyperbolic polynomial  $\mathrm{HU}_{G}$, defined by \eqref{defHU} for any ribbon graph with flags, is multiplicative over disjoint unions, obeys the reduction relation
\begin{equation}
\mathrm{HU}_{G}=t_{e}\,\mathrm{HU}_{G-e}+t_{e}\Omega_{e}^{2}\,\mathrm{HU}_{G\cut e}
+\Omega_{e}\,\mathrm{HU}_{G^{e}-e}+\Omega_{e}t_{e}^{2}\,\mathrm{HU}_{G^{e}\cut e}\label{reduction}
\end{equation}
for any edge $e$. Furthermore,  for the graph $V_{n}$ consisting of an isolated vertex with $n$ flags, we have
\begin{equation}
\mathrm{HU}_{V_{n}}=
\left\{
\begin{array}{l}
2\quad\mbox{if }n\,\mbox{is even},\cr
0\quad\mbox{if }n\,\mbox{is odd}.
\end{array}
\right.\label{isolated}
\end{equation}
\end{thm}
\begin{proof}
  Let us recall the defining relation \eqref{defHU} of the first
  hyperbolic polynomial as a determinant,
  \begin{equation}
    \mathrm{HU}_{G}(\Omega,t)=2^{v(G)}\,\big[{\textstyle \prod_{e}t_{e}}\big]\,\det(D+R).
  \end{equation}
  The multiplicativity follows readily from \eqref{defHU} since the
  adjacency and incidence matrices of a disjoint union are block
  diagonal.

  Although all the graphical operations appearing in the reduction
  relations can be performed on the lines and columns of $D+R$, it is
  much more economical to derive them using technics from Grassmannian
  calculus (see for instance \cite{Abdesselam2009aa} for a recent overview of Grassmannian calculus). To proceed, write the determinant as a Gau\ss ian
  integral over Grassmann variables with
  $\left\{\rho,\sigma\right\}\in\left\{u_{e},v_{e},\lambda_{e},\mu_{e},y_{i}\right\}$\footnote{For
    the sake of clarity we use here the same letter for indices and
    the corresponding integration variables in the previous section.},
  \begin{equation}
    \det(D+R)=\int \prod_{\rho}d\overline{\psi}_{\rho}d\psi_{\rho}\exp-\Big\{\sum_{\rho,\sigma}\overline{\psi}_{\rho}(D+R)_{\rho\sigma}\psi_{\sigma}\Big\}.
  \end{equation}
  Next, we perform the change of variables of
  \begin{equation}
    \left\{
      \begin{array}{rcl}
        \overline{\psi}_{\rho}&=&\frac{1}{\sqrt{2}}(\chi_{\rho}-\imath\eta_{\rho}),\cr
        \psi_{\rho}&=&\frac{1}{\sqrt{2}}(\chi_{\rho}+\imath\eta_{\rho}),
      \end{array}
    \right.\quad \mbox{with Jacobian}\quad
    \frac{D(\overline{\psi},\psi)}{D(\chi,\eta)}=\imath.
  \end{equation}

  Because all Grassmann variables anticommute, the determinant is
  expressed as
  \begin{multline}
    \det (D+R)=\cr\int \prod_{\rho}\big[-\imath d{\chi}_{\rho}d\eta_{\rho}\big]\,\exp\imath\Big\{\sum_{\rho}d_{\rho}\chi_{\rho}\eta_{\rho}\Big\}
    \exp-{\textstyle \frac{1}{2}}\Big\{\sum_{\rho,\sigma}R_{\rho\sigma}(\chi_{\rho}\chi_{\sigma}+\eta_{\rho}\eta_{\sigma})\Big\},
  \end{multline}
  with $d_{\rho}=\frac{\Omega_{e}}{t_{e}}$
  (resp. $d_{\rho}=\Omega_{e}t_{e}$) for $\rho=u_{e}$
  (resp. $\rho=v_{e}$) and 0 otherwise. Note that $\int
  \prod_{\rho}d{\chi}_{\rho}\exp-\frac{1}{2}\Big\{\sum_{\rho,\sigma}R_{\rho\sigma}\chi_{\rho}\chi_{\sigma}\Big\}$
  (or the equivalent expression using $\eta$) is the Pfaffian of the
  antisymmetric matrix $R_{\rho\sigma}$.

  Let us select a particular edge $e$ from the corner $i$ to the
  corner $j$ and expand the related exponential
  \begin{equation}
    \begin{array}{rcl}
      \det (D+R)&=&{\displaystyle \int }\prod_{\rho}\big[-\imath d{\chi}_{\rho}d\eta_{\rho}]
      \Big(1+\imath\frac{\Omega_{e}}{t_{e}}\eta_{u_{e}}\chi_{u_{e}}+\imath\Omega_{e}t_{e}\eta_{v_{e}}\chi_{v_{e}}
      -(\Omega_{e})^{2}\eta_{u_{e}}\chi_{u_{e}}\eta_{v_{e}}\chi_{v_{e}}\Big)\cr
      &&\times \exp\imath\Big\{\sum_{\rho\neq u_{e},v_{e}}d_{\rho}\chi_{\rho}\eta_{\rho}\Big\}
      \exp-\frac{1}{2}\Big\{\sum_{\rho,\sigma}R_{\rho\sigma}(\chi_{\rho}\chi_{\sigma}+\eta_{\rho}\eta_{\sigma})\Big\}
      \label{reductiongrass}
    \end{array}.
  \end{equation}
  Moreover, since the operations on the variables $\eta$ and $\chi$
  are identical and independent, we perform them explicitely only on
  $\eta$. In the sequel, we repeatedly use the following elementary
  result from Grassmannian calculus
  \begin{lemma}
    Let $F$ be a function of the Grassmann variables
    $\eta_{1},\eta_{2},\dots$ (i.e. an element of the exterior algebra
    generated by $\eta_{1},\eta_{2},\dots $). Then
    \begin{equation}
      \int d\eta_{1}\,\eta_{1}F(\eta_{1},\eta_{2},\dots)=
      F(0,\eta_{2},\dots),\label{delta1}
    \end{equation}
    and its corollary, the integral representation of the Grassmannian
    $\delta$ function
    \begin{equation}
      \int d\eta_{0}d\eta_{1}\,\exp a\left\{\eta_{0}\eta_{1}\right\}F(\eta_{1},\eta_{2},\dots)=
      d\eta_{1}\,a\delta(\eta_{1})F(\eta_{1},\eta_{2},\dots)=
      aF(0,\eta_{2},\dots).\label{delta2}
    \end{equation}
  \end{lemma}

  It is convenient to explicit all the terms involving the edge $e$ in
  the Pfaffian
  \begin{equation}
    \sum_{\rho,\sigma}R_{\rho\sigma}\eta_{\rho}\eta_{\sigma}=
    \eta_{u_{e}}\eta_{\lambda_{e}}+\lambda_{v_{e}}\eta_{\mu_{e}}
    +{\textstyle \frac{1}{\sqrt{2}}}\eta_{\lambda_{e}}\big(\eta_{y_{j}}-\eta_{y_{i}}\big)
    +{\textstyle \frac{1}{\sqrt{2}}}\eta_{\mu_{e}}\big(\eta_{y_{j}}+\eta_{y_{i}}\big)+\dotsm.
  \end{equation}
  To alleviate the expressions, we make the convention that in the
  following we only represent the part of the Grasmann integral
  affected by the equations.

  The first term in \eqref{reductiongrass},
  \begin{gather}
    \int d\eta_{u_{e}}d\eta_{v_{e}}d\eta_{\lambda_{e}}d\eta_{\mu_{e}}\cr
    \exp-\left\{\eta_{u_{e}}\eta_{\lambda_{e}}+\lambda_{v_{e}}\eta_{\mu_{e}}
      +{\textstyle\frac{1}{\sqrt{2}}}\eta_{\lambda_{e}}\big(\eta_{y_{j}}-\eta_{y_{i}}\big)
      +{\textstyle \frac{1}{\sqrt{2}}}\eta_{\mu_{e}}\big(\eta_{y_{j}}+\eta_{y_{i}}\big)+\dotsm\right\},
  \end{gather}
  corresponds to the deletion of $e$ in $G$ since the integration over
  $\eta_{u_{e}}$ and $\eta_{v_{e}}$ sets
  $\eta_{\lambda_{e}}=\eta_{\mu_{e}}=0$ by using \eqref{delta2}. Then,
  the corners $i$ and $j$ remain as isolated corners. Let us note that
  the factors of $\imath$ cancel since we integrate over 4 pairs
  $\chi_{\rho}\eta_{\rho}$ and that no extra sign arise form the
  commutation of $d\chi_{\rho}$ and $d\eta_{\rho}$, since the latter
  are always performed pairwise on $\chi_{\rho}$ and $\eta_{\rho}$.

  In the second term,
  \begin{gather}
    \int d\eta_{u_{e}}d\eta_{v_{e}}d\eta_{\lambda_{e}}d\eta_{\mu_{e}}\cr
    \eta_{u_{e}}\exp-\left\{\eta_{u_{e}}\eta_{\lambda_{e}}\!+\!\lambda_{v_{e}}\eta_{\mu_{e}}
      \!+\!{\textstyle\frac{1}{\sqrt{2}}}\eta_{\lambda_{e}}\big(\eta_{y_{j}}\!-\!\eta_{y_{i}}\big)
      \!+\!{\textstyle \frac{1}{\sqrt{2}}}\eta_{\mu_{e}}\big(\eta_{y_{j}}\!+\!\eta_{y_{i}}\big)+\dotsm\right\},
  \end{gather}
  we have $\eta_{u_{e}}=0$ and the integration over $\eta_{v_{e}}$
  sets $\eta_{\mu_{e}}=0$ with an extra sign. The remaining
  integration over $\eta_{\lambda_{e}}$ enforces
  $\delta(\eta_{y_{i}}-\eta_{y_{j}})$ after integration over with an
  extra factor of $-1/\sqrt{2}$ using $\eqref{delta2}$. To relate this
  operation to the deletion in the partial dual $G^{e}$, we need to
  distinguish two cases.
  \begin{itemize}

  \item If $e$ is not a loop, then $G^{e}-e$ results from the
    identification of the corners $i$ and $j$ (belonging to two
    different vertices) to get a single vertex, as required by
    $\eta_{y_{i}}=\eta_{y_{j}}$.  Taking into account both Pfaffians
    and the prefactor $2^{v(G)}$, we get
    $\left(\frac{1}{\sqrt{2}}\right)^{2}\times2^{v(G)}=2^{v(G^{e}-e)}$.
    After taking into account the variables $\chi_{\rho}$ and
    $\eta_{\rho}$, we integrate over an odd number of pairs, so that a
    factor of $\imath$ remains, which cancels with the one in
    \eqref{reductiongrass}.

  \item Let us suppose that $e$ is a loop. Using the freedom we have
    in representing the vertex using triangles, we may always assume
    that $i$ and $j$ lie on adjacent triangles $(ikl)$ and $(jkm)$
    with a common corner $k$ and related to the remaining part of the
    graph by two additional corners $l$ and $m$.
    \begin{figure}[!htp]
      \centering
      \includegraphics[scale=.8]{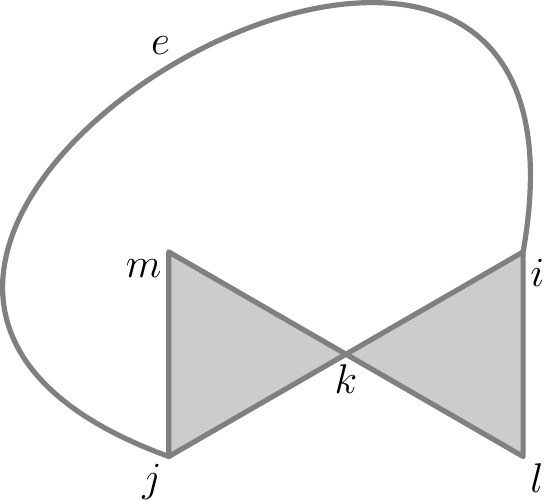}
      \caption{A loop $e$}
      \label{fig:LoopE}
    \end{figure}

    \noindent
    The contribution of the two triangles to the Pfaffian is
    \begin{equation}
      \exp-\left\{\eta_{y_{i}}\eta_{y_{k}}+\eta_{y_{k}}\eta_{y_{l}}+\eta_{y_{l}}\eta_{y_{i}}+
        \eta_{y_{j}}\eta_{y_{k}}+\eta_{y_{k}}\eta_{y_{m}}+\eta_{y_{m}}\eta_{y_{j}}\right\}.\label{twotriangles}
    \end{equation}
    After the identification $\eta_{y_{i}}=\eta_{y_{j}}$, the
    contribution of the triangles $ikl$ and $jkm$ reads
    \begin{align}
      \int d\eta_{y_{i}}d\eta_{y_{k}}\exp-\left\{
        2\eta_{y_{k}}\eta_{y_{i}}+(\eta_{y_{i}}-\eta_{y_{k}})(\eta_{y_{l}}+\eta_{y_{m}})\right\}=\cr
      \int d\eta_{+}d\eta_{-}\exp-\left\{2
        \eta_{+}\eta_{-}+\sqrt{2}\eta_{-}(\eta_{y_{l}}+\eta_{y_{m}})\right\},
    \end{align}
    using the change of variables $
    \eta_{\pm}=\frac{\eta_{y_{i}}\pm\eta_{y_{k}}}{\sqrt{2}}$. Using
    \eqref{delta2}, the integration over $\eta_{+}$ sets $\eta_{-}=0$
    with an extra $-2$, so that the contribution of the two triangles is
    trivial. Therefore, we suppress the latter, which is nothing but
    the deletion of $e$ in $G^{e}$. Finally, the factors of $2$ are
    $\left(\frac{1}{\sqrt{2}}\right)^{2}\times 2^{2}\times
    2^{v(G)}=2\times 2^{v(G)}=2^{v(G^{e}-e)}$ since $e$ is a loop.
    The signs and factors of $\imath$ also cancel since we
    integrate over 7 pairs of variables.
  \end{itemize}

  The third term,
  \begin{gather}
    \int d\eta_{u_{e}}d\eta_{v_{e}}d\eta_{\lambda_{e}}d\eta_{\mu_{e}}\cr
    \eta_{v_{e}}\!\exp-\left\{\eta_{u_{e}}\eta_{\lambda_{e}}\!+\!\eta_{v_{e}}\eta_{\mu_{e}}
      \!+\!{\textstyle\frac{1}{\sqrt{2}}}\eta_{\lambda_{e}}\big(\eta_{y_{j}}\!-\!\eta_{y_{i}}\big)
      \!+\!{\textstyle \frac{1}{\sqrt{2}}}\eta_{\mu_{e}}\big(\eta_{y_{j}}\!+\!\eta_{y_{i}}\big)+\dotsm\right\},
  \end{gather}
  is very similar to the second one, except that the integration over
  $\eta_{u_{e}}$, $\eta_{v_{e}}$, $\eta_{\lambda_{e}}$ results in
  \begin{equation}
    \int d\eta_{\mu_{e}}\exp-\left\{ (\eta_{\mu_{e}})\frac{\eta_{y_{i}}+\eta_{y_{j}}}{\sqrt{2}}\right\}.
  \end{equation}
  Again, we distinguish two cases
  \begin{itemize}
  \item If $e$ is not a loop, let us set a new variable
    $\eta_{y_{p}}=\frac{\eta_{\mu_{e}}}{\sqrt{2}}$ ($p$ does not
    correspond to an existing corner in $G$), so that
    \begin{equation}
      \int d\eta_{\mu_{e}}\exp-\left\{ (\eta_{\mu_{e}})\frac{\eta_{y_{i}}+\eta_{y_{j}}}{\sqrt{2}}\right\}
      =\frac{1}{\sqrt{2}}\int d\eta_{y_{p}}\exp-\left\{ \eta_{y_{p}}\eta_{y_{i}}+\eta_{y_{p}}\eta_{y_{j}}\right\}.
    \end{equation}
    This is the contribution of two triangles $(piq)$ and $(pjr)$
    attached by a common corner $p$ with flags on $q$ and $r$, so that
    there are no terms in $\eta_{y_{q}}=\eta_{y_{r}}0$. Graphically,
    it corresponds to identifying the corners $i$ and $j$ with two
    extra flags separating the two parts of the graph that were
    attached to the corners $i$ and $j$. This is the cut of $e$ in the
    partial dual $G^{e}$.
  \item If $e$ is a loop, then we perform the integration over
    $\eta_{\mu_{e}}$ which enforces $\eta_{y_{i}}+\eta_{y_{j}}=0$ with
    an extra factor $\frac{-1}{\sqrt{2}}$. As in the discussion of the
    second case, without loss of generality we assume that $i$ and $j$
    lie on adjacent triangles $(ikl)$ and $(jkm)$ whose contribution
    is given by \eqref{twotriangles}.  After the identification
    $\eta_{y_{j}}=-\eta_{y_{i}}$, we are left with
    \begin{align}
      \int d\eta_{y_{i}}d\eta_{y_{k}}\exp-\left\{
        (\eta_{y_{i}}-\eta_{y_{k}})\eta_{y_{l}}+
        (\eta_{y_{i}}+\eta_{y_{k}})\eta_{y_{m}}\right\}=\cr \int
      d\eta_{+}d\eta_{-}\exp-\left\{
        \sqrt{2}\eta_{-}\eta_{y_{l}}+{\textstyle
          \sqrt{2}}\eta_{+}\eta_{y_{m}}\right\},
    \end{align}
    using the change of variables $
    \eta_{\pm}=\frac{\eta_{y_{i}}\pm\eta_{y_{k}}}{\sqrt{2}}$. Using
    \eqref{delta2}, the integration over $\eta_{+}$ and $\eta_{-}$
    sets $\eta_{y_{l}}=\eta_{y_{m}}=0$ with an extra factor of 2, so
    that the contribution of the two triangles is trivial. Therefore,
    we suppress the latter, which is nothing but the deletion of $e$
    in $G^{e}$. As in the previous case, all the factors of $-1$, $2$
    and $\imath$ cancel after we take into account the
    contributions of both Pfaffians.
  \end{itemize}

  The fourth term,
  \begin{gather}
    \int d\eta_{u_{e}}d\eta_{v_{e}}d\eta_{\lambda_{e}}d\eta_{\mu_{e}}\cr
    \eta_{u_{e}}\eta_{v_{e}}\!\exp-\left\{\eta_{u_{e}}\eta_{\lambda_{e}}+\lambda_{v_{e}}\eta_{\mu_{e}}
      +{\textstyle\frac{1}{\sqrt{2}}}\eta_{\lambda_{e}}\big(\eta_{y_{j}}-\eta_{y_{i}}\big)
      +{\textstyle \frac{1}{\sqrt{2}}}\eta_{\mu_{e}}\big(\eta_{y_{j}}+\eta_{y_{i}}\big)+\dotsm\right\},
  \end{gather}
  represents the cut of $e$ in $G$ since the integration over
  $\eta_{u_{e}}$ and $\eta_{v_{e}}$ sets
  $\eta_{\lambda_{e}}=\eta_{\mu_{e}}=0$ by using \eqref{delta1}. The
  remaining integrations over $\eta_{\lambda_{e}}$ and
  $\eta_{\mu_{e}}$ can be written as
  \begin{equation}
    \int d{\eta^{+}}d{\eta^{-}}\exp\left\{\eta_{+}\eta_{y_{j}}+\eta_{-}\eta_{y_{i}}\right\}
    \quad\mbox{with}\quad \eta_{\pm}=\frac{\eta_{\mu_{e}}\pm\eta_{\lambda_{e}}}{\sqrt{2}},
  \end{equation} 
  which imposes $\eta_{y_{i}}=\eta_{y_{j}}=0$. Graphically, this means
  that the corners $i$ and $j$ become flags, which yields $G\cut
  e$. Here, the are neither powers of 2, nor extra signs arising from
  the operations. However, we integrate over 6 pairs of variables, so
  that the Jacobians yield -1, which cancel with the sign in
  \eqref{reductiongrass}.

  Finally, let us prove the assertion concerning the isolated
  vertices. In this case, $D+R$ reduces to $\alpha$, the antisymmetric
  adjacency matrix $\alpha$ of the internal corners of the graph,
  defined in the proof of theorem \ref{defHUHV}. For a vertex with an
  even number of flags, we have an odd number of internal corners
  because of the relation \eqref{numberoftriangles}, so that
  \begin{equation}
    \mathrm{HU}_{V_{2n}}=2\det(\alpha)=0.
  \end{equation}
  In case of an even number of flags
  \begin{equation}
    \mathrm{HU}_{V_{2n+1}}=2\det(\alpha)=2\big[\mbox{Pf}(\alpha)\big]^{2}.
  \end{equation}
  Recall that the Pfaffian of a $2n\times 2n$ antisymmetric matrix is
  defined as
  \begin{equation}
    \mbox{Pf}(\alpha)=\sum_{\pi\in \Pi_{n}}(-1)^{\mathrm{sign}(\pi)}\alpha_{\pi(1),\pi(2)}\alpha_{\pi(3),\pi(4)}\cdots\alpha_{\pi(2n\!-\!1),\pi(2n)}
    ,\end{equation}
  with $\Pi_{n}$ the subset of the permutations of $\left\{1,2\dots,2n\right\}$ such that $\pi(2i\!-\!1)<\pi(2i)$ for any $1\leq i\leq n$ and $\pi(1)<\pi(3)<\dots<\pi(2n-1)$. Accordingly, if $\alpha$ is the adjacency matrix of a graph, its Pfaffian is a sum over all its perfect matchings, with relative signs. In the case of the graph build with the edges of the triangles pertaining to a vertex of odd degree and with all the external corners and the triangle edges attached to them removed, it is easy to show by induction on the number of triangles, that there is a unique perfect matching on the triangle edges, with the convention that the empty graph has a unique perfect matching, the empty one. Therefore $\mbox{Pf}(\alpha)=\pm 1$, so that $\mathrm{HU}_{V_{2n+1}}=2$.
\end{proof}

For a graph with $e(G)$ edges, the reduction relation \eqref{reduction} involves $4^{e(G)}$ operations, many of them leading to terminal forms containing a vertex of even degree. For $E(G)\geq 3$, it is therefore not very convenient to compute $\mathrm{HU}_{G}$ using the reduction relation. However, it is instructive to see how it works on the simplest examples with 1 and 2 edges. 

\begin{example}[Bridge with flags]
Let $B_{m,n}$ be the bridge (i.e. one edge and two vertices) with $m$ flags on one vertex and $n$ flags on the other one. Then, the reduction relation reads 
\begin{equation}
\mathrm{HU}_{B_{m,n}}(\Omega_{1},t_{1})=t_{1}\,\mathrm{HU}_{V_{m}\cup V_{n}}+t\Omega_{1}^{2}\,\mathrm{HU}_{V_{m+1}\cup V_{n+1}}
+\Omega_{1}\,\mathrm{HU}_{V_{m+n}}+\Omega_{1} t_{1}^{2}\,\mathrm{HU}_{V_{m+n+2}},
\end{equation},
so that we obtain
\begin{equation}
\mathrm{HU}_{B_{m,n}}(\Omega_{1},t)=
\left\{
\begin{array}{ll}
4\,t_{1}\Omega_{1}^{2}&\quad\mbox{if }\;m\;\mbox{and}\; n\; \mbox{are even},\cr
4\,t_{1}&\quad\mbox{if }\;m\;\mbox{and}\; n\; \mbox{are odd},\cr
2\,\Omega_{1}(1+t_{1}^{2})&\quad\mbox{otherwise}.
\end{array}
\right.\label{bridgeexample}
\end{equation}
\end{example}
.

\begin{example}[Loop with flags]
Let $L_{m,n}$ be the loop (i.e. one edge and one vertex) with $m$ flags on one face and $n$ flags on the other one. The reduction relation  
\begin{equation}
\mathrm{HU}_{L_{m,n}}(\Omega_{1},t_{1})=\Omega_{1}\,\mathrm{HU}_{V_{m}\cup V_{n}}+t\,\mathrm{HU}_{V_{m+n}}
+\Omega_{1} t_{1}^{2}\,\mathrm{HU}_{V_{m+1}\cup V_{n+1}}+t_{1}\Omega_{1}^{2}\,\mathrm{HU}_{V_{m+n+2}},
\end{equation}
implies
\begin{equation}
\mathrm{HU}_{L_{m,n}}(\Omega_{1},t_{1})=
\left\{
\begin{array}{ll}
4\,\Omega_{1} t_{1}^{2}&\quad\mbox{if }\;m\;\mbox{and}\; n\; \mbox{are even},\cr
4\,\Omega_{1}&\quad\mbox{if }\;m\;\mbox{and}\; n\; \mbox{are odd},\cr
2\,t_{1}(1+\Omega_{1}^{2})&\quad\mbox{otherwise}.
\end{array}
\right.\label{loopexample}
\end{equation}
\end{example}

\begin{example}[Cycle of length 2 without flags]\label{2cycleex}
Let us consider a cycle of length two without flags. The reduction relation reads 
\begin{align}
\mathrm{HU}_{\mbox{\tiny cycle}\atop\mbox{\tiny 2 edges, no flag}}(\Omega_{1},\Omega_{2},t_{1},t_{2})=\hskip3cm{}\cr
t_{1}\,\mathrm{HU}_{B_{0,0}}(\Omega_{2},t_{2})+t_{1}\Omega_{1}^{2}\,\mathrm{HU}_{B_{1,1}}(\Omega_{2},t_{2})
+\Omega_{1}\,\mathrm{HU}_{L_{0.0}}(\Omega_{2},t_{2})+\Omega_{1}t_{1}^{2}\,\mathrm{HU}_{V_{1,1}}(\Omega_{2},t_{2}),
\end{align}
and we get, using the previous two examples, 
\begin{equation}
\mathrm{HU}_{\mbox{\tiny cycle}\atop\mbox{\tiny 2 edges, no flag}}(\Omega_{1},\Omega_{2},t_{1},t_{2})=4\big(t_{1}^{2}+t_{2}^{2}\big)\Omega_{1}\Omega_{2}+4\big(\Omega_{1}^{2}+\Omega_{2}^{2}\big)t_{1}t_{2}.\label{2cycle}
\end{equation}
\end{example}

\subsection{Some properties of $\mathrm{HU}_{G}$ as a graph polynomial}

We are now ready to give the combinatorial expression of the first hyperbolic polynomial.

\begin{thm}\label{HU=Q}
The first hyperbolic polynomial can be expressed as
\begin{equation}
\mathrm{HU}_{G}(\Omega,t)={\cal Q}_{G}(t,\Omega,\Omega t^{2},t\Omega^{2},r),
\end{equation}
with $r_{2n+1}=2$ and $r_{2n}=0$, or explicitly,
\begin{multline}
\mathrm{HU}_{G}(\Omega,t)=\cr
\sum_{A,B\subset E(G)\atop\mathrm{admissible}}\left\{2^{v(G^{A})}
\Big(\prod_{e\in A^{c}\cap B^{c}}t_{e}\Big)
\Big(\prod_{e\in A^{c}\cap B}t_{e}\Omega_{e}^{2}\Big)
\Big(\prod_{e\in A\cap B^{c}}\Omega_{e}\Big)
\Big(\prod_{e\in A\cap B}\Omega_{e}t_{e}^{2}\Big)\right\},\label{expansion}
\end{multline}
with $(A,B)$ admissible if each vertex of the graph obtained from $G^{A}$ by cutting the edges in $B$ and deleting those in $B^{c}$ has an odd number of flags.
\end{thm}
\begin{proof}
  Recall that for ribbon graph with flags the polynomial ${\cal
    Q}_{G}(x,y,z,w,r)$, depending on four variables
  $(x_{e},t_{e},z_{e},w_{e})$ for each edge and a sequence
  $(r_{n})_{n\in{\Bbb N}}$, is defined as
  \begin{multline} {\cal Q}_{G}(x,y,z,w,r)=\hskip3cm{}\cr
    \sum_{A,B\subset E(G)}\left\{\Big(\prod_{v\in
        V(G^{A})}r_{d_{v}}\Big) \Big(\prod_{e\in A^{c}\cap
        B^{c}}x_{e}\Big) \Big(\prod_{e\in A^{c}\cap B}w_{e}\Big)
      \Big(\prod_{e\in A\cap B^{c}}y_{e}\Big) \Big(\prod_{e\in A\cap
        B}z_{e}\Big)\right\}
    \label{Qexpansion}
  \end{multline}
  The graph polynomial ${\cal Q}_{G}$ can be characterized as the
  unique graph polynomial which is multiplicative over disjoint
  unions, that obeys the reduction relation
  \begin{equation} {\cal Q}_{G}=x_{e}{\cal Q}_{G-e}+w_{e}{\cal
      Q}_{G\vee e}+ y_{e}{\cal Q}_{G^{e}-e}+z_{e}{\cal Q}_{G^{e}\vee
      e},
  \end{equation}
  for any edge $e\in E(G)$ and that takes the value ${\cal
    Q}_{V_{n}}(x,y,z,w,r)=r_{n}$ on a isolated vertex with $n$
  flags, see definition \ref{def:RecursDefQ}. These three conditions are precisely the content of theorem
  \ref{reductionthm}, with $x_{e}=t_{e}$, $y_{e}=\Omega_{e}$,
  $z_{e}=\Omega_{e}t_{e}^{2}$, $w_{e}=t_{e}\Omega_{e}^{2}$, $r_{2n}=0$
  and $r_{2n+1}=2$. The relation $r_{2n}=0$ reduces the summation to
  admissible subsets $(A,B)$ and $r_{2n+1}=2$ yields a factor of 2 for
  each vertex of $G_{A}$.
\end{proof}

This formula can be used to compute $\mathrm{HU}_{G}(\Omega,t)$ for
simple examples that admit many symmetries.  Otherwise there are many
possibilities for the subsets $A$ and $B$ that have to be treated, many of them being non admissible.

\begin{example}[Planar banana with three edges]\label{planarbananaex}

Let us consider the planar graph with two vertices and three edges, all of three having both ends attached to different vertices. With $A=\emptyset$, we must have $|B|$ odd. Thus, we get four terms
\begin{equation}
4t_{1}t_{2}t_{3}(\Omega_{1}^{2}+\Omega_{2}^{2}+\Omega_{3}^{2}+\Omega_{1}^{2}\Omega_{2}^{2}\Omega_{3}^{2}),
\end{equation}
If $|A|=1$, $G^{A}$ has a single vertex without flags, so that no cut could yield an odd number of flags. When $|A|=2$, let us suppose that $A=\left\{1,2\right\}$ for definiteness. Then, $G^{A}$ is a cycle with two edges 1 and 2 and an extra loop 3 attached to one of the vertices. We have 4 possibilities for $B$: $\left\{ 1\right\}$, $\left\{ 2\right\}$, $\left\{ 1,3\right\}$ and $\left\{ 2,3\right\}$, that yield the monomials:
\begin{equation}   
4\Omega_{1}\Omega_{2}(t_{1}^{2}t_{3}+t_{2}^{2}t_{3}+t_{1}^{2}t_{3}\Omega_{3}^{2}+t_{2}^{2}t_{3}\Omega_{3}^{2}).
\end{equation}
By cyclic symmetry, we construct 8 other terms that correspond to $A=\left\{2,3\right\}$ and $A=\left\{1,3\right\}$. Finally, with $|A|=3$, $G^{A}$ is a triangle and there is no way to get only odd vertices after cutting. Therefore, we obtain
\begin{multline}
\mathrm{HU}_{\mbox{\tiny planar}\atop \mbox{\tiny 3-banana}}(\Omega,t)=4\Big[t_{1}t_{2}t_{3}\big[\Omega_{1}^{2}+\Omega_{2}^{2}+\Omega_{3}^{2}+\Omega_{1}^{2}\Omega_{2}^{2}\Omega_{3}^{2}\big]
+t_{1}\Omega_{2}\Omega_{3}\big[t_{2}^{2}+t_{3}^{2}+\Omega_{1}^{2}(t_{2}^{2}+t_{3}^{2})\big]\cr
+t_{2}\Omega_{1}\Omega_{3}\big[t_{1}^{2}+t_{3}^{2}+\Omega_{2}^{2}(t_{1}^{2}+t_{3}^{2})\big]
+t_{3}\Omega_{1}\Omega_{2}\big[t_{1}^{2}+t_{2}^{2}+\Omega_{3}^{2}(t_{1}^{2}+t_{2}^{2})\big]
\Big],
\end{multline}

\end{example}

A first consequence of theorem \ref{HU=Q} is the invariance of the first hyperbolic polynomial under partial duality, provided we interchange some of the variables $\Omega$ and $t$.

\begin{cor}
For any $A\subset E(G)$, the first hyperbolic polynomial transforms under partial duality as
\begin{equation}
\mathrm{HU}_{G^{A}}(\Omega_{A},t_{A})=\mathrm{HU}_{G}(\Omega,t),
\end{equation}
with 
\begin{equation}
\left\{
\begin{array}{llcc}
[\Omega_{A}]_{e}=t_{e},& [t_{A}]_{e}=\Omega_{e}&\mbox{for }&e\in A,\cr
[\Omega_{A}]_{e}=\Omega_{e},& [t_{A}]_{e}=t_{e}&\mbox{for }&e\notin A.\cr
\end{array}
\right.
\end{equation}
\end{cor}
\begin{proof}
  This an immediate consequence of the relation between
  $\mathrm{HU}_{G}$ and ${\cal Q}_{G}$ and of the transformation of
  ${\cal Q}_{G}$ under partial duality, see theorem \ref{thm:PartialDualityQ}.
\end{proof}
\begin{rem}
It is worthwhile to notice that this is a rather strong result, since the duality holds with respect to any subset of edges for all graphs, in contradistinction with the commutative case, where only the duality with respect to all edges holds for a planar graph.  Note that this property also holds for the noncommutative field theories with Moyal interaction and heat-kernel propagator (see corollary \ref{dualityU}), since in this case we obtain an evaluation of the multivariate Bollob\`as-Riodan polynomial, which is invariant under partial duality.  
\end{rem}

Let us illustrate the use of the partial duality on a simple example.

\begin{example}[Non-planar double tadpole]\label{doubletadpole}

The partial dual of a cycle of length 2 with respect to one of its edges  is the non-planar double tadpole (i.e. the non planar graph with one vertex and two edges). Thus, using the result of example \ref{2cycleex} 
\begin{equation}
\mathrm{HU}_{\mbox{\tiny non-planar}\atop
\mbox{\tiny double tadpole}}(t_{1},t_{2},\Omega_{1},\Omega_{2})=
\mathrm{HU}_{C_{2}}(t_{1},\Omega_{2},\Omega_{1},t_{2})=4\big(\Omega_{1}^{2}+t_{2}^{2}\big)t_{1}\Omega_{2}+4\big(t_{1}^{2}+\Omega_{2}^{2}\big)\Omega_{1}t_{2}.
\end{equation} 
Note that we obtain the same result if we perform the partial duality
on edge 2, since they are symmetric. Partial duality with respect to
both edges yields another cycle of length $2$, with variables all variables $\Omega$ and $t$ interchanged.
\end{example}

Before we deal with particular classes of graphs, let us show that $\mathrm{HU}_{G}$ is not identically 0, except for a particular case.

\begin{prop}
\label{nonzero}
$\mathrm{HU}_{G}$ is identically 0 only for a graph containing an isolated vertex of even degree. 
\end{prop}
\begin{proof}
  We have already seen that on a isolated vertex $\mathrm{HU}_{G}=0$
  if only if $G$ has an even number of flags. Using the
  multiplicativity over disjoint unions, it remains to show that
  $\mathrm{HU}_{G}$ is not identically zero for a graph with at least
  one edge.  To construct a monomial with a non zero coefficient, let
  us choose a spanning tree $T$ in $G$ and an edge $e\in E(T)$. The
  tree $T/(E(T)-\left\{e\right\})$ obtained by contracting all edges
  of $T$ but $e$ has two vertices $v_{1}$ and $v_{2}$. If $v_{1}$ and
  $v_{2}$ both carry an odd number of flags then set
  $A=E(T)-\left\{e\right\}$ and $B=\emptyset$.  If $v_{1}$ and
  $v_{2}$ both carry an odd number of flags then set
  $A=E(T)-\left\{e\right\}$ and $B=\left\{e\right\}$. If one of the
  vertices carries an odd number of flags and the other an even one,
  then set $A=E(T)$ and $B=\emptyset$. Then, with these choices of $A$
  and $B$, the corresponding monomial in \eqref{expansion} is always
  non zero.
\end{proof}

\begin{rem}
Since all the coefficients of the monomials of $\mathrm{HU}_{G}$ are positive as a consequence of the reduction relation, this shows that  $\mathrm{HU}_{G}(\Omega,t)=0$ is possible for $t_{e}>0$ and $\Omega_{e}>0$ if and only if $G$ contains an isolated vertex with an even number of flags. Thus, $\det A>0$ in the Gaussian integration \eqref{Gaussian2} if there is no isolated vertex of even degree.
\end{rem}

For trees, it is possible to obtain a formula  that collects the contribution of various subsets $A$ and $B$. 
\begin{prop}\label{proptree}
For a tree $T$ with flags, the first hyperbolic polynomial reads
\begin{equation}
\mathrm{HU}_{T}(\Omega,t)=\sum_{A\subset E(T)}\sum_{B\subset E(T)-A\atop (B,V(T/A))\,\mathrm{odd}}
\Bigg\{2^{|E(T)|-|A|+1}
\prod_{e\in A}\Omega_{e}(1+t_{e}^{2})\prod_{e'\in B}\Omega_{e'}t_{e'}
\hskip-.5cm\prod_{e''\in E(T)-(A\cup B)}\hskip-.5cmt_{e''}\Bigg\},\label{tree}
\end{equation} 
with $T/A$ the graph resulting from the contraction of the edges in $A$ and a graph is said to be odd if all its vertices have an odd number of attached half-lines, flags included.
\end{prop}
\begin{proof}
  If $e(T)=0$, then $T=V_{n}$ is an isolated vertex and
  $A=B=\emptyset$, so that $(B,V(T/A))=V_{n}$ and we recover
  \eqref{isolated}. If $e(T)=1$, then $T=B_{m,n}$ is a bridge with
  flags and \eqref{tree} reproduces \eqref{bridgeexample}. Let us now
  prove the result by induction on $e(T)$, singling out an edge $e$
  and using the reduction relation
  \begin{equation}
    \mathrm{HU}_{T}=\Omega_{e}\,\mathrm{HU}_{T^{e}-e}+\Omega_{e}t_{e}^{2}\,\mathrm{HU}_{T^{e}\cut
      e}+
    t_{e}\,\mathrm{HU}_{T-e}+t_{e}\Omega_{e}^{2}\,\mathrm{HU}_{T\cut e}.
  \end{equation}
  The graphs $T_{1}=T^{e}-e$ and $T_{2}=T^{e}\cut e$ are trees whereas
  $T-e=T_{3}\cup T_{4}$ and $T\cut e=T_{5}\cup T_{6}$ are disjoint
  unions of 2 trees. All the trees have less than $e(T)$ edges so that
  we may apply the induction assumption, with a sum over
  $A_{i},B_{i}\subset E(T_{i})$.

  For the first two terms, we gather terms for which $A_{1}$=$A_{2}$
  and define $A=A_{1}\cup \left\{ e\right\}$. Then, with $B=B_{1}$ or
  $B=B_{2}$, the graph $(B,V(T/A))$ is odd if only if
  $(B_{1},V(T_{1}/A_{1}))$ or $(B_{2},V(T_{2}/A_{2}))$ are and the
  powers of 2 agree,
  $2^{E(T)-|A|+1}=2^{E(T^{e}-e)-|A_{1}|+1}=2^{E(T^{e}\vee
    e)-|A_{2}|+1}$. This reproduces the terms in \eqref{tree} such
  that $e\in A$.

  In the case of $T-e$, $\mathrm{HU}_{T-e}$ factorizes as two
  independent summations over $(A_{3},B_{3})$ and $(A_{4},B_{4})$ and
  we set $A=A_{3}\cup A_{4}$ and $B=B_{3}\cup B_{4}$.  The graph
  $(B,V(T/A))$ is odd if only if $(B_{1},V(T_{1}/A_{1}))$ and
  $(B_{2},V(T_{2}/A_{2}))$ are and the powers of 2 agree,
  $2^{E(T)-|A|+1}=2^{E(T_{1})-|A_{1}|+1}2^{E(T_{2})-|A_{3}|+1}$. This
  reproduces in \eqref{tree} the terms such that $e\notin A$ and
  $e\notin B$.

  For $T\cut e$, we proceed similarly with $A=A_{5}\cup A_{6}$ and
  $B=B_{5}\cup B_{6}\cup\{ e\}$ and recover the terms in \eqref{tree}
  for which $e\notin A$ and $e\in B$.
\end{proof}

Let us illustrate the use of proposition \ref{proptree} on some simple examples.

\begin{example}[n-star tree without flags]
Consider the n-star tree $\star_{n}$ is made of one $n$-valent vertex, attached to $n$ univalent ones, all without flags. Since all the edges not in $A$ are necessarily in $B$ (otherwise the leaves yield vertices without flag),
\begin{equation}
\mathrm{HU}_{\star_{n}}(\Omega,t)=\sum_{A\subset E(\star_{n})\atop |A|+n\,\mbox{\tiny odd}}
\Big\{2^{n-|A|+2}
\prod_{e\in A}\Omega_{e}(1+t_{e}^{2})\prod_{e'\in E(\star_{n})-A}\Omega_{e'}t_{e'}\Big\}.
\end{equation} 
\end{example}

Using partial duality, one can compute the first hyperbolic polynomial for every graph made of loops attached to the vertices of a tree. Indeed, the partial duality with respect to the loops transforms the diagram into another  tree.

\begin{example}[Dumbbell]
\label{dumbbellex}
Let us consider the dumbbell graph (an edge labelled 1 attached to two vertices, each carrying a loop labelled 2 and 3).  Let us perform the partial duality with respect to the loops 2 and 3 to obtain a linear tree  with two three edges and no flag, for which proposition \ref{proptree} immediately yields
\begin{align}
\mathrm{HU}_{\mbox{\tiny linear tree}\atop\mbox{\tiny 3 edges no flag}}(\Omega,t)=
16t_{1}t_{2}\Omega_{2}^{2}t_{3}\Omega_{3}^{2}+4t_{1}\Omega_{1}^{2}\Omega_{2}(1+t_{2}^{2})\Omega_{3}(1+t_{3}^{2})
\hskip2cm{}\cr
+4t_{2}\Omega_{2}^{2}\Omega_{1}(1+t_{1}^{2})\Omega_{3}(1+t_{3}^{2})
+4t_{3}\Omega_{3}^{2}\Omega_{1}(1+t_{1}^{2})\Omega_{2}(1+t_{2}^{2}).
\end{align}
Using the partial duality $\mathrm{HU}_{\mbox{\tiny dumbbell}}(\Omega_{1},\Omega_{2},\Omega_{3},t_{1},t_{2},t_{3})=
\mathrm{HU}_{\mbox{\tiny linear tree}\atop\mbox{\tiny 3 edges no flag}}(t_{1},t_{2},\Omega_{3},t_{1},\Omega_{2},\Omega_{3})$ we get
\begin{align}
\mathrm{HU}_{\mbox{\tiny dumbbell}}(\Omega,t)=
16t_{1}\Omega_{2}t_{2}^{2}\Omega_{3}t_{3}^{2}+4t_{1}\Omega_{1}^{2}t_{2}(1+\Omega_{2}^{2})t_{3}(1+\Omega_{3}^{2})
\hskip2cm{}\cr
+4\Omega_{2}t_{2}^{2}\Omega_{1}(1+t_{1}^{2})t_{3}(1+\Omega_{3}^{2})
+4\Omega_{3}t_{3}^{2}\Omega_{1}(1+t_{1}^{2})t_{2}(1+\Omega_{2}^{2}).
\end{align}
\end{example}

Beyond trees, it is also possible to give a useful formula for cycles, i.e. a connected graph in which every vertex has valence 2.

\begin{prop}\label{propcycle}
For a cycle $C$ with $m$ flags in one face and $n$ in the other one, the first hyperbolic polynomial reads:
\begin{align}
\mathrm{HU}_{C}(\Omega,t)=\quad4\delta_{(-1)^{m},(-1)^{n}}\Big(\prod_{e\in E(C)}\Omega_{e}\Big)\sum_{A\subset E(C)\atop |A|+n\,\mbox{\tiny odd}}\Bigg\{\prod_{e'\in A}t_{e'}^{2}\Bigg\}\quad+\hskip3cm{}\cr
\sum_{A\subset E(C)\atop (A,V(C))\,\mathrm{acyclic}}
\sum_{B\subset E(C/A)\atop (B,V(C/A))\,\mathrm{odd}}
\Bigg\{2^{|E(C)|-|A|}
\prod_{e\in A}\Omega_{e}(1+t_{e}^{2})\prod_{e'\in B'}\Omega_{e'}t_{e'}
\hskip-.5cm\prod_{e''\in E(C)-(A\cup B)}\hskip-.5cmt_{e''}\Bigg\},\label{cycle}
\end{align} 
where a graph is acyclic if it does not contain a (non necessarily spanning) subgraph isomorphic to a cycle.
\end{prop}
\begin{proof}
  We prove this result by induction on the number of edges of $C$,
  starting with $e(G)=1$. In this case, $C$ is a loop with flags and
  \eqref{cycle} reduces to \eqref{loopexample}. Let us consider a
  cycle with $e(G)>1$ edges, $m$ flags on one face and $n$ flags on
  the other one and apply the reduction relation to an edge $e$,
  \begin{equation}
    \mathrm{HU}_{C}=\Omega_{e}\,\mathrm{HU}_{C^{e}-e}+\Omega_{e}t_{e}^{2}\,\mathrm{HU}_{C^{e}\vee e}+
    t_{e}\,\mathrm{HU}_{C-e}+t_{e}\Omega_{e}^{2}\,\mathrm{HU}_{C\vee e}.
  \end{equation}
  $C^{e}-e$ (resp. $C^{e}\cut e$) are cycles with $e(C)-1$ edges and
  $m$ flags on one face and $n$ flags on the other one (resp. $m+1$
  and $n+1$), so that we apply the induction assumption and express
  both of them using subsets $A'$ and $B'$ of $E(G)-\{e\}$ as in
  \eqref{cycle}. Setting $A=A'\cup \{e\}$ and $B'=B$, these terms can
  be collected and correspond to those terms in \eqref{cycle} such
  that $e\in A$. The numerical factors agree and $\big(B,V(C/A)\big)$
  is odd if only if $\big(B',V\big((C^{e}\!-\!e)/A'\big)\big)$ and
  $\big(B',V\big((C^{e}\!-\!e)/A'\big)\big)$ are because the graphs
  $\big(A,V(C)\big)$, $\big(A',V(C^{e}\!-\!e)\big)$ and
  $\Big(A',V(C^{e}\!\vee\!e)\big)$ are acyclic.

  The graphs $C-e$ and $C\cut e$ are trees, so that we may apply
  proposition \ref{proptree} to expand $\mathrm{HU}_{C-e}$ and
  $\mathrm{HU}_{C\vee e}$ using subsets $A'$ and $B'$ of
  $E(C)-\left\{e\right\}$. Setting $A=A'$ and $B=B'$, terms in
  $\mathrm{HU}_{C-e}$ correspond to terms in $\mathrm{HU}_{C}$ such
  that neither $A$ nor $B$ contains $e$. With $A=A'$ and
  $B=B'\cup\{e\}$, the expansion of $\mathrm{HU}_{C\vee e}$ reproduces
  those terms in the expansion of $\mathrm{HU}_{C}$ for which $e\notin
  A$ and $e\in B$.
\end{proof}

\begin{example}[Triangle without flags]
Consider a triangle (cycle with three edges) and no flags. Applying proposition \ref{propcycle}, we get
\begin{align}
\mathrm{HU}_{\mbox{\tiny triangle}\atop\mbox{\tiny without flag}}(\Omega,t)=
4\Omega_{1}\Omega_{2}\Omega_{3}(t_{1}^{2}+t_{2}^{2}+t_{3}^{2}+t_{1}^{2}t_{2}^{2}t_{3}^{2})
\hskip2cm{}\cr
+4\Omega_{1}(1+t_{1}^{2})t_{2}t_{3}(\Omega_{2}^{2}+\Omega_{3}^{2})
+4\Omega_{2}(1+t_{2}^{2})t_{1}t_{3}(\Omega_{1}^{2}+\Omega_{3}^{2})
+4\Omega_{3}(1+t_{3}^{2})t_{1}t_{2}(\Omega_{1}^{2}+\Omega_{2}^{2}).
\end{align}
As we perform the duality with respect to all three edges, we recover the planar banana with three edges (see example \ref{planarbananaex}), with $\Omega_{e}\leftrightarrow t_{e}$ for all edges.
\end{example}

\begin{example}[Triangle with flags]
For a triangle with one flag on each vertex, all in the same face, proposition \ref{propcycle} immediately yields
\begin{align}
\mathrm{HU}_{\mbox{\tiny triangle}\atop\mbox{\tiny with flags}}(\Omega,t)=
8t_{1}t_{2}t_{3}(1+\Omega_{1}^{2}\Omega_{2}^{2}\Omega_{3}^{2})
\hskip2cm{}\cr
+2t_{1}\Omega_{2}(1+t_{2}^{2})\Omega_{3}(1+t_{3}^{2})
+2t_{2}\Omega_{1}(1+t_{1}^{2})\Omega_{3}(1+t_{3}^{2})
+2t_{3}\Omega_{1}(1+t_{1}^{2})\Omega_{2}(1+t_{2}^{2}).\label{trianglewithflags}
\end{align}
Note that the first term in \eqref{cycle} vanishes, since there are three flags in one face and none in the other one.
\end{example}

\subsection{The second hyperbolic polynomial}

Let us now evaluate the second hyperbolic polynomial $\mathrm{HV}_{G}$ in terms of $\mathrm{HU}_{G}$, which is itself an evaluation of the graph polynomial ${\cal Q}_{G}$.

\begin{thm}
The second hyperbolic polynomial can be expressed as
\begin{align}
\mathrm{HV }_{G}=
\sum_{i}\mathrm{HU }_{G_{i}}x_{i}^{2}+
\frac{1}{2}\sum_{i\neq j}
\Big[\mathrm{HU }_{(G_{ij})^{e_{ij}}-e_{ij}}-\mathrm{HU }_{(G_{ij})^{e_{ij}}\vee e_{ij}}\Big]\,x_{i}\cdot x_{j}\cr
+\frac{1}{2}\sum_{i\neq j}
\Big[\mathrm{HU }_{(\check{G}_{ij})^{e_{ij}}-e_{ij}}-\mathrm{HU }_{(\check{G}_{ij})^{e_{ij}}\vee e_{ij}}\Big]\,x_{i}\cdot Jx_{j},
\end{align}
where $G_{i}$ is the graph obtained from $G$ by removing the flag on the corner $i$, $G_{ij}$  by joining the external corners $i$ and $j$ by an extra edge $e_{ij}$ and $\check{G}_{ij}$ by attaching an extra flag to $G_{ij}$ immediately after $i$ in counterclokwise order around the vertex $i$ is attached.
\end{thm}
\begin{proof}
  Let us isolate two external corners $i$ and $j$ and write
  \begin{equation}
    \mathrm{HV }_{G}= a_{ii}\,x_{i}^{2}+ a_{jj}\,x_{j}^{2} +2a_{ij}\,x_{i}\cdot x_{j}+2\imath b_{ij}\,x_{i}\cdot J x_{j}+\cdots\,,
  \end{equation} 
  where the dots stand for terms that vanish when $x_{i}=x_{j}=0$. To
  determine $a_{ii}$, we set $x_{k}=0$ for $k\neq i$ and integrate
  over $x_{i}$,
  \begin{equation}
    \int d^{D}x_{i}\;{\cal A}_{G}\Big|_{x_{k}=0\atop k\neq i}={\cal A}_{G_{i}}\Big|_{x_{k}=0\atop k\neq i}.
  \end{equation}
  Comparing both sides with \eqref{amplitude}, we readily get
  $a_{ii}=\mathrm{HU}_{G_{i}}$.

  Similarly, to compute $a_{ij}$, we insert an extra edge $e_{ij}$
  between the flags $i$ and $j$
  \begin{equation}
    \int d^{D}x_{i}d^{D}x_{j}\;K_{\widetilde{\Omega}_{e_{ij}}}(x_{i},x_{j}){\cal A}_{G}\Big|_{x_{k}=0\atop k\neq i,k\neq j}={\cal A}_{G_{ij}}\Big|_{x_{k}=0\atop k\neq i,k\neq j}.
  \end{equation}
  The integral is Gau\ss ian over $X=\begin{pmatrix}x_{i}\cr
    x_{j}\end{pmatrix}$
  \begin{equation}
    \int d^{D}x_{i}d^{D}x_{j}\;K_{\widetilde{\Omega}_{e_{ij}}}(x_{i},x_{j}){\cal A}_{G}\Big|_{x_{k}=0\atop k\neq i,k\neq j}=
    {\cal N}\;\int d^{2D}X\exp-\frac{1}{2}{}^{t}XAX,
  \end{equation}
  with a normalization factor
  \begin{equation}\ {\cal N}=
    \left[\frac{\Omega_{e_{ij}}(1-t_{e_{ij}}^{2})}{2\pi\theta\,t_{e_{ij}}}\right]^{D/2}\times
    \left[\frac{\prod_{e}\Omega_{e}(1-t_{e}^{2})}
      {2^{v(G)-f(G)}(2\pi\theta)^{e(G)+f(G)-v(G)}\mathrm{HU}_{G}(\Omega,t)}
    \right]^{D/2}
  \end{equation}
  and
  \begin{align}
    A=\frac{1}{\theta{\mathrm{HU}_{G}}}\begin{pmatrix}
      {\mathrm{HU}_{G}}\Omega_{e_{ij}}\big(
      t_{e_{ij}}+\frac{1}{t_{e_{ij}}}\big)+2a_{ii}&
      {\mathrm{HU}_{G}}\Omega_{e_{ij}}\big(
      t_{e_{ij}}-\frac{1}{t_{e_{ij}}}\big)+2a_{ij}\cr
      {\mathrm{HU}_{G}}\Omega_{e_{ij}}\big(
      t_{e_{ij}}-\frac{1}{t_{e_{ij}}}\big)+2a_{ji}&
      {\mathrm{HU}_{G}}\Omega_{e_{ij}}\big(
      t_{e_{ij}}+\frac{1}{t_{e_{ij}}}\big)+2a_{jj}
    \end{pmatrix}\otimes I_{D}\cr
    +\frac{1}{\theta{\mathrm{HU}_{G}}}
    \begin{pmatrix}
      0& 2b_{ij}\cr -2b_{ij}&0
    \end{pmatrix}
    \otimes\imath J.
  \end{align}
  This determinant can be expressed as $\xi^{D/2}$, with
  \begin{align}
    \xi =\left[\frac{2}{\theta\mathrm{HU}_{G}}\right]^{2}
    \Big[(\Omega_{e_{ij}}\mathrm{HU}_{G})^{2}+a_{ii}a_{jj}-a_{ij}^{2}+b_{ij}^{2}+\cr
    \Omega_{e_{ij}}\mathrm{HU}_{G}\big(
    t_{e_{ij}}+\frac{1}{t_{e_{ij}}}\big)\big(\frac{a_{ii}+a_{jj}}{2}\big)-
    \Omega_{e_{ij}}\mathrm{HU}_{G}\big(
    t_{e_{ij}}-\frac{1}{t_{e_{ij}}}\big)a_{ij} \Big].
  \end{align}
  We perform the Gau\ss ian integration over $X$ to obtain ${\cal
    A}_{G_{ij}}\Big|_{x_{k}=0\atop k\neq i,k\neq j} $ and identify
  $\mathrm{HU}_{G_{ij}}$
  \begin{align}
    \mathrm{HU}_{G_{ij}}=(\Omega_{e_{ij}})^{2}t_{e_{ij}}\mathrm{HU}_{G}+
    \frac{a_{ii}a_{jj}-a_{ij}^{2}+b_{ij}^{2}}{\mathrm{HU}_{G}}t_{e_{ij}}+\cr
    \frac{\Omega_{e_{ij}}}{2}\big[
    (t_{e_{ij}})^{2}+1\big]\big(a_{ii}+a_{jj}\big)-
    \Omega_{e_{ij}}\big[ (t_{e_{ij}})^{2}-1\big]a_{ij}.
  \end{align}
  Using the reduction relation, we identify the first term with
  $(\Omega_{e_{ij}})^{2}t_{e_{ij}}\mathrm{HU}_{G_{e_{ij}}}\!\cut
  e_{ij}$, the second with $t_{e_{ij}}\mathrm{HU}_{G_{e_{ij}}}-
  e_{ij}$ (this proves that $\mathrm{HU}_{G}$ divides\footnote{It is a
    simple case of the Dodgson condensation identities.}
  $a_{ii}a_{jj}-a_{ij}^{2}+b_{ij}^{2}$) and the sum of the last two
  terms with
  $\Omega_{e_{ij}}\mathrm{HU}_{(G_{e_{ij}})^{e_{ij}}-e_{ij}}+\Omega_{e_{ij}}(t_{e_{ij}})^{2}\mathrm{HU}_{(G_{e_{ij}})^{e_{ij}}\vee
    e_{ij}}$. Thus we have
  \begin{equation}
    a_{ij}=\frac{1}{2}\Big[\mathrm{HU }_{(\check{G}_{ij})^{e_{ij}}-e_{ij}}-\mathrm{HU }_{(\check{G}_{ij})^{e_{ij}}\vee e_{ij}}\Big].
  \end{equation}
  To compute $b_{ij}$, we use a similar method but introduce an extra
  flag on the vertex $i$ is attached to, immediately after $i$ in
  counterclockwise order. Then, we connect $i$ and $j$ with an extra
  edge $e_{ij}$ to obtain $\check{G}_{ij}$. In terms of graph
  amplitudes, this can be expressed as
  \begin{equation}
    \int d^{D}x_{i}d^{D}x_{j}d^{D}y\;K_{\widetilde{\Omega}_{e_{ij}}}(y,x_{j}){\cal V}_{3}(x_{i},y,0){\cal A}_{G}\Big|_{x_{k}=0\atop k\neq i,k\neq j}={\cal A}_{\check{G}_{ij}}\Big|_{x_{k}=0\atop k\neq i,k\neq j}.
  \end{equation}
  As before, the integral over $X=\begin{pmatrix}x_{i}\cr y\cr
    x_{j}\end{pmatrix}$ is Gau\ss ian,
  \begin{equation}
    \int d^{D}x_{i}d^{D}x_{j}d^{D}y\;K_{\widetilde{\Omega}_{e_{ij}}}(y,x_{j}){\cal V}_{3}(x_{i},y,0){\cal A}_{G}\Big|_{x_{k}=0\atop k\neq i,k\neq j}=
    {\cal N}\;\int d^{2D}X\exp-\frac{1}{2}{}^{t}XAX,
  \end{equation}
  with a normalization factor
  \begin{equation} {\cal N}=
    \left[\frac{\Omega_{e_{ij}}(1-t_{e_{ij}}^{2})}{2\pi\theta\,t_{e_{ij}}}\right]^{D/2}\times
    \left[\frac{\prod_{e}\Omega_{e}(1-t_{e}^{2})}
      {2^{v(G)-f(G)}(2\pi\theta)^{e(G)+f(G)-v(G)}\mathrm{HU}_{G}(\Omega,t)}
    \right]^{D/2}\times\frac{1}{(\pi\theta)^{D}}
  \end{equation}
  and
  \begin{align}
    A=\frac{1}{\theta{\mathrm{HU}_{G}}}\begin{pmatrix}
      2a_{ii}&0&2a_{ij}\cr 0&{\mathrm{HU}_{G}}\Omega_{e_{ij}}\big(
      t_{e_{ij}}+\frac{1}{t_{e_{ij}}}\big)&
      {\mathrm{HU}_{G}}\Omega_{e_{ij}}\big(
      t_{e_{ij}}-\frac{1}{t_{e_{ij}}}\big)\cr
      2a_{ij}&{\mathrm{HU}_{G}}\Omega_{e_{ij}}\big(
      t_{e_{ij}}-\frac{1}{t_{e_{ij}}}\big)&
      2a_{jj}+{\mathrm{HU}_{G}}\Omega_{e_{ij}}\big(
      t_{e_{ij}}+\frac{1}{t_{e_{ij}}}\big)
    \end{pmatrix}\otimes I_{D}\cr
    +\frac{1}{\theta{\mathrm{HU}_{G}}}
    \begin{pmatrix}
      0&2{\mathrm{HU}_{G}}&2b_{ij}\cr -2{\mathrm{HU}_{G}}&0&0\cr
      -2b_{ij}&0&0
    \end{pmatrix}
    \otimes\imath J
  \end{align}
  Its determinant is $\det A=\xi^{D/2}$ with
  \begin{align}
    \xi=\frac{8}{\theta^{3}\mathrm{HU}_{G}}\Omega_{e_{ij}}\big(
    t_{e_{ij}}-\frac{1}{t_{e_{ij}}}\big)b_{ij}+
    \xi_{1}\Omega_{e_{ij}}\big(
    t_{e_{ij}}+\frac{1}{t_{e_{ij}}}\big)+\xi_{2}\Omega_{e_{ij}}^{2}+\xi_{3},
  \end{align}
  with $\xi_{1}$, $\xi_{2}$ and $\xi_{3}$ independent of
  $\Omega_{e_{ij}}$ and $t_{e_{ij}}$.

  We perform the Gau\ss ian integration over $X$ to obtain ${\cal
    A}_{\check{G}_{ij}}\Big|_{x_{k}=0\atop k\neq i,k\neq j} $ and
  identify the terms in $\Omega_{e_{ij}}\big( t_{e_{ij}}^{2}-1\big)$
  to obtain
  \begin{align}
    2b_{ij}=\mathrm{HU
    }_{(\check{G}_{ij})^{e_{ij}}-e_{ij}}-\mathrm{HU
    }_{(\check{G}_{ij})^{e_{ij}}\vee e_{ij}},
  \end{align}
  which proves our expression for the antisymmetric part of
  $\mathrm{HV}_{G}$.

  Let us note that up to a change of sign, we could have attached the
  extra flag before $i$ or on the vertex $j$ is attached to.
\end{proof}
As a consequence, the second hyperbolic polynomial is also invariant under partial duality.
\begin{cor}
The second hyperbolic polynomial transforms under partial duality as
\begin{equation}
\mathrm{HV}_{G^{A}}(\Omega_{A},t_{A},x)=
\mathrm{HV}_{G}(\Omega,t,x),
\end{equation}
with 
\begin{equation}
\left\{
\begin{array}{llcc}
[\Omega_{A}]_{e}=t_{e},& [t_{A}]_{e}=\Omega_{e}&\mbox{for }&e\in A,\cr
[\Omega_{A}]_{e}=\Omega_{e},& [t_{A}]_{e}=t_{e}&\mbox{for }&e\notin A.\cr
\end{array}
\right.
\end{equation}
The variables $x$ attached to the flags are left unchanged.
\end{cor}
\begin{proof}
  This follows immediately from the invariance of $\mathrm{HU}_{G}$
  and the fact that partial duality commutes with the operations we
  performed on the flags.
\end{proof}

Let us illustrate the computation of $\mathrm{HV}_{G}$ on some simple examples.
\begin{example}[Bridge]
Consider the graph with a single edge, two vertices, each with one flag, labeled 1 and 2. Thus, $G_{1}$ and $G_{2}$ are graphs with one edge, two vertices and a single flag, $G_{12}$ is a banana with two edges and no flag and $\check{G}_{12}$ is a banana with a single flag. This immediately leads to
\begin{equation}
\mathrm{HV}_{G}=2\Omega(t^{2}+1)(x_{1}^{2}+x_{2}^{2})+4\Omega(t^{2}-1)x_{1}x_{2}.
\end{equation}
Since we also have $\mathrm{HU}_{G}=4t$, the amplitude reads
\begin{multline}
{\cal A}_{G}(\Omega,x_{1},x_{2})=\cr\int d\alpha\,\left[\frac{\Omega(1-t^{2})}{2\pi\theta t}\right]^{D/2}
\exp-\frac{\Omega}{2\theta}\left\{\big(t+\frac{1}{t}\big)\big(x_{1}^{2}+x_{2}^{2}\big)+2\big(t-\frac{1}{t}\big)x_{1}x_{2}\right\}.
\end{multline}
To compare this amplitude with the corresponding one in the commutative theory (see proposition \ref{commlimrop}), recall that we are working with an oscillator of frequency $\frac{2\Omega}{\theta}$. Therefore, we have to substitute $\Omega\rightarrow\frac{\theta\Omega}{2}$,
\begin{multline}
{\cal A}_{G}\big(\frac{\theta\Omega}{2},x_{1},x_{2}\big)=\cr\int d\alpha\,\left[\frac{\Omega}{2\pi}\times \frac{(1-t^{2})}{2t}\right]^{D/2}
\exp-\frac{\Omega}{4}\left\{\big(t+\frac{1}{t}\big)\big(x_{1}^{2}+x_{2}^{2}\big)+2\big(t-\frac{1}{t}\big)x_{1}x_{2}\right\},
\end{multline}
which is nothing but the Mehler kernel of an oscillator of frequency $\Omega$, as it should since there is no integration on the external flags. Strictly speaking, the commutative amplitude is recovered after the limit $\theta\rightarrow 0$, but the latter is trivial since the $\theta$-dependence drops from ${\cal A}_{G}\big(\frac{\theta\Omega}{2},t,x_{1},x_{2}\big)$.
\end{example}

\begin{example}[Tadpole]

Let us now perform the partial duality on the unique edge of the bridge treated in the last example. We obtain $G^{e}$ which is a loop with a single vertex and one flag in each of its two faces. The corresponding amplitude reads
\begin{multline}
{\cal A}_{G^{e}}(\Omega,t,x_{1},x_{2})=\cr\int d\alpha\,\left[\frac{(1-t^{2})}{(2\pi\theta)^{2}}\right]^{D/2}
\exp-\frac{t}{2\theta}\left\{\big(\Omega+\frac{1}{\Omega}\big)\big(x_{1}^{2}+x_{2}^{2}\big)+2\big(\Omega-\frac{1}{\Omega}\big)x_{1}x_{2}\right\}.
\end{multline} 
Let us note that we exchanged $\Omega$ and $t$ in the hyperbolic polynomials, but not in the prefactor. It is also worthwhile to point out that we have traded the simple graph with two 2-valent vertices for a more complicated one with one 4-valent vertex. While a direct evaluation of the former is straightforward, it becomes more complicated for the latter, because of the structure of the 4-valent vertex.

To compare it with the commutative case, we substitute   $\Omega\rightarrow\frac{\theta\Omega}{2}$ and take the limit $\theta\rightarrow 0$, so that
\begin{equation}
{\cal A}_{G}\big(\frac{\theta\Omega}{2},x_{1},x_{2}\big)=\int d\alpha\,\left[\frac{(1-t^{2})}{(2\pi\theta)^{2}}\right]^{D/2}
\exp-\frac{t\Omega}{4}\big(x_{1}+x_{2}\big)^{2}
\exp-\frac{t}{4\Omega\theta}\big(x_{1}-x_{2}\big)^{2}.
\end{equation}
Then, using 
\begin{equation}
\lim_{\sigma\rightarrow 0}\frac{1}{(2\pi\sigma^{2})^{D/2}}\exp-\frac{(x_{1}-x_{2})^{2}}{2\sigma^{2}}=\delta^{D}(x_{1}-x_{2}),
\end{equation}
we recover
\begin{equation}
\lim_{\theta\rightarrow0}{\cal A}_{G}\big(\frac{\theta\Omega}{2},x_{1},x_{2}\big)=\delta^{D}(x_{1}-x_{2})\int d\alpha\,\left[\frac{\Omega}{2\pi}\times \frac{(1-t^{2})}{2t}\right]^{D/2}
\exp-\frac{t\Omega}{4}\big(x_{1}+x_{2}\big)^{2}.
\end{equation}
This is indeed the commutative amplitude, since the 4-valent vertex reduces in the limit $\theta\rightarrow 0$ to a product of Dirac distributions (see \eqref{vertexcomm}).

\end{example}

\begin{example}[Sunset] Consider the graph with two vertices related
  by three edges labeled 1,2 and 3 and one flag on each vertex, both
  in the face bounded by the edges 1 and 3. It is simpler to compute
  the hyperbolic polynomial of its dual, which is a cycle with three
  edges and two faces, each broken by a flag on the vertex not
  adjacent to the edge 2, All the graphs involved in the expression of
  the hyperbolic polynomial are cycles or trees with flags so that an
  immediate application of propositions \ref{proptree} and
  \ref{propcycle} provides us with
\begin{align}
\mathrm{HU}_{\mbox{\tiny cycle with 3 edges}\atop\mbox{\tiny 2 broken faces}}=
4\Omega_{1}\Omega_{2}\Omega_{3}\big[
1+t_{1}^{2}t_{2}^{2}+t_{1}^{2}t_{3}^{2}+t_{2}^{2}t_{3}^{2}\big]+4\Omega_{1}t_{2}t_{3}(1+t_{1}^{2})(\Omega_{2}^{2}+\Omega_{3}^{2})\cr
+4\Omega_{2}t_{1}t_{3}(1+t_{2}^{2})(\Omega_{1}^{2}+\Omega_{3}^{2})+4\Omega_{3}t_{1}t_{2}(1+t_{3}^{2})(\Omega_{1}^{2}+\Omega_{2}^{2})
\end{align} 
and
\begin{gather}
\mathrm{HV}_{\mbox{\tiny cycle with 3 edges}\atop\mbox{\tiny 2 broken faces}}(x_{1},x_{2})=
\big[x_{1}^{2}+x_{2}^{2}\big]\big[8t_{1}t_{2}t_{3}(\Omega_{2}^{2}+\Omega_{1}^{2}\Omega_{3}^{2})
+2t_{1}\Omega_{2}\Omega_{3}(1+\Omega_{1}^{2})(1+t_{2}^{2})(1+t_{3}^{2})\cr
+2t_{2}\Omega_{1}\Omega_{3}(1+\Omega_{2}^{2})(1+t_{1}^{2})(1+t_{3}^{2})
+2t_{3}\Omega_{1}\Omega_{2}(1+\Omega_{3}^{2})(1+t_{1}^{2})(1+t_{2}^{2})\big]\cr
+x_{1}\cdot x_{2}\,\big[
16t_{1}t_{2}t_{3}(\Omega_{1}^{2}\Omega_{3}^{2}-1)
+4t_{1}(1+t_{2}^{2})(1+t_{3}^{2})\Omega_{2}\Omega_{3}(\Omega_{1}^{2}-1)\cr
\quad+4t_{2}(1+t_{1}^{2})(1+t_{3}^{2})\Omega_{1}\Omega_{3}(\Omega_{2}^{2}-1)
+4t_{3}(1+t_{1}^{2})(1+t_{2}^{2})\Omega_{1}\Omega_{2}(\Omega_{3}^{2}-1)\big]\cr
+x_{1}\cdot Jx_{2}\,\big[4(1+t_{1}^{2})t_{2}t_{3}\Omega_{1}(\Omega_{3}^{2}-\Omega_{2}^{2})\cr
\quad+4(1+t_{2}^{2})t_{1}t_{3}\Omega_{2}(\Omega_{3}^{2}-\Omega_{1}^{2})
+4(1+t_{3}^{2})t_{1}t_{2}\Omega_{3}(\Omega_{2}^{2}-\Omega_{1}^{2})\big].
\end{gather} 
We readily obtain the hyperbolic polynomials  of the sunset by interchanging $\Omega_{e}$ and $t_{e}$ for all edges, 
\begin{align}
\mathrm{HU}_{\mbox{\tiny sunset}}=
4t_{1}t_{2}t_{3}\big[
1+\Omega_{1}^{2}\Omega_{2}^{2}+\Omega_{1}^{2}\Omega_{3}^{2}+\Omega_{2}^{2}\Omega_{3}^{2}\big]+4t_{1}(t_{2}^{2}+t_{3}^{2})\Omega_{2}\Omega_{3}(1+\Omega_{1}^{2})\cr
+4t_{2}(t_{1}^{2}+t_{3}^{2})\Omega_{1}\Omega_{3}(1+\Omega_{2}^{2})+4t_{3}(t_{1}^{2}+t_{2}^{2})\Omega_{1}\Omega_{2}(1+\Omega_{3}^{2})
\end{align} 
and
\begin{gather}
\mathrm{HV}_{\mbox{\tiny sunset}}(x_{1},x_{2})=
\big[x_{1}^{2}+x_{2}^{2}\big]\big[8\Omega_{1}\Omega_{2}\Omega_{3}(t_{2}^{2}+t_{1}^{2}t_{3}^{2})
+2\Omega_{1}t_{2}t_{3}(1+t_{1}^{2})(1+\Omega_{2}^{2})(1+\Omega_{3}^{2})\cr
\quad+2\Omega_{2}t_{1}t_{3}(1+t_{2}^{2})(1+\Omega_{1}^{2})(1+\Omega_{3}^{2})
+2\Omega_{3}t_{1}t_{2}(1+t_{3}^{2})(1+\Omega_{1}^{2})(1+\Omega_{2}^{2})\big]\cr
+x_{1}\cdot x_{2}\,\big[
16\Omega_{1}\Omega_{2}\Omega_{3}(t_{1}^{2}t_{3}^{2}-1)
+4\Omega_{1}(1+\Omega_{2}^{2})(1+\Omega_{3}^{2})t_{2}t_{3}(t_{1}^{2}-1)\cr
\quad+4\Omega_{2}(1+\Omega_{1}^{2})(1+\Omega_{3}^{2})t_{1}t_{3}(t_{2}^{2}-1)
+4\Omega_{3}(1+\Omega_{1}^{2})(1+\Omega_{2}^{2})t_{1}t_{2}(t_{3}^{2}-1)\big]\cr
+x_{1}\cdot Jx_{2}\,\big[4(1+\Omega_{1}^{2})\Omega_{2}\Omega_{3}t_{1}(t_{3}^{2}-t_{2}^{2})\cr
\quad+4(1+\Omega_{2}^{2})\Omega_{1}\Omega_{3}t_{2}(t_{3}^{2}-t_{1}^{2})
+4(1+\Omega_{3}^{2})\Omega_{1}\Omega_{2}t_{3}(t_{2}^{2}-t_{1}^{2})\big].
\end{gather}
In the commutative limit,  we keep only the lowest order terms in $\Omega$ in the hyperbolic polynomials and we recover the product of three independent Mehler kernels for the amplitude. Moreover, if we denote by $a$ (resp. $b$, $c$) the coefficient of the term in $(x_{1}^{2}+x_{2}^{2})$ (resp. half of the coefficient of $x_{1}x_{2}$, half of the coefficient of $x_{1}Jx_{2}$), then the Dodgson condensation identity $a^{2}-b^{2}+c^{2}=\mathrm{HU}_{\mbox{\tiny sunset}}\mathrm{HU}_{\mbox{\tiny 3-banana}}$ is obeyed.
\end{example}

\begin{example}[3-star tree with flags]
We compute the hyperbolic polynomials for the 3-star tree is made of one trivalent vertex, attached to 3 univalent ones, each with one flag. The first hyperbolic polynomial results from a direct application of proposition \ref{proptree}
\begin{align}
\mathrm{HU}_{\mbox{\tiny 3-star tree}\atop\mbox{\tiny with flags}}=2\Omega_{1}\Omega_{2}\Omega_{3}(1+t_{1}^{2})
(1+t_{2}^{2})(1+t_{3}^{2})\qquad\qquad{}\cr
+8\Omega_{1}(1+t_{1}^{2})t_{2}t_{3}+8\Omega_{2}(1+t_{2}^{2})t_{1}t_{3}+8\Omega_{3}(1+t_{3}^{2})t_{1}t_{2}.
\end{align}
All the graphs involved in the computation of the second hyperbolic polynomial reduce to trees and cycles after a single use of the reduction relation, so that propositions \ref{proptree} and \ref{propcycle} yield
\begin{gather}
\mathrm{HV}_{\mbox{\tiny 3-star tree}\atop\mbox{\tiny with flags}}(x_{1},x_{2},x_{3})=
x_{3}^{2}\,\big[8t_{1}t_{2}t_{3}\Omega_{3}^{2}+4\Omega_{1}^{2}t_{1}(1+t_{2}^{2})\Omega_{2}(1+t_{3}^{2})\Omega_{3}\cr
+4\Omega_{2}^{2}t_{2}(1+t_{1}^{2})\Omega_{1}(1+t_{3}^{2})\Omega_{3}+4\Omega_{3}^{2}t_{3}(1+t_{1}^{2})\Omega_{1}(1+t_{2}^{2})\Omega_{2}
\big]\cr
+x_{1}\cdot x_{2}\,\big[8(1-t_{1}^{2})(t_{2}^{2}-1)\Omega_{1}\Omega_{2}t_{3}\big]+\mbox{2 cyclic permutations}\cr
+x_{1}\cdot Jx_{2}\,\big[4(1-t_{1}^{2})(1-t_{2}^{2})\Omega_{1}\Omega_{2}\Omega_{3}\big]+\mbox{2 cyclic permutations}.\cr
\end{gather}

\end{example}

\label{second}

\section{Various limiting cases}
\label{sec:vari-limit-cases}

\subsection{The critical model \texorpdfstring{$\Omega=1$}{Omega = 1}}
\label{sec:crit-model-omeg1}
When we set $\Omega_{e}=1$ for all edges, the hyperbolic polynomial $\mathrm{HU}_{G}$ can be factorized over the faces  of $G$ (i.e. the connected components of the boundary).  Before we give a combinatorial proof of a general factorization theorem at $\Omega=1$, let us present a heuristic derivation of this result for ribbon graphs without flags, based on the matrix basis.

The Moyal algebra of Schwartz functions on ${\Bbb R}^{D}$ is isomorphic to an algebra of infinite dimensional matrices $M_{pq}$ whose indices $p,q$ are elements of ${\Bbb N}^{D/2}$ and whose entries decrease faster than any polynomials in $p,q$. Using this isomorphism $\phi\rightarrow M$, the interaction \eqref{intstar} can be written as
\begin{equation}
S_{\mbox{\tiny int}}[M]=(2\pi\theta)
\sum_{n\geq1}\frac{g_{n}}{n}\,\mbox{Tr}\big[M^{n}\big],
\end{equation}
which is the standard interaction familiar from matrix models. The associated vertex reads
\begin{equation}
{\cal V}_{n}(p_{i},q_{1},p_{2},q_{2},\dots p_{n},q_{n})=(2\pi\theta)\,\delta_{q_{1},p_{2}}\delta_{q_{2},p_{3}}\cdots\delta_{q_{n-1},p_{n}}\delta_{q_{n},p_{1}}.
\end{equation}
The quadratic term reads
\begin{equation}
S_{\mbox{\tiny 0}}[M]=
\frac{1}{2}\sum_{p,q,r,s}M_{pq}\Delta_{pq,rs}M_{rs}.
\end{equation}
In the critical case $\Omega=1$,
\begin{equation}
\Delta_{pq,rs}=(2\pi\theta)\delta_{ps}\delta_{qr}\,\frac{4(|p|+|q|+1)}{\theta},
\end{equation}
where $|p|=p_{1}+\cdots+_{D/2}$ for any multi-index $p=(p_{1},\dots,p_{D/2})\in{\Bbb N}^{D/2}$.

Because of the Kronecker symbols $\delta$, the multi-indices are identical around each faces (as in ordinary matrix models), so that the amplitude factorizes over the faces for a graph without flags,
\begin{equation}
{\cal A}_{G}=\int\prod_{e}d\alpha_{e}\,\frac{1}{(2\pi\theta)^{e(G)-v(G)}}\prod_{\sigma\,\mbox{\tiny faces of }G}
\prod_{e\,\mbox{\tiny edges}\atop \mbox{\tiny bounding }\sigma}\Big\{\sum_{i_{e}\in{\Bbb N}^{D/2}}\exp-\frac{4\alpha_{e}}{\theta}\big(|p_{e}|+\frac{1}{2}\big)\Big\}.
\end{equation}
Summing up the geometric series and expressing the amplitudes in terms of $t_{e}=\tanh\frac{2\alpha_{e}}{\theta}$, we obtain 
\begin{equation}
{\cal A}_{G}=\int\prod_{e}d\alpha_{e}\,\bigg[\frac{1}{(2\pi\theta)^{e(G)-v(G)}}\times\prod_{e}\frac{1-t_{e}}{1+t_{e}}
\prod_{\sigma\,\mbox{\tiny faces of }G}
\Big(1-\prod_{e\,\mbox{\tiny edges}\atop \mbox{\tiny bounding }\sigma}\frac{1-t_{e}}{1+t_{e}}\Big)^{-1}\bigg]^{D/2}.
\end{equation}
Then, identifying a face $\sigma$ of $G$ with a vertex $v^{*}$ of $G^{*}$,
\begin{equation}
\prod_{e\,\mbox{\tiny edges}\atop \mbox{\tiny bounding }\sigma}\big(1+t_{e}\big)
-\prod_{e\,\mbox{\tiny edges}\atop \mbox{\tiny bounding }\sigma}\big(1-t_{e}\big)
=2\sum_{A\subset E_{v^{*}},\atop|A|\;\mbox{\tiny odd}}\;\prod_{e\in A}t_{e},
\end{equation}
with $E_{v^{*}}$ as the set of half-edges of $G^{*}$ incident to $v^{*}$. Comparing with the general expression of the amplitude $\eqref{amplitude}$, this suggests that
\begin{equation}
\mathrm{HU}_{G}(1,t)=2^{v(G^{\ast})}\prod_{v^{\ast}\in V(G^{\ast})}
\Big\{\sum_{A\subset E_{v^{*}},\atop|A|\;\mbox{\tiny odd}}\;\prod_{e\in A}t_{e}\Big\}.\label{eq:FactHU}
\end{equation}

\begin{example}[Dumbbell]
Let us consider the dumbbell graph (an edge labelled 1 attached to to vertices, each carrying a loop labelled 2 and 3). The graph has 3 faces and we get
\begin{equation}
\mathrm{HU}_{\mbox{\tiny dumbbell}}(1,t)=8t_{2}t_{3}\big[2t_{1}(1+t_{2}t_{3})+(1+t_{1}^{2})(t_{2}+t_{3})\big].
\end{equation}
\end{example}

Let us now prove the factorization of $\HU$ at $\Omega =1$ in a
completely combinatorial way. To this aim, we will use the bijections
introduced in section \ref{sec:biject-betw-class}. Moreover, the
polynomial $\HU$ can be extended to ribbon graphs with flags and we
show that the factorization (\ref{eq:FactHU}) holds in this
case too.

\paragraph{Statement of the problem}
\label{sec:statement-problem}

Via the $x$-space representation, we
computed the parametric representation of the Grosse-Wulkenhaar model, see section \ref{sec:feynm-ampl-grosse}. This representation
involves a new ribbon graph invariant $\cQ$, see
equation (\ref{eq:QDef}). In fact, this is only a special evaluation $\HU$
of $\cQ$ which is used in the Feynman amplitudes:
\begin{align}
  \HU(G;\bt,\bO)=&\cQ(G;\bt,\bO,\bt^{2}\bO,\bt\bO^{2},\br)\label{eq:HUQ}
\end{align}
with $r_{2n}=0$ and $r_{2n+1}=2$. Then, with a slight abuse of
notation, and using definition \ref{def:ColGraphs}, the polynomial $\HU$ can be written:
\begin{align}
  \HU(G;\bt,\bO)=&\sum_{A\subset
      E(G)}t^{A}\,\Omega^{A^{c}}\sum_{B\in\COdd(G^{A^{c}})}\big(t^{B\cap
      A^{c}}\big)^{2}\big(\Omega^{B\cap A}\big)^{2}.
\end{align}
Note that if $G$ is a ribbon graph with flags, $\HU$ is also
well-defined.

On another side, we computed the parametric representation of the
\emph{critical}  ($\Omega=1$) Grosse-Wulkenhaar model via the matrix
base. It involves the following polynomial, see (\ref{eq:FactHU}) and definition \ref{def:ColCuttGraphs}:
\begin{align}
  U(G;\bt)\defi \sum_{H\in\COddF(G^{\star})}t^{H}.
\end{align}
Uniqueness of the parametric representation implies
\begin{align}
  \HU(G;\bt,\mathbf{1})=&U(G;\bt).\label{eq:ToProveBij}
\end{align}
Our task is now to give a bijective proof of \eqref{eq:ToProveBij}. To
this aim, given a ribbon graph $G$ with flags, we are going to present a
bijection $\chi_{G}$ between the colored odd cutting subgraphs of $G^{\star}$
and the colored odd subgraphs of all the partial duals of $G$. Finally
the monomial in $\HU$ corresponding to a subgraph $g$ will be proven
to be equal to the monomial of $\chi_{G}(g)$ in $U$.

\paragraph{A bijection between colored odd subgraphs}
\label{sec:biject-betw-parity}

\begin{lemma}\label{lem:bijection}
  Let $G$ be an orientable ribbon graph with flags. For any total order $<$ on the set
  $E(G)$ of edges of $G$, there is a bijection $\chi_{G}$ between $P\defi\bigcup_{S\subset E(G)}\COdd(G^{S})$
  and $\COddF(G^{\star})$.
\end{lemma}
Before entering into the proof of lemma \ref{lem:bijection}, let us
first give a preliminary definition:
\begin{defn}[Restrictions]\label{def:restr}
  Let $G$ be a ribbon graph with flags. For any $E'\subset E(G)$,
  the restriction of the map $\chi_{G}$ to $\COdd(G^{E'})$ is
  denoted by $\chi_{G,E'}:\COdd(G^{E'})\to\COddF(G^{\star})$.
\end{defn}
\begin{proof}
  We first explain how the map $\chi_{G}$ is defined. Let $G$ be a
  colored ribbon graph with flags. Let $g\in\bigcup_{S\subset
    E(G)}\COdd(G^{S})$ be a colored odd subgraph of a partial dual of
  $G$, say $G^{E'}$ for $E'\subset E(G)$. The subgraph
  $\check{g}\defi\chi_{G}(g)\in\COddF(G^{\star})$ has edges in
  $E(g)\cap E'$ and flags in $E'^{c}$. Here is how it is constructed
  from $g$.

  Each of the maps $\chi_{G,E'}$ is defined as the
  composition of $|E'^{c}|$ maps that we describe now. In section
  \ref{sec:color-cutt-subgr}, we introduced bijections
  \begin{align}
    \chi_{G}^{\set e}:\COddF(G)\bij\COddF(F^{\set e}).
  \end{align}
  We saw that given any flag-set $F'$ of $G$, these maps restrict to
  bijections
  \begin{align}
    \chi_{G}^{\set e}:\COddF(G)\flag F'\bij\COddF(G)\flag F'_{e}.
  \end{align}
  Given any order on $E(G)$, we can write
  $E'^{c}\fide\set{e_{1},\dotsc,e_{|E'^{c}|}}$. Then we define
  \begin{align}
    \chi_{G,E'}\defi\chi_{G^{E(G)\setminus\set{e_{|E'^{c}|}}}}^{\set{e_{|E'^{c}|}}}\circ\dotsm\circ\chi_{G^{\set{e_{1}}}}^{\set{e_{2}}}\circ\chi_{G}^{\set{e_{1}}}.
  \end{align}
  This map is well defined and is a bijection from $\COdd(G^{E'})$ to
  $\COddF(G^{\star})\flag E'^{c}$, as shown by the following diagram:
  \begin{center}
    \begin{tikzpicture}[>=stealth]
      \definecolor{fabgray}{gray}{.6}; \foreach \u in {6} { \node (un)
        at (0,0) [shape=rectangle] {$\COdd(G^{E'})$}; \node (deux) at
        (\u,0) [shape=rectangle]
        {$\COddF(G^{E'\cup\set{e_{1}}})\flag\set{e_{1}}$}; \node
        (trois) at (2*\u,0) [shape=rectangle]
        {$\COddF(G^{E'\cup\set{e_{1},e_{2}}})\flag\set{e_{1},e_{2}}$};
        \node (quatre) at (2*\u,-.5*\u) [shape=circle] {$\vdots$};
        \node (cinq) at (2*\u,-\u) [shape=rectangle]
        {$\COddF(G^{\star})\flag E'^{c}$}; \draw [thick,<->,fabgray]
        (un.east)--(deux.west) node [above,text width=3cm,text
        centered,midway] {\color{black} $\chi_{G}^{\set{e_{1}}}$};
        \draw [thick,<->,fabgray] (deux.east)--(trois.west) node
        [above,text width=3cm,text centered,midway] {\color{black}
          $\chi_{G^{\set{e_{1}}}}^{\set{e_{2}}}$}; \draw
        [thick,<->,fabgray] (trois.south)--(quatre.north) node
        [right,text width=2cm,text centered,midway] {\color{black}
          $\chi_{G^{\set{e_{1},e_{2}}}}^{\set{e_{3}}}$}; \draw
        [thick,<->,fabgray] (quatre.south)--(cinq.north) node
        [right,text width=2cm,text centered,midway] {\color{black}
          $\chi_{G^{E(G)\setminus\set{e_{|E'^{c}|}}}}^{\set{e_{|E'^{c}|}}}$};
        \draw [thick,<->,fabgray] (un.south)--(cinq.west) node
        [below,text width=3cm,text height=.5cm,text centered,midway]
        {\color{black} $\chi_{G,E'}$}; }
    \end{tikzpicture}
  \end{center}
  This proves lemma \ref{lem:bijection}.
\end{proof}

\paragraph{Factorization of \textbf{$\HU$}}
\label{sec:factorization-hu}
Let us define the monomials of $\HU$ (for $\bO\equiv 1$) and $U$ by
\begin{align}
  \HU(G;\bt,\mathbf{1})\fide&\ \sum_{g\in P}\cM_{\HU}(G;g),\\
  U(G;\bt)\fide&\ \sum_{h\in\COddF(G^{\star})}\cM_{U}(G^{\star};h).
\end{align}
Let $g\in\COdd(G^{A^{c}})$, $\chi_{G}(g)\in\COddF(G^{\star})\flag
A$. Moreover $E(\chi_{G}(g))=E(g)\cap A^{c}$. Thus
$\cM_{U}(G^{\star};\chi_{G}(g))=t^{A}(t^{E(g)\cap
  A^{c}})^{2}=\cM_{\HU}(G;g)$. This implies
\begin{align}
  \HU(G;\bt,\mathbf{1})=&\ \sum_{g\in P}\cM_{\HU}(G;g)\\
  =&\ \sum_{g\in
    P}\cM_{U}(G^{\star};\chi_{G}(g))=\sum_{g'\in\COddF(G^{\star})}\cM_{U}(G^{\star};g')\\
  =&\ U(G^{\star};\bt).
\end{align}

\begin{example}[Triangle with flags]
Consider the triangle with one flag on each vertex, all in the same face. In this case, one face has an even number of flags while the other has an odd number, which yields
\begin{equation}
\mathrm{HU}_{\mbox{\tiny triangle}\atop \mbox{\tiny 3 flags}}(1,t)=4\big[t_{1}+t_{2}+t_{3}+t_{1}t_{2}t_{3}\big]\big[1+t_{1}t_{2}+t_{1}t_{3}+t_{2}t_{3}\big],
\end{equation}
in accordance with  \eqref{trianglewithflags}.
\end{example}

\subsection{An algorithm for computing \texorpdfstring{$\mathrm{HU}_{G}(\Omega,t)$}{HU} based on the critical model}

The previous factorization over faces of $G$ provides us with a useful
algorithm to compute $\mathrm{HU}_{G}(\Omega,t)$, for any ribbon graph
with flags. $\mathrm{HU}_{G}(1,t)$ has indeed the same monomials in $t$ as $\mathrm{HU}_{G}(\Omega,t)$: all its coefficients are positive and no cancellation is possible. We only have to write each of the coefficient of each monomial in $t$ as a polynomial in $\Omega$. To proceed, we first determine the monomials in $\mathrm{HU}_{G}(1,t)$ by expanding
\begin{equation}
\mathrm{HU}_{G}(1,t)=2^{v(G^{\ast})}\prod_{v^{\ast}\in V(G^{\ast})}
\Big\{\sum_{A\subset E_{v^{*}},\atop|A|\;\mbox{\tiny odd}}\;\prod_{e\in A}t_{e}\Big\}.
\end{equation}
Then, for each monomial (discarding the prefactor)
\begin{itemize}
\item
perform the partial duality with respect to the set $A$ of edges with an even power of $t_{e}$ and multiply the monomial by $\prod_{e}\Omega_{e}$,  
\item
cut in $G^{A}$ the edges with a factor $t_{e}^{2}$ (edges in $A\cap B$) and delete those with $t_{e}^{0}$ (edges in $A\cap B^{c}$),
\item
sum over all possibilities of cutting the edges not in $A$ , with a factor $\Omega_{e}^{2}$, or deleting, with a factor 1,
\item
multiply by $2^{v(G^{A})}$.
\end{itemize}
At the end, it is useful to check the result by evaluating it at $\Omega=1$. The interest of this algorithm is that we are performing the operations only on the subsets $A$ and $B$ that are admissible, in contradistinction with the general expansion formula   \eqref{expansion}, where the admissibility can be tested only after having performed the partial duality and the cuts. Therefore, we avoid non admissible sets right from the beginning.

\begin{example}[Non-planar 3-banana]\label{nonplanarbananaex}
In the case of the non planar banana, the critical model yields
\begin{multline}
\mathrm{HU}_{\mbox{\tiny non planar}\atop \mbox{\tiny 3-banana}}(1,t)=\cr2\big[8t_{1}t_{2}t_{3}+2t_{1}(1+t_{2}^{2})(1+t_{3}^{2})
+2t_{2}(1+t_{1}^{2})(1+t_{3}^{2})+2t_{3}(1+t_{1}^{2})(1+t_{3}^{2})\big].
\end{multline}
Applying the algorithm, we deduce
\begin{gather}
\mathrm{HU}_{\mbox{\tiny non planar}\atop \mbox{\tiny 3-banana}}(\Omega,t)= 4\Big[t_{1}t_{2}t_{3}\big[\Omega_{1}^{2}+\Omega_{2}^{2}+\Omega_{3}^{2}+\Omega_{1}^{2}\Omega_{2}^{2}\Omega_{3}^{2}\big]
+t_{1}\Omega_{2}\Omega_{3}\big[t_{2}^{2}+t_{3}^{2}+\Omega_{1}^{2}(t_{2}^{2}+t_{3}^{2})\big]\cr
+t_{2}\Omega_{1}\Omega_{3}\big[t_{1}^{2}+t_{3}^{2}+\Omega_{2}^{2}(t_{1}^{2}+t_{3}^{2})\big]
+t_{3}\Omega_{1}\Omega_{2}\big[t_{1}^{2}+t_{2}^{2}+\Omega_{3}^{2}(t_{1}^{2}+t_{2}^{2})\big]
\Big].
\end{gather}
\end{example}

\subsection{The noncommutative heat kernel
  limit\texorpdfstring{ $\Omega\rightarrow 0$}{: Omega to zero}}

In this section, we study the amplitude \eqref{defamplitude} and the first hyperbolic polynomial $\mathrm{HU}_{G}(\Omega,t)$ in the limit of vanishing oscillator frequency. In order to avoid a lengthy discussion of the second hyperbolic polynomial, we restrict ourselves to graphs without flags. The general case can be treated along the same lines. Without further loss of generality, we also assume the graph to be connected. 

In the limit $\Omega\rightarrow 0$, the Mehler kernel reduces to the heat kernel,
\begin{equation}
\lim_{\widetilde{\Omega}\rightarrow 0}{\cal K}_{\widetilde{\Omega}}(x,y)=
{\cal K}_{0}(x,y)=
\frac{1}{(4\pi)^{D/2}}\int_{1/\Lambda^{2}}^{\infty}\frac{d\alpha}{\alpha^{D/2}}
\exp-\frac{(x-y)^{2}}{4\alpha}.\label{heatkernel}
\end{equation}
Notice that ${\cal K}_{0}(x,y)$ only depends on $x-y$ , so that it is invariant under translations, ${\cal K}_{0}(x+a,y+a)={\cal K}_{0}(x,y)$. Because the heat kernel and the vertex are both invariant under translations, the integrand in \eqref{defamplitude} only depends on $2e(G)-1$ variables for a connected graph without flags. Therefore the integral over the variables attached to the half-lines is trivially divergent and the limit $\Omega\rightarrow 0$ of the amplitude is not defined. 

In order to cure this problem, graph amplitudes with heat kernel propagators are usually defined by an integration over all variables associated to the half-lines, save one. 

\begin{defn}
Let $G$ be a connected ribbon graph without flags and let us attach a variable $y_{i}\in{\Bbb R}^{D}$ to each half-edge of $G$, with the convention that  $y_{i_{0}}=0$ for a fixed half-edge $i_{0}$. The (generalized) amplitude of a ribbon graph in the heat kernel theory is defined as
\begin{equation} 
{\cal A}_{G}^{\mbox{\tiny heat kernel}}=\int \prod_{i\neq i_{0}} d^{D}y_{i}
\prod_{e\in E(G)}{\cal K}_{0}(y_{i_{e,+}},y_{i_{e,-}})\prod_{v\in V(G)}{\cal V}_{d_{v}}(y_{i_{v,1}},\dots,y_{i_{v.d_{v}}}),\label{defamplitudeheatkernel}
\end{equation}
with $y_{i_{e,+}},y_{i_{e,-}}$ the variables attached to the ends of $e$ and $y_{i_{v,1}},\dots,y_{i_{v.d_{v}}}$ the variables attached in cyclic order around vertex $v$.
\end{defn}

After the removal of one of these integration variables, the limit $\Omega\rightarrow 0$ is well-defined and related to the first Symanzik polynomial $\mathrm{U}_{G}$ of a non commutative field theory, which is itself an evaluation of the Bollob\`as-Riodan polynomial.  In order to see how this results from the limit $\Omega\rightarrow$ of an amplitude with Mehler kernel, we first define a new graph whose amplitude is  
obtained by integrating over all half-lines but $i_{0}$.

\begin{prop}\label{amplitudeGhat}
et $G$ be a connect ribbon graph without flags and $i_{0}$ one of its half-lines.
We define $\widehat{G}_{i_{0}}$ as the graph constructed by replacing
the half-line $i_{0}$ by a flag on the vertex it is attached to in $G$
and inserting a bivalent vertex with one flag on its other end, see
figure \ref{fig:Ghat}. Then, the amplitude of  $\widehat{G}_{i_{0}}$ with variables $x=0$ for the two extra flags is
\begin{equation} 
{\cal A}_{\widehat{G}_{i_{0}}}(\Omega,0)=\int \prod_{i\neq i_{0}} d^{D}y_{i}
\prod_{e\in E(G)}{\cal K}_{\Omega_{e}}(y_{i_{e,+}},y_{i_{e,-}})\prod_{v\in V(G)}{\cal V}_{d_{v}}(y_{i_{v,1}},\dots,y_{i_{v.d_{v}}}),
\label{i0removed}
\end{equation}
with the convention $y_{i_{0}}=0$.
\end{prop}
\begin{figure}[!htp]
  \centering
  \subfloat[An edge $e$ made of $2$ half-edges $i_0$ and $j_0$ in a graph
  $G$]{{\label{fig:halfEdgei0}}\includegraphics[scale=.8]{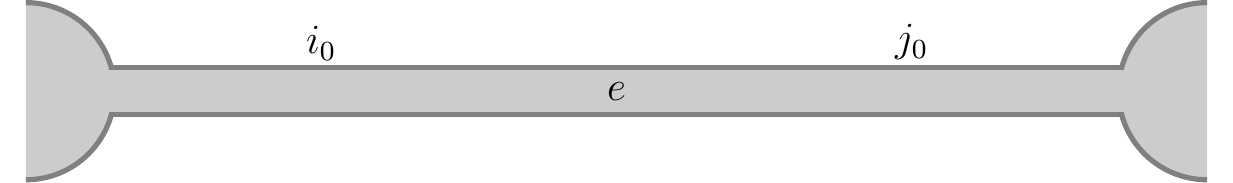}}\\
  \subfloat[The transformed edge in $\widehat{G}_{i_{0}}$]{{\label{fig:transformedEdge}}\includegraphics[scale=.8]{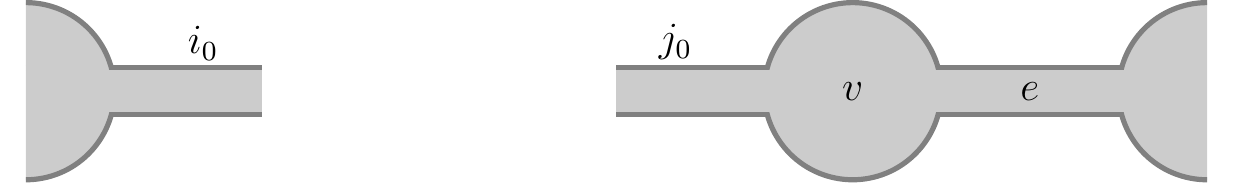}}
  \caption{From $G$ to $\widehat{G}_{i_{0}}$}
  \label{fig:Ghat}
\end{figure}
\begin{proof}
  The amplitudes ${\cal A}_{\widehat{G}_{i_{0}}}(\Omega,0)$ and ${\cal
    A}_{G}$ only differ by the vertex and the edge involving the
  half-line $i_{0}$. Since the two flags of $\widehat{G}_{i_{0}}$
  carry $x=0$, the relevant variable in the interaction and in the
  propagator is set to 0, which reproduces \eqref{i0removed}. Then,
  the heat kernel limit follows immmediately from isolating
  $\prod_{e}\Omega_{e}$ in \eqref{amplitude}.
\end{proof}

\begin{rem}
If $\widehat{G}_{e_{0}}$ is the graph obtained by encircling $i_{0}$ by an extra loop $e_{0}$, then $\widehat{G}_{i_{0}}=\widehat{G}_{e_{0}}^{e_{0}}\cut e_{0}$.
\end{rem}

Then, the heat-kernel limit can be taken as follows.

\begin{thm}
For a connected ribbon graph without flag,
\begin{equation}
{\cal A}_{G}^{\mbox{\tiny heat kernel}}
=\int \prod_{e}d\alpha_{e}\left[ \frac{1}{(4\pi)^{e(G)-v(G)+1}\mathrm{U}_{G}(\alpha,\theta)}\right]^{\frac{D}{2}},
\end{equation}
with
\begin{equation}
\mathrm{U}_{G}(\alpha,\theta)=\sum_{A\subset E(G)\atop (A,V(G))\,\mathrm{quasi-tree}}
\bigg(\frac{\theta}{2}\bigg)^{|A|\!-\!|V(G)|\!+\!1}\bigg\{\prod_{e\notin A}\alpha_{e}\bigg\}\label{defU},
\end{equation}
where a quasi-tree is a ribbon graph whose boundary is connected \footnote{a connected ribbon graph with a single face, in the quantum field theory terminology.}.
\end{thm}
\begin{proof}
  Using theorem \ref{defHUHV}, we can express ${\cal
    A}_{\widehat{G}_{i_{0}}}(\Omega,0)$ as
  \begin{equation} {\cal A}_{\widehat{G}_{i_{0}}}(\Omega,0)=\int
    {\textstyle \prod_{e}\!d\alpha_{e}}\,
    \left[\frac{2^{f(\widehat{G}_{i_{0}})}\prod_{e}\Omega_{e}(1-t_{e}^{2})}
      {(2\pi\theta)^{e(\widehat{G}_{i_{0}})+f(\widehat{G}_{i_{0}})-v(\widehat{G}_{i_{0}})}\mathrm{HU}_{\widehat{G}_{i_{0}}}(\Omega,t)},
    \right]^{D/2}
  \end{equation}
  since the variables attached to the flags vanish. Then, using
  proposition \ref{amplitudeGhat}, we take the Mehler kernel limit
  $\Omega\rightarrow 0$ and get
  \begin{align}
    \lim_{\Omega\rightarrow 0}{\cal
      A}_{\widehat{G}_{i_{0}}}(\Omega,0)={\cal A}_{G}^{\mbox{\tiny
        heat kernel}}\cr =\int \prod_{e}d\alpha_{e}\left[
      \frac{1}{(4\pi)^{e(G)-v(G)+1}\mathrm{U}_{G}(\alpha,\theta)}\right]^{\frac{D}{2}},
  \end{align}
  with
  \begin{equation}
    \mathrm{U}_{G}(\alpha,\theta)=\bigg(\frac{\theta}{2}\bigg)^{e(G)\!-\!v(G)\!+\!1}\lim_{\Omega\rightarrow0}
    \frac{\mathrm{HU}_{\widehat{G}_{i_{0}}}(\Omega,t)}{4\prod_{e}\Omega_{e}}
  \end{equation}
  and $t_{e}=\tanh\frac{2\Omega_{e}\alpha_{e}}{\theta}$. To express
  this limit in terms of quasi-trees, recall that theorem \ref{HU=Q}
  shows that
  \begin{equation}
    \frac{\mathrm{HU}_{\widehat{G}_{i_{0}}}(\Omega,t)}{\prod_{e}\Omega_{e}}=
    \sum_{A,B\subset E(\widehat{G}_{i_{0}})\atop\mathrm{admissible}}\left\{2^{V(\widehat{G}_{i_{0}}^{A})}
      \Big(\prod_{e\in A^{c}\cap B^{c}}\frac{t_{e}}{\Omega_{e}}\Big)
      \Big(\prod_{e\in A^{c}\cap B}t_{e}\Omega_{e}\Big)
      \Big(\prod_{e\in A\cap B}t_{e}^{2}\Big)\right\}.
  \end{equation}
  In the limit $\Omega\rightarrow 0$ with
  $t_{e}=\tanh\frac{2\Omega_{e}\alpha_{e}}{\theta}$, only those terms
  with $B=\emptyset$ do not vanish. Accordingly
  \begin{equation}
    \lim_{\Omega\rightarrow 0}\frac{\mathrm{HU}_{\widehat{G}_{i_{0}}}(\Omega,t)}{\prod_{e}\Omega_{e}}=
    \sum_{A\subset E(G)\atop(A,\emptyset)\mathrm{admissible}}\left\{2^{V(\widehat{G}_{i_{0}}^{A})}
      \Big(\prod_{e\in A^{c}}\frac{\alpha_{e}}{\theta}\Big)\right\}.
  \end{equation}
  Next, notice that $(A,\emptyset)$ is admissible if and only if the
  boundary of $(A,V(\widehat{G}_{i_{0}}))$ has two connected
  components, each carrying one of the flags. To conclude, we need the
  following lemma.
  \begin{lemma}
  The natural bijection between the edges of $G$ and of
  $\widehat{G}_{i_{0}}$ induces a bijection
  \begin{equation}
    \big(A,V(G)\big)\mapsto \big(A,V(\widehat{G}_{i_{0}})\big),
  \end{equation}
  between spanning quasi-trees of $G$ and spanning subgraphs of
  $\widehat{G}_{i_{0}}$ whose boundary has two components, each
  carrying one flag.
\end{lemma}
\begin{proof}
  In $\widehat{G}_{i_{0}}$, let us call $v$ the additional vertex, as
  in figure \ref{fig:transformedEdge}. The set $\cQ_{G}$ of spanning quasi-trees
  in $G$ is the union of two disjoint subsets, respectively
  $\cQ_{G,e}$ and $\cQ_{G}^{e}$, who contain or do not contain
  $e$. Let $Q\in\cQ_{G}^{e}$. By definition, $e\notin E(Q)$. In
  $\widehat{G}_{i_{0}}$, $v$ being connected to the rest of the graph
  only be $e$, the subgraph $F_{E(Q)}\subset \widehat{G}_{i_{0}}$ has
  obviously two boundaries: the boundary of $v$ and its flag $j_{0}$,
  and the boundary of its other component, which is a quasi-tree. On
  the contrary, let $F\subset \widehat{G}_{i_{0}}$ be a subgraph with
  two boundaries, each of which bearing a flag and such that $e\notin
  E(F)$. Then, one boundary of $F$ is the boundary of $v$ and its
  flag. The rest of the graph $F$ has thus only one boundary and is
  therefore a quasi-tree: $F\subset G$ is a quasi-tree.

  Let us now consider the case of subgraphs, which do contain
  $e$ as an edge. First of all, notice that the subgraphs of
  $\widehat{G}_{i_{0}}$ which contain $e$ are in one-to-one
  correspondence with the subgraphs of $G\cut e$ and that this map is
  also a bijection on the subgraphs with two boundaries, each of which
  bears a flag. So we are going to prove that $\cQ_{G,e}$ is in
  one-to-one correspondence with the spanning subgraphs of $G\cut e$
  with two boundaries, one flag per boundary.

  For any ribbon graph with flags $G$ and any $e\in E(G)$, $(G\cut
  e)^{\star}=(G^{\star})^{\set e}\cut e$. Let $Q\in\cQ_{G,e}$. Its
  dual $Q^{\star}$ is a one-vertex ribbon graph. The edge $e$ is a
  loop in $Q^{\star}$ which implies that $(Q^{\star})^{\set e}\cut
  e=(Q\cut e)^{\star}$ has two vertices, each of which bears a
  flag. It is exactly the dual of a subgraph of $G\cut e$ with two
  boundaries and one flag per boundary.

  On the contrary, let $F\subset G\cut e$ be a subgraph with two
  faces, one flag per face. Its dual has two vertices and one flag per
  vertex. To map it to a subgraph of $G^{\star}$, one needs to \emph{uncut}
  $e$ that is glue the two flags together and perform a partial
  duality wrt $e$. This new edge links the two
  vertices of $F$ so that its partial dual has only one vertex. Its
  (natural) dual has therefore one boundary and is then a spanning
  quasi-tree of $G$.
\end{proof}

Therefore, we always have $2^{v(\widehat{G}_{i_{0}}^{A})}=4$ and
\begin{equation}
  \lim_{\Omega\rightarrow0}
  \frac{\mathrm{HU}_{\widehat{G}_{i_{0}}}(\Omega,t)}{4\prod_{e}\Omega_{e}}=
  \sum_{A\subset E(G)\atop (A,V(G))\,\mathrm{quasi-tree}}\bigg\{\prod_{e\notin A}\frac{2\alpha_{e}}{\theta}\bigg\}.
\end{equation}
Finally, \eqref{defU} follows from the factorization of powers of
$\frac{\theta}{2}$.
\end{proof}

\begin{example}[Planar banana and non planar banana]

In the case of the planar and non planar bananas (see examples \ref{planarbananaex} and \ref{nonplanarbananaex}) bananas , let us remove one of the half lines of edge 1. Then,
\begin{align}
\mathrm{HU}_{\widehat{\mbox{\tiny planar}\atop \mbox{\tiny 3-banana}}}(\Omega,t)=\hskip3cm{}\cr
	4\Omega_{1}(1+t_{1}^{2})\big[\Omega_{2}\Omega_{3}(t_{2}^{2}+t_{3}^{2})+t_{2}t_{3}(1+\Omega_{2}^{2}\Omega_{3}^{2})\big]
+4t_{1}\big[t_{2}(1+t_{3})^{2}+t_{3}(1+t_{2}^{2})\big]
\end{align}
and
\begin{align}
\mathrm{HU}_{\widehat{\mbox{\tiny non planar}\atop \mbox{\tiny 3-banana}}}(\Omega,t)=\hskip3cm{}\cr
	4\Omega_{1}(1+t_{1}^{2})\big[\Omega_{2}\Omega_{3}(1+t_{2}^{2}t_{3}^{2})+t_{2}t_{3}(1+\Omega_{2}^{2}\Omega_{3}^{2})\big]
+4t_{1}\big[t_{2}(1+t_{3})^{2}+t_{3}(1+t_{2}^{2})\big],
\end{align}
from which we deduce
\begin{align}
\mathrm{U}_{\mbox{\tiny planar}\atop \mbox{\tiny 3-banana}}(\alpha,\theta)=
\alpha_{1}\alpha_{2}+\alpha_{1}\alpha_{3}+\alpha_{2}\alpha_{3}\label{planarU}
\end{align}
and
\begin{align}
\mathrm{U}_{\mbox{\tiny non planar}\atop \mbox{\tiny 3-banana}}(\alpha,\theta)=
\alpha_{1}\alpha_{2}+\alpha_{1}\alpha_{3}+\alpha_{2}\alpha_{3}+\bigg(\frac{\theta}{2}\bigg)^{2}.\label{nonplanarU}
\end{align}
All the terms in \eqref{planarU} and the first three terms in \eqref{nonplanarU} correspond to the spanning trees. The last term in \eqref{nonplanarU} is the quasi-tree made of all edges.
\end{example}

In fact, $\mathrm{U}_{G}$ is an evaluation of the multivariate  Bollob\'as-Riordan polynomial ${\cal Z}(a,q,c)$
\begin{equation}
\mathrm{U}_{G}(\alpha,\theta)=\bigg(\frac{\theta}{2}\bigg)^{e(G)\!-\!v(G)\!+\!1}\lim_{c\rightarrow0}c^{-1}{\cal Z}_{G}(\frac{2\alpha}{\theta},1,c).
\end{equation}
Equivalently, it can be expressed in terms of the polynomial ${\cal Q}$ as
\begin{equation} 
\mathrm{U}_{G}(\alpha,\theta)=\bigg(\frac{\theta}{2}\bigg)^{e(G)\!-\!v(G)\!+\!1}{\cal Q}_{G}(\frac{2\alpha}{\theta},1,0,0,r)\label{U=Q},
\end{equation}
with $r_{1}=1$ and $r_{n}=0$  for $n\neq 1$
This suggest that $\mathrm{U}_{G}$ has a natural transformation under partial duality.

\begin{cor}\label{dualityU}
For any $A\subset E(G)$, the first Symanzik polynomial transforms under partial duality as
\begin{equation}
\mathrm{U}_{G^{A}}(\alpha,\theta)=\bigg(\frac{\theta}{2}\bigg)^{v(G)-v(G^{A})}
\bigg(\prod_{e\in A}\frac{2\alpha_{e}}{\theta}\bigg)\,\mathrm{U}_{G}(\alpha_{A},\theta),
\end{equation}
with $\left[\alpha_{A}\right]_{e}=\frac{\theta^{2}}{4\alpha_{e}}$ if $e\in A$ and $\left[\alpha_{A}\right]_{e}=\alpha_{e}$ if $e\notin A$.
\end{cor}
\begin{proof}
  First write \eqref{U=Q} as
  \begin{equation}
    \bigg(\frac{2}{\theta}\bigg)^{e(G^{A})\!-\!v(G^{A})\!+\!1}\mathrm{U}_{G^{A}}(\alpha,\theta)={\cal Q}_{G^{A}}(x,y,0,0,r),
  \end{equation}
  with $x_{e}= \frac{2\alpha_{e}}{\theta}$ and $y_{e}=1$. Then,
  partial duality for ${\cal Q}$ reads
  \begin{equation} {\cal Q}_{G^{A}}(x,y,0,0,r)={\cal
      Q}_{G}(x',y',0,0,r)
  \end{equation}
  with $x'_{e}=1$ and $y'_{e}= \frac{2\alpha_{e}}{\theta}$ for $e\in
  A$ and $x'_{e}= \frac{2\alpha_{e}}{\theta}$ and $y'_{e}=1$ for
  $e\notin A$. Next, we expand
  \begin{align} {\cal Q}_{G}(x',y',0,0,r)=\sum_{A'\subset E(G)\atop
      (A',V(G))\,\mbox{\tiny quasi-tree}}\left\{ \Big(\prod_{e\in
        A'^{c}\cap A^{c}}\frac{2\alpha_{e}}{\theta}\Big)
      \Big(\prod_{e\in A'\cap
        A}\frac{2\alpha_{e}}{\theta}\Big)\right\}\cr =\sum_{A'\subset
      E(G)\atop (A',V(G))\,\mbox{\tiny quasi-tree}}\left\{
      \Big(\prod_{e\in A'^{c}\cap
        A^{c}}\frac{2\alpha_{e}}{\theta}\Big) \Big(\prod_{e\in
        A'^{c}\cap
        A}\frac{2\alpha_{e}}{\theta}\frac{\theta}{2\alpha_{e}}\Big)
      \Big(\prod_{e\in A'\cap
        A^{c}}\frac{2\alpha_{e}}{\theta}\frac{\theta}{2\alpha_{e}}\Big)
      \Big(\prod_{e\in A'\cap
        A}\frac{2\alpha_{e}}{\theta}\Big)\right\}\cr =\Big(\prod_{e\in
      A}\frac{2\alpha_{e}}{\theta}\Big)\sum_{A'\subset E(G)\atop
      (A',V(G))\,\mbox{\tiny quasi-tree}}\left\{ \Big(\prod_{e\in
        A'^{c}\cap A^{c}}\frac{2\alpha_{e}}{\theta}\Big)
      \Big(\prod_{e\in A'^{c}\cap A}\frac{\theta}{2\alpha_{e}}\Big)
    \right\} ={\cal Q}_{G}(x'',y'',0,0,r),
  \end{align}
  with $x''_{e}=\frac{\theta}{2\alpha_{e}}$ for $e\in A$ and $x''_{e}=
  \frac{2\alpha_{e}}{\theta}$ for $e\notin A$ and $y''_{e}=1$ for all
  $e$ .  Reverting to the Symanzik polynomials $\mathrm{U}_{G}$ and
  $\mathrm{U}_{G^{A}}$, we get the announced result.
\end{proof}

\begin{example}[Non-planar double tadpole in the heat kernel theory]
The partial dual of a cycle of length 2 with respect to one of its edge  is the non-planar double tadpole (see example \ref{doubletadpole}).  For a cycle of length two, we have a sum over 2 spanning trees
\begin{equation}
\mathrm{U}_{\mbox{\tiny cycle with}\atop
\mbox{\tiny 2 edges}}(\alpha_{1},\alpha_{2},\theta)=\alpha_{1}+\alpha_{2},
\end{equation} 
from which we deduce, using partial duality,
\begin{equation}
\mathrm{U}_{\mbox{\tiny non-planar}\atop
\mbox{\tiny double tadpole}}(\alpha_{1},\alpha_{2},\theta)=\alpha_{1}\alpha_{2}+\bigg(\frac{\theta}{2}\bigg)^{2}.
\end{equation} 
\end{example}

Finally, in the commutative limit $\theta\rightarrow0$ we recover the well known expression of the first Symanzik polynomial as a sum over spanning trees.

\begin{cor}
\begin{equation}
\lim_{\theta\rightarrow 0}\mathrm{U}_{G}(\alpha,\theta)=
\sum_{A\subset E(G)\atop (A,V(G))\,\mathrm{tree}}\bigg\{\prod_{e\notin A}\alpha_{e}\bigg\}.
\end{equation}
\end{cor}
\begin{proof}
  In this limit, only the subsets $A$ such that $|A|-|V|+1=0$
  contribute to \eqref{defU}. This condition characterizes spanning
  trees.
\end{proof}

\subsection{The commutative Mehler kernel limit \texorpdfstring{$\theta\rightarrow 0$}{}}

In this section, we derive a combinatorial formula for the first hyperbolic polynomial in the commutative limit $\theta\rightarrow 0$ in terms of trees and unicyclic graphs. First of all, to recover a commutative quantum field theory with the Mehler kernel corresponding to an harmonic oscillator of frequency $\Omega$ instead of $\widetilde{\Omega}=\frac{2\Omega}{\theta}$ we have to substitute $\Omega\rightarrow\frac{\theta\Omega}{2}$ in \eqref{amplitude}.  

In order to simplify the analysis, we restrict ourselves to graphs without flags\footnote{Otherwise there are extra powers of $\theta$ on the external corners that arise from Dirac distribution on the flags, as we have seen on the examples in section \ref{second}.}. For such a graph, the commutative limit  of the amplitude reads (see proposition \ref{commlimrop})
\begin{equation}
\lim_{\theta\rightarrow 0}
{\cal A}_{G}({\textstyle \frac{\theta\Omega}{2}})=
\lim_{\theta\rightarrow 0}\int \prod_{e}d\alpha_{e}\, \left[\frac{\prod_{e}\Omega_{e}(1-t_{e}^{2})}
{(4\pi)^{e(G)}(2\pi\theta)^{-v(G)}\mathrm{HU}_{G}\big(\frac{\theta\Omega}{2},t\big)}
\right]^{\frac{D}{2}}=
{\cal A}^{\mbox{\tiny commutative}}_{\underline{G}}(\Omega).
\end{equation}
In the limit $\theta\rightarrow 0$, the only terms that survive in $\theta^{-v(G)}\mathrm{HU}_{G}(\frac{\Omega\theta}{2},t)$ are associated with subgraphs of $G$ having at most one cycle per connected component.

\begin{prop}
For a ribbon graph $G$ without flag,
\begin{multline}
\lim_{\theta\rightarrow0}\theta^{-v(G)}\mathrm{HU}_{G}\big(\frac{\theta\Omega}{2},t\big)
=\cr\sum_{A'\subset E(G)\,\mathrm{s.t.}\,
(A',V(G))\atop\mathrm{commutative\, admissible}}\Big\{\prod_{e\in E(G)-A'}t_{e}
\prod_{{K\,\mathrm{connected\,components}}\atop\mathrm{of}\,(A',V(G))}{\cal W}_{K}(\Omega,t)\Big\},\label{commlim}
\end{multline}
where a spanning subgraph is commutative admissible if its connected components are trees (with a least one edge) and unicyclic graphs (i.e. connected graphs with a single cycle). If  $K$ is a tree $T$, its weight is
\begin{equation}
{\cal W}_{T}(\Omega,t)=2^{1-|T|}\sum_{t\in T}\Big\{\;\;\Omega_{e}^{2}t_{e}\!\!\prod_{e'\in T-\{e\}}\Omega_{e'}(1+t_{e'}^{2})\Big\}\label{tree1}
\end{equation}
and if $K$ is a unicylic graph $U$ with cycle edges $C$, its weight is
\begin{equation}
{\cal W}_{U}(\Omega,t)=2^{2-|U|}\label{unicyclic}
\sum_{C'\subset C\atop
|C'|\,\mathrm{odd}}\Big\{
\prod_{e\in C'}\Omega_{e}t_{e}^{2}
\prod_{e'\in U-C}\Omega_{e'}(1+t_{e'}^{2})\Big\}.
\end{equation}
\end{prop}
\begin{proof}
  First recall that
  \begin{equation}
    \mathrm{HU}_{G}(\Omega,t)=\sum_{A,B\subset E(G)\atop\mathrm{admissible}}2^{V(G^{A})}
    \Big(\prod_{e\in A^{c}\cap B^{c}}t_{e}\Big)
    \Big(\prod_{e\in A^{c}\cap B}t_{e}\Omega_{e}^{2}\Big)
    \Big(\prod_{e\in A\cap B^{c}}\Omega_{e}\Big)
    \Big(\prod_{e\in A\cap B}\Omega_{e}t_{e}^{2}\Big),
  \end{equation}
  with $(A,B)$ admissible if each vertex of the graph obtained from
  $G^{A}$ by cutting the edges in $B$ and removing those in $B^{c}$
  has an even number of flags. After the rescaling $\Omega\rightarrow
  \frac{\Omega\theta}{2}$, only those graphs for which
  \begin{equation}
    |A|+2|A^{c}\cap B|\leq v(G)\label{commineq} 
  \end{equation}
  contribute to the commutative limit \eqref{commlim}. Let
  $A'=A\cup(A^{c}\cap B)$ and let $\left\{K_{n}\right\}$ be the
  connected components of $(A',V(G))$. We first show that each $K_{n}$
  is either a unicyclic graph with no edge in $B$ or a tree with one
  edge in $B$ and then compute its weight.

  Let $A'_{n}$ denote the edge set of $K_{n}$, $V_{n}$ its vertex set
  and $B_{n}=A'_{n}\cap A^{c}\cap B$. Thus \eqref{commineq} can be
  written as a sum over connected components
  \begin{equation}
    \sum_{n}|A'_{n}|-|V_{n}|+|B_{n}|\leq 0.\label{commineq2}
  \end{equation}
  With $(A,B)$ admissible, this implies that for each $n$
  \begin{equation}
    |A'_{n}|-|V_{n}|+|B_{n}|= 0.\label{commeq}
  \end{equation}
  Indeed, if this is not the case, then there is $n_{0}$ such that
  $|A'_{n_{0}}|-|V_{n_{0}}|+|B_{n_{0}}|\neq0$. Without loss of
  generality, we may assume that
  $|A'_{n_{0}}|-|V_{n_{0}}|+|B_{n_{0}}|<0$, since if it is strictly
  positive in one connected component, it has to be strictly negative
  in another one to obey \eqref{commineq2}. Then
  $|A'_{n_{0}}|-|V_{n_{0}}|+1+|B_{n_{0}}|\leq0$, but since
  $|A'_{n_{0}}|-|V_{n_{0}}|+1$ (the dimension of the cycle space of
  $K_{n_{0}}$) and $|B_{n_{0}}|$ are positive, this implies that
  $|A'_{n_{0}}|-|V_{n_{0}}|+1=|B_{n_{0}}|=0$. Therefore, $K_{n_{0}}$
  is a tree and $A'_{n}\cap B\subset A$, which means that all the
  edges of $K_{n_{0}}$ belong to $A$ and no edge in $B\cap A^{c}$ is
  incident to a vertex of $K_{n_{0}}$. In the partial dual $G^{A}$,
  $K_{n_{0}}$ gives rise to a single vertex with loops and the cuts of
  the edges in $B$ always yields an even number of flags since there
  is no edge in $B\cap A^{c}$ incident to this vertex. This is in
  contradiction with the fact that $(A,B)$ is admissible, so that
  \eqref{commeq} holds.

  Let us rewrite \eqref{commeq} as
  \begin{equation}
    |A'_{n}|-|V_{n}|+1+|B_{n}|-1= 0.
  \end{equation}
  Because $|A'_{n}|-|V_{n}|+1\geq 0$, $|B_{n}|\geq 2$ is
  impossible. With $|B_{n}|=1$, we have $|A'_{n}|-|V_{n}|+1=0$ so that
  $K_{n}$ is a tree with a single edge in $B$. For $|B_{n}|=0$, we
  obtain $|A'_{n}|-|V_{n}|+1=1$, so that $K_{n}$ is a unicyclic graph
  with no edge in $B$.

  To compute the weights, let us first note that $E-A'=A^{c}\cap
  B^{c}$, so that the contributions of the connected components
  $K_{n}$ factorize and each $e\in E-A'$ yields a factor of
  $t_{e}$. If $K_{n_{0}}$ is a tree, then the partial duality with
  respect to $A$ yields two vertices with loops attached joined by the
  edge in $B$. Each loop contributes a factor of
  $\frac{\Omega_{e}(1+t_{e}^{2})}{2}$, the edge in $B$
  $\frac{t_{e}\Omega^{2}}{4}$ and there is an additional factor of $4$
  since $k_{n_{0}}$ yields two vertices in $G^{A}$. Summing terms that
  only differ by the position of the edge in $B$ on the tree, we
  obtain \eqref{tree1}. If $K_{n}$ is a unicyclic graph, then in the
  partial dual it becomes two vertices with loops, joined by the cycle
  edges. Each loop contributes a factor of
  $\frac{\Omega_{e}(1+t_{e}^{2})}{2}$ and we cut an odd number of
  cycle edges for $(A,B)$ to be admissible. Finally, this yields two
  vertices in $G^{A}$ so that we have an additional factor of 4. This
  proves \eqref{unicyclic}.
\end{proof}

\begin{example}[Dumbbell]
For the dumbbell graph (see \ref{dumbbellex}), the commutative limit is
\begin{align}
\lim_{\theta\rightarrow 0}\theta^{-2}\mathrm{HU}_{\mbox{\tiny dumbbell}}(\frac{\theta\Omega}{2},t)=
4t_{1}\Omega_{2}t_{2}^{2}\Omega_{3}t_{3}^{2}+4t_{1}\Omega_{1}^{2}t_{2}t_{3}
\hskip2cm{}\cr
+4\Omega_{2}t_{2}^{2}\Omega_{1}(1+t_{1}^{2})t_{3}
+4\Omega_{3}t_{3}^{2}\Omega_{1}(1+t_{1}^{2})t_{2},
\end{align}
which corresponds to the covering by two disjoint cycles, one tree and the two unicycles.
\end{example}

\begin{example}[Planar banana and non planar banana]
For the planar and non planar bananas (see examples \ref{planarbananaex} and \ref{nonplanarbananaex}) bananas , we have
\begin{align}
\lim_{\theta\rightarrow 0}\mathrm{HU}_{\mbox{\tiny planar}\atop \mbox{\tiny 3-banana}}(\frac{2\Omega}{\theta},t)=
\lim_{\theta\rightarrow 0}\mathrm{HU}_{\mbox{\tiny non planar}\atop \mbox{\tiny 3-banana}}(\frac{2\Omega}{\theta},t)=\cr
t_{1}t_{2}t_{3}\big[\Omega_{1}^{2}+\Omega_{2}^{2}+\Omega_{3}^{2}\big]\hskip2.5cm{}\cr
+t_{1}\Omega_{2}\Omega_{3}\big[t_{2}^{2}+t_{3}^{2}\big]
+t_{2}\Omega_{1}\Omega_{3}\big[t_{1}^{2}+t_{3}^{2}\big]
+t_{3}\Omega_{1}\Omega_{2}\big[t_{1}^{2}+t_{2}^{2}\big]
\Big].
\end{align}
The first term corresponds to the contribution of the three spanning trees and the last one to the three cycles with two edges. As expected, there is no difference between the two polynomial since the two graphs only differ by a non cyclic permutation of the half-lines at one of the vertices.
\end{example}

\section*{Conclusion and outlooks}
\addcontentsline{toc}{section}{Conclusion and outlooks}
Motivated by the quest of an explicit combinatorial expression of the polynomial appearing in the parametric expression of the Feynman graph amplitudes of the Grosse-Wulkenhaar model, we have introduced a new topological polynomial for ribbon graphs with flags. This polynomial is a natural extension of the multivariate Bollob\'as-Riordan polynomial, with a reduction relation that involves two additional operations and that preserves the invariance under partial duality. This work raises the following questions.

From a purely mathematical point of view, the Bollob\'as-Riordan polynomial is intimately tied with knot theory. This relation relies on  
its invariance under partial duality so that it is natural to inquire whether our newly introduced polynomial could also be related to knot invariants.

Moreover, graph theoretical techniques have proven instrumental in the
evaluation of some of the Feynman amplitude as multiple z\^eta
functions \cite{Brown2009aa,Brown2009ab}. This may also be the case for Grosse-Wulkenhaar model with special properties expected to occur in the critical case $\Omega=1$. A first step towards a study of the Grosse-Wulkenhaar amplitudes from the point of view of algebraic geometry has already been taken in \cite{Aluffi2008aa}.

Finally, attempts at a quantum theory of gravity based on generalized matrix models yield new graph polynomials, as pioneered in \cite{Gurau2009aa}.

{\footnotesize
\bibliographystyle{fababbrvnat}
\bibliography{biblio-articles,biblio-books}
}

\end{document}